\documentclass[nonacm=true, acmsmall]{acmart}
\usepackage[utf8]{inputenc}
\usepackage[T1]{fontenc}

\newcommand{\packageGraphicx}{\usepackage{graphicx}}
\newcommand{\packageHyperref}{\usepackage{hyperref}}
\newcommand{\renewrmdefault}{\renewcommand{\rmdefault}{ptm}}
\newcommand{\packageRelsize}{\usepackage{relsize}}
\newcommand{\packageMathabx}{\usepackage{mathabx}}
\newcommand{\packageWasysym}{
  \let\leftmoon\relax \let\rightmoon\relax \let\fullmoon\relax \let\newmoon\relax \let\diameter\relax
  \usepackage{wasysym}}
\newcommand{\packageTextcomp}{\usepackage{textcomp}}
\newcommand{\packageFramed}{\usepackage{framed}}
\newcommand{\packageHyphenat}{\usepackage[htt]{hyphenat}}
\newcommand{\packageColor}{\usepackage[usenames,dvipsnames]{color}}
\newcommand{\doHypersetup}{\hypersetup{bookmarks=true,bookmarksopen=true,bookmarksnumbered=true}}
\newcommand{\packageTocstyle}{\IfFileExists{tocstyle.sty}{\usepackage{tocstyle}\usetocstyle{standard}}{}}
\newcommand{\packageCJK}{\IfFileExists{CJK.sty}{\usepackage{CJK}}{}}
\renewcommand\packageColor\relax
\renewcommand\packageTocstyle\relax
\renewcommand\packageMathabx{\ifx\bigtimes\undefined \usepackage{mathabx} \else \relax \fi}

\renewcommand{\renewrmdefault}{}

\packageGraphicx
\packageHyperref
\renewrmdefault
\packageRelsize
\packageMathabx
\packageWasysym
\packageTextcomp
\packageFramed
\packageHyphenat
\packageColor
\doHypersetup
\packageTocstyle
\packageCJK

\newcommand{\sectionNewpage}{}

\newcommand{\preDoc}{}
\newcommand{\postDoc}{}

\newcommand{\BookRef}[2]{\emph{#2}}
\newcommand{\ChapRef}[2]{\SecRef{#1}{#2}}
\newcommand{\SecRef}[2]{section~#1}
\newcommand{\PartRef}[2]{part~#1}
\newcommand{\BookRefUC}[2]{\BookRef{#1}{#2}}
\newcommand{\ChapRefUC}[2]{\SecRefUC{#1}{#2}}
\newcommand{\SecRefUC}[2]{Section~#1}
\newcommand{\PartRefUC}[2]{Part~#1}

\newcommand{\BookRefLocal}[3]{\hyperref[#1]{\BookRef{#2}{#3}}}
\newcommand{\ChapRefLocal}[3]{\hyperref[#1]{\ChapRef{#2}{#3}}}
\newcommand{\SecRefLocal}[3]{\hyperref[#1]{\SecRef{#2}{#3}}}
\newcommand{\PartRefLocal}[3]{\hyperref[#1]{\PartRef{#2}{#3}}}
\newcommand{\BookRefLocalUC}[3]{\hyperref[#1]{\BookRefUC{#2}{#3}}}
\newcommand{\ChapRefLocalUC}[3]{\hyperref[#1]{\ChapRefUC{#2}{#3}}}
\newcommand{\SecRefLocalUC}[3]{\hyperref[#1]{\SecRefUC{#2}{#3}}}
\newcommand{\PartRefLocalUC}[3]{\hyperref[#1]{\PartRefUC{#2}{#3}}}

\newcommand{\BookRefUN}[1]{\BookRef{}{#1}}

\newcommand{\SecRefUN}[1]{``#1''}

\newcommand{\BookRefLocalUN}[2]{\hyperref[#1]{\BookRefUN{#2}}}

\newcommand{\SecRefLocalUN}[2]{\hyperref[#1]{\SecRefUN{#2}}}

\newcommand{\SectionNumberLink}[2]{\hyperref[#1]{#2}}

\newcommand{\Scribtexttt}[1]{{\texttt{#1}}}

\newcommand{\planetName}[1]{PLane\hspace{-0.1ex}T}

\makeatletter
\def\empty@finalstrut#1{%
  \unskip\ifhmode\nobreak\fi\vrule\@width\z@\@height\z@\@depth\z@}
\def\no@strut{\global\setbox\@arstrutbox\hbox{%
    \vrule \@height\z@
           \@depth\z@
           \@width\z@}%
    \gdef\@endpbox{\empty@finalstrut\@arstrutbox\par\egroup\hfil}%
}%
\def\yes@strut{\global\setbox\@arstrutbox\hbox{%
    \vrule \@height\arraystretch \ht\strutbox
           \@depth\arraystretch \dp\strutbox
           \@width\z@}%
    \gdef\@endpbox{\@finalstrut\@arstrutbox\par\egroup\hfil}%
}%
\def\@mkpream#1{\@firstamptrue\@lastchclass6
  \let\@preamble\@empty\def\empty@preamble{\add@ins}%
  \let\protect\@unexpandable@protect
  \let\@sharp\relax\let\add@ins\relax
  \let\@startpbox\relax\let\@endpbox\relax
  \@expast{#1}%
  \expandafter\@tfor \expandafter
    \@nextchar \expandafter:\expandafter=\reserved@a\do
       {\@testpach\@nextchar
    \ifcase \@chclass \@classz \or \@classi \or \@classii \or \@classiii
      \or \@classiv \or\@classv \fi\@lastchclass\@chclass}%
  \ifcase \@lastchclass \@acol
      \or \or \@preamerr \@ne\or \@preamerr \tw@\or \or \@acol \fi}
\def\@addamp{%
  \if@firstamp
    \@firstampfalse
    \edef\empty@preamble{\add@ins}%
  \else
    \edef\@preamble{\@preamble &}%
    \edef\empty@preamble{\expandafter\noexpand\empty@preamble &\add@ins}%
  \fi}
\newif\iftw@hlines \tw@hlinesfalse
\def\@xhline{\ifx\reserved@a\hline
               \tw@hlinestrue
             \else\ifx\reserved@a\Hline
               \tw@hlinestrue
             \else
               \tw@hlinesfalse
             \fi\fi
      \iftw@hlines
        \aftergroup\do@after
      \fi
      \ifnum0=`{\fi}%
}
\def\do@after{\emptyrow[\the\doublerulesep]}
\def\emptyrow{\noalign\bgroup\@ifnextchar[\@emptyrow{\@emptyrow[\z@]}}
\def\@emptyrow[#1]{\no@strut\gdef\add@ins{\vrule \@height\z@ \@depth#1 \@width\z@}\egroup%
\empty@preamble\\
\noalign{\yes@strut\gdef\add@ins{\vrule \@height\z@ \@depth\z@ \@width\z@}}%
}
\def\tabrow#1{\noalign\bgroup\@ifnextchar[{\@tabrow{#1}}{\@tabrow{#1}[]}}
\def\@tabrow#1[#2]{\no@strut\egroup#1\ifx.#2.\\\else\\[#2]\fi\noalign{\yes@strut}}
\def\endpltstabular{\crcr\egroup\egroup \egroup}
\expandafter \let \csname endpltstabular*\endcsname = \endpltstabular
\def\pltstabular{\let\@halignto\@empty\@pltstabular}
\@namedef{pltstabular*}#1{\def\@halignto{to#1}\@pltstabular}
\def\@pltstabular{\leavevmode \bgroup \let\@acol\@tabacol
   \let\@classz\@tabclassz
   \let\@classiv\@tabclassiv \let\\\@tabularcr\@stabarray}
\def\@stabarray{\m@th\@ifnextchar[\@sarray{\@sarray[c]}}
\def\@sarray[#1]#2{%
  \bgroup
  \setbox\@arstrutbox\hbox{%
    \vrule \@height\arraystretch\ht\strutbox
           \@depth\arraystretch \dp\strutbox
           \@width\z@}%
  \@mkpream{#2}%
  \edef\@preamble{%
    \ialign \noexpand\@halignto
      \bgroup \@arstrut \@preamble \tabskip\z@skip \cr}%
  \let\@startpbox\@@startpbox \let\@endpbox\@@endpbox
  \let\tabularnewline\\%
    \let\@sharp##%
    \set@typeset@protect
    \lineskip\z@skip\baselineskip\z@skip
    \@preamble}

\makeatother

\newenvironment{bigtabular}{\begin{pltstabular}}{\end{pltstabular}}

\newlength{\stabLeft}
\newcommand{\bigtableleftpad}{\hspace{\stabLeft}}
\newcommand{\atItemizeStart}[0]{\addtolength{\stabLeft}{\labelsep}
                                \addtolength{\stabLeft}{\labelwidth}}

\newenvironment{SingleColumn}{\begin{list}{}{\topsep=0pt\partopsep=0pt%
\listparindent=0pt\itemindent=0pt\labelwidth=0pt\leftmargin=0pt\rightmargin=0pt%
\itemsep=0pt\parsep=0pt}\item}{\end{list}}

\newcommand{\Sendabbrev}[1]{#1\@}

\newcommand{\SCodePreSkip}{\vskip\abovedisplayskip}
\newcommand{\SCodePostSkip}{\vskip\belowdisplayskip}

\newcommand{\SVInsetPreSkip}{\vskip\abovedisplayskip}
\newcommand{\SVInsetPostSkip}{\vskip\belowdisplayskip}

\newenvironment{SCentered}{\begin{trivlist}\item \centering}{\end{trivlist}}

\newcommand{\titleAndVersionAndAuthors}[3]{\title{#1\\{\normalsize \SVersionBefore{}#2}}\author{#3}\maketitle}

\newcommand{\titleAndEmptyVersionAndAuthors}[3]{\title{#1}\author{#3}\maketitle}

\newcommand{\SAuthor}[1]{#1}
\newcommand{\SAuthorSep}[1]{\qquad}
\newcommand{\SVersionBefore}[1]{Version }

\newcommand{\SNumberOfAuthors}[1]{}

\let\SOriginalthesubsection\thesubsection
\let\SOriginalthesubsubsection\thesubsubsection

\newcommand{\Ssection}[2]{\section[#1]{#2}\let\thesubsection\SOriginalthesubsection}
\newcommand{\Ssubsection}[2]{\subsection[#1]{#2}\let\thesubsubsection\SOriginalthesubsubsection}

\newcommand{\Ssectionstar}[1]{\section*{#1}\renewcommand*\thesubsection{\arabic{subsection}}\setcounter{subsection}{0}}

\newcommand{\Ssectionstarx}[2]{\Ssectionstar{#2}\phantomsection\addcontentsline{toc}{section}{#1}}

\newcounter{GrouperTemp}

\newcommand{\Snolinkurl}[1]{\nolinkurl{#1}}

\newcommand{\SAuthorinfo}[4]{#1}
\newcommand{\SAuthorPlace}[1]{#1}
\newcommand{\SAuthorEmail}[1]{#1}

\newcommand{\SConferenceInfo}[2]{}
\newcommand{\SCopyrightYear}[1]{}
\newcommand{\SCopyrightData}[1]{}
\newcommand{\Sdoi}[1]{}

\newcommand{\SCategory}[3]{}
\newcommand{\SCategoryPlus}[4]{}
\newcommand{\STerms}[1]{}
\newcommand{\SKeywords}[1]{}

\newcommand{\AutobibLink}[1]{\color{ACMPurple}{#1}}

\newcommand{\NoteBox}[1]{\footnote{#1}}
\newcommand{\NoteContent}[1]{#1}

\newcommand{\FootnoteRef}[1]{}
\newcommand{\FootnoteTarget}[1]{}

\newcommand{\FootnoteBlockContent}[1]{}
\usepackage{ccaption}

\makeatletter
\@ifundefined{belowcaptionskip}{}{}
\makeatother

\newcommand{\Legend}[1]{~

                        \hrule width \hsize height .33pt
                        \vspace{4pt}
                        \legend{#1}}

\newcommand{\FigureTarget}[2]{#1}
\newcommand{\FigureRef}[2]{#1}

\newlength{\FigOrigskip}
\FigOrigskip=\parskip

\newcommand{\FigureSetRef}{\refstepcounter{figure}}

\newenvironment{Figure}{\begin{figure}\FigureSetRef}{\end{figure}}
\newenvironment{FigureMulti}{\begin{figure*}[t!p]\FigureSetRef}{\end{figure*}}

\newenvironment{Centerfigure}{\begin{Xfigure}\centering\item}{\end{Xfigure}}

\newenvironment{Xfigure}{\begin{list}{}{\leftmargin=0pt\topsep=0pt\parsep=\FigOrigskip\partopsep=0pt}}{\end{list}}

\newenvironment{FigureInside}{}{}

\newcommand{\Centertext}[1]{\begin{center}#1\end{center}}

\newcommand{\SColorize}[2]{\color{#1}{#2}}

\newcommand{\SHyphen}[1]{#1}

\newcommand{\inColor}[2]{{\SHyphen{\Scribtexttt{\SColorize{#1}{#2}}}}}
\definecolor{PaleBlue}{rgb}{0.90,0.90,1.0}
\definecolor{LightGray}{rgb}{0.90,0.90,0.90}
\definecolor{CommentColor}{rgb}{0.76,0.45,0.12}
\definecolor{ParenColor}{rgb}{0.52,0.24,0.14}
\definecolor{IdentifierColor}{rgb}{0.15,0.15,0.50}
\definecolor{ResultColor}{rgb}{0.0,0.0,0.69}
\definecolor{ValueColor}{rgb}{0.13,0.55,0.13}
\definecolor{OutputColor}{rgb}{0.59,0.00,0.59}

\newcommand{\RktSym}[1]{\inColor{IdentifierColor}{#1}}

\newcommand{\RBackgroundLabel}[1]{}

\newenvironment{AutoBibliography}{\begin{small}}{\end{small}}
\newcommand{\Autobibentry}[1]{\hspace{0.05\linewidth}\parbox[t]{0.95\linewidth}{\parindent=-0.05\linewidth#1\vspace{1.0ex}}}

\usepackage{calc}
\newlength{\ABcollength}

\newcommand{\Autobibref}[1]{#1}

\providecommand{\AutobibLink}[1]{#1}

\renewcommand{\titleAndVersionAndAuthors}[3]{\title{#1}#3\maketitle}
\renewcommand{\titleAndEmptyVersionAndAuthors}[3]{\titleAndVersionAndAuthors{#1}{#2}{#3}}

\def\SAuthor#1{\SAutoAuthor#1\SAutoAuthorDone{#1}}
\def\SAutoAuthorDone#1{}
\def\SAutoAuthor{\futurelet\next\SAutoAuthorX}
\def\SAutoAuthorX{\ifx\next\SAuthorinfo \let\Snext\relax \else \let\Snext\SToAuthorDone \fi \Snext}
\def\SToAuthorDone{\futurelet\next\SToAuthorDoneX}
\def\SToAuthorDoneX#1{\ifx\next\SAutoAuthorDone \let\Snext\SAddAuthorInfo \else \let\Snext\SToAuthorDone \fi \Snext}
\newcommand{\SAddAuthorInfo}[1]{\SAuthorinfo{#1}{}{}}

\renewcommand{\SAuthorinfo}[4]{\author{#1}{#2}{#3}{#4}}
\renewcommand{\SAuthorSep}[1]{}

\renewcommand{\SAuthorPlace}[1]{\affiliation{#1}}
\renewcommand{\SAuthorEmail}[1]{\email{#1}}

\renewcommand{\SConferenceInfo}[2]{\conferenceinfo{#1}{#2}}
\renewcommand{\SCopyrightYear}[1]{\copyrightyear{#1}}
\renewcommand{\SCopyrightData}[1]{\copyrightdata{#1}}

\renewcommand{\SCategory}[3]{\category{#1}{#2}{#3}}
\renewcommand{\SCategoryPlus}[4]{\category{#1}{#2}{#3}[#4]}
\renewcommand{\STerms}[1]{\terms{#1}}
\renewcommand{\SKeywords}[1]{\keywords{#1}}

\usepackage{thmtools}
\usepackage{ulem}
\usepackage{tikz}
\usetikzlibrary{shapes.multipart}
\usepackage{textgreek}
\usepackage{enumerate}

\begin{document}
\preDoc

\begin{abstract}The effectiveness of concolic testing deteriorates as the size of
 programs increases. A promising way out is to concolic{-}test programs in
 a component{-}by{-}component fashion, e.g., one function or class at a time.
 Alas, this idea hits an important roadblock in modern languages such as
 JavaScript, Python, and Racket. In these languages, components expect
 functions, objects, and even classes as inputs.  The crux of the problem
 is that existing  concolic testing techniques cannot faithfully capture
 the complex interactions between a higher{-}order program and its inputs
 in order to distill it in a first{-}order formula that an SMT solver can work
 with.

In this paper, we take the first step towards solving the problem; we
offer a design, semantics, and prototype for concolic testing of higher{-}order functions.
Inspired by work on higher{-}order symbolic execution, our model
constructs inputs for higher{-}order functions with a
canonical shape. This enables the concolic tester to keep track of which
pieces of the control{-}flow path of the higher{-}order function depend on
the shape of its input and which do not. The concolic tester encodes the
pieces that do not depend on the shape of the input as a first{-}order
formula. Subsequently, similar to a first{-}order concolic tester, it
leverages an SMT solver to produce another input with the same shape  but
that explores a different control{-}flow path of the higher{-}order function.
As a separate dimension, the concolic tester iteratively explores the
canonical shapes for the input and, investigating all the ways a
higher{-}order function can interact with its input, searching for bugs.

To validate our design, we prove that if a higher{-}order function has a bug,
our concolic tester will eventually construct an input that triggers the
bug. Using our design as a blueprint, we
implement a prototype concolic tester and confirm that it
discovers bugs in a variety of higher{-}order programs from the literature.\end{abstract}\titleAndEmptyVersionAndAuthors{Dynamic Symbolic Execution of Higher{-}Order Functions}{}{\SNumberOfAuthors{3}\SAuthor{\SAuthorinfo{Shu{-}Hung You}{}{\SAuthorPlace{\institution{PLT@Northwestern}}}{\SAuthorEmail{shu-hung.you@eecs.northwestern.edu}}}\SAuthorSep{}\SAuthor{\SAuthorinfo{Robert Bruce Findler}{}{\SAuthorPlace{\institution{PLT@Northwestern}}}{\SAuthorEmail{robby@cs.northwestern.edu}}}\SAuthorSep{}\SAuthor{\SAuthorinfo{Christos Dimoulas}{}{\SAuthorPlace{\institution{PLT@Northwestern}}}{\SAuthorEmail{chrdimo@eecs.northwestern.edu}}}}
\label{t:x28part_x22Dynamicx5fSymbolicx5fExecutionx5fofx5fHigherx2dOrderx5fFunctionsx22x29}

\noindent 

\noindent

\sectionNewpage

\Ssection{Introduction}{Introduction}\label{t:x28part_x22Introductionx22x29}

Concolic testing\Autobibref{~(\hyperref[t:x28autobib_x22Cristian_Cadar_and_Dawson_EnglerExecution_Generated_Test_Casesx3a_How_to_Make_Systems_Code_Crash_ItselfIn_Procx2e_International_SPINConference_on_Model_Cheching_Softwarex2c_ppx2e_2x2dx2d232005x22x29]{\AutobibLink{Cadar and Engler}} \hyperref[t:x28autobib_x22Cristian_Cadar_and_Dawson_EnglerExecution_Generated_Test_Casesx3a_How_to_Make_Systems_Code_Crash_ItselfIn_Procx2e_International_SPINConference_on_Model_Cheching_Softwarex2c_ppx2e_2x2dx2d232005x22x29]{\AutobibLink{2005}}; \hyperref[t:x28autobib_x22Patrice_Godefroidx2c_Nils_Klarlundx2c_and_Koushik_SenDARTx3a_Directed_Automated_Random_TestingIn_Procx2e_ACM_Conference_on_Programming_Language_Design_and_Implementationx2c_ppx2e_213x2dx2d2232005x22x29]{\AutobibLink{Godefroid et al\Sendabbrev{.}}} \hyperref[t:x28autobib_x22Patrice_Godefroidx2c_Nils_Klarlundx2c_and_Koushik_SenDARTx3a_Directed_Automated_Random_TestingIn_Procx2e_ACM_Conference_on_Programming_Language_Design_and_Implementationx2c_ppx2e_213x2dx2d2232005x22x29]{\AutobibLink{2005}})} explores a
program{'}s behavior in a gradual fashion to discover bugs. First, the
concolic tester supplies a random input to the program under test
(hereafter the \emph{user program}\NoteBox{\NoteContent{We underline the first occurrence of each new technical term.}}) and
monitors how the random input forces the user program evaluation down a
specific control{-}flow path.  The concolic tester records this path as a
first{-}order formula and uses an SMT solver to induce a new input that is
designed to force the user program to take a different control{-}flow path. The
process repeats until the concolic tester discovers a bug or times out. In
effect, concolic testing enhances random testing with symbolic execution
to guide input generation to hard{-}to{-}reach corners of the control flow of
a program.

Testament to the success of the technique is the number and diversity of
its adaptations to (i) different languages and platforms:
CUTE\Autobibref{~(\hyperref[t:x28autobib_x22Koushik_Senx2c_Darko_Marinovx2c_and_Gul_AghaCUTEx3a_A_Concolic_Unit_Testing_Engine_for_CIn_Procx2e_International_Symposium_on_on_the_Foundations_of_Software_Engineeringx2c_ppx2e_263x2dx2d2722005x22x29]{\AutobibLink{Sen et al\Sendabbrev{.}}} \hyperref[t:x28autobib_x22Koushik_Senx2c_Darko_Marinovx2c_and_Gul_AghaCUTEx3a_A_Concolic_Unit_Testing_Engine_for_CIn_Procx2e_International_Symposium_on_on_the_Foundations_of_Software_Engineeringx2c_ppx2e_263x2dx2d2722005x22x29]{\AutobibLink{2005}})} and CREST\Autobibref{~(\hyperref[t:x28autobib_x22Jacob_Burnim_and_Koushik_SenHeuristics_for_Scalable_Dynamic_Test_GenerationIn_Procx2e_ACMx2fIEEE_International_Conference_on_Automated_Software_Engineeringx2c_ppx2e_443x2dx2d4462008x22x29]{\AutobibLink{Burnim and Sen}} \hyperref[t:x28autobib_x22Jacob_Burnim_and_Koushik_SenHeuristics_for_Scalable_Dynamic_Test_GenerationIn_Procx2e_ACMx2fIEEE_International_Conference_on_Automated_Software_Engineeringx2c_ppx2e_443x2dx2d4462008x22x29]{\AutobibLink{2008}})} for C,
KLOVER\Autobibref{~(\hyperref[t:x28autobib_x22Guodong_Lix2c_Indradeep_Ghoshx2c_and_Sreeranga_Px2e_RajanKLOVERx3a_A_Symbolic_Execution_and_Automatic_Test_Generation_Tool_for_Cx2bx2b_ProgramsIn_Procx2e_International_Conference_on_Computer_Aided_Verificationx2c_ppx2e_609x2dx2d6152011x22x29]{\AutobibLink{Li et al\Sendabbrev{.}}} \hyperref[t:x28autobib_x22Guodong_Lix2c_Indradeep_Ghoshx2c_and_Sreeranga_Px2e_RajanKLOVERx3a_A_Symbolic_Execution_and_Automatic_Test_Generation_Tool_for_Cx2bx2b_ProgramsIn_Procx2e_International_Conference_on_Computer_Aided_Verificationx2c_ppx2e_609x2dx2d6152011x22x29]{\AutobibLink{2011}})} for C++, jCUTE\Autobibref{~(\hyperref[t:x28autobib_x22Koushik_Sen_and_Gul_AghaCUTE_and_jCUTEx3a_Concolic_Unit_Testing_and_Explicit_Path_Modelx2dchecking_ToolsIn_Procx2e_International_Conference_on_Computer_Aided_Verificationx2c_ppx2e_419x2dx2d4232006x22x29]{\AutobibLink{Sen and Agha}} \hyperref[t:x28autobib_x22Koushik_Sen_and_Gul_AghaCUTE_and_jCUTEx3a_Concolic_Unit_Testing_and_Explicit_Path_Modelx2dchecking_ToolsIn_Procx2e_International_Conference_on_Computer_Aided_Verificationx2c_ppx2e_419x2dx2d4232006x22x29]{\AutobibLink{2006}})} and
JDart\Autobibref{~(\hyperref[t:x28autobib_x22Marko_Dimjax161evix107x2c_Dimitra_Giannakopouloux2c_Falk_Howarx2c_Falk_Howarx2c_Falk_Howarx2c_and_Falk_HowarThe_Dartx2c_the_Psycox2c_and_the_Doopx3a_Concolic_Execution_in_JavaACM_SIGSOFT_Software_Engineering_Notes_40x281x29x2c_ppx2e_1x2dx2d52015x22x29]{\AutobibLink{Dimja\v{s}evi\'{c} et al\Sendabbrev{.}}} \hyperref[t:x28autobib_x22Marko_Dimjax161evix107x2c_Dimitra_Giannakopouloux2c_Falk_Howarx2c_Falk_Howarx2c_Falk_Howarx2c_and_Falk_HowarThe_Dartx2c_the_Psycox2c_and_the_Doopx3a_Concolic_Execution_in_JavaACM_SIGSOFT_Software_Engineering_Notes_40x281x29x2c_ppx2e_1x2dx2d52015x22x29]{\AutobibLink{2015}})} for Java, Jalangi\Autobibref{~(\hyperref[t:x28autobib_x22Koushik_Senx2c_Swaroop_Kalasapurx2c_Brutch_Tasneemx2c_and_Simon_GibbsJalangix3a_A_Selective_Recordx2dreplay_and_Dynamic_Analysis_Framework_for_JavaScriptIn_Procx2e_International_Symposium_on_on_the_Foundations_of_Software_Engineeringx2c_ppx2e_488x2dx2d4982013x22x29]{\AutobibLink{Sen et al\Sendabbrev{.}}} \hyperref[t:x28autobib_x22Koushik_Senx2c_Swaroop_Kalasapurx2c_Brutch_Tasneemx2c_and_Simon_GibbsJalangix3a_A_Selective_Recordx2dreplay_and_Dynamic_Analysis_Framework_for_JavaScriptIn_Procx2e_International_Symposium_on_on_the_Foundations_of_Software_Engineeringx2c_ppx2e_488x2dx2d4982013x22x29]{\AutobibLink{2013}})} and
SymJS\Autobibref{~(\hyperref[t:x28autobib_x22Li_Guodongx2c_Esben_Andreasenx2c_and_Indradeep_GhoshSymJSx3a_Automatic_Symbolic_Testing_of_JavaScript_Web_ApplicationsIn_Procx2e_International_Symposium_on_on_the_Foundations_of_Software_Engineeringx2c_ppx2e_449x2dx2d4592014x22x29]{\AutobibLink{Guodong et al\Sendabbrev{.}}} \hyperref[t:x28autobib_x22Li_Guodongx2c_Esben_Andreasenx2c_and_Indradeep_GhoshSymJSx3a_Automatic_Symbolic_Testing_of_JavaScript_Web_ApplicationsIn_Procx2e_International_Symposium_on_on_the_Foundations_of_Software_Engineeringx2c_ppx2e_449x2dx2d4592014x22x29]{\AutobibLink{2014}})} for JavaScript and CutEr\Autobibref{~(\hyperref[t:x28autobib_x22Aggelos_Giantsiosx2c_Nikolaos_Papaspyroux2c_and_Konstantinos_SagonasConcolic_Testing_for_Functional_LanguagesIn_Procx2e_ACM_International_Conference_on_Principles_and_Practice_of_Declarative_Programmingx2c_ppx2e_137x2dx2d1482015x22x29]{\AutobibLink{Giantsios et al\Sendabbrev{.}}} \hyperref[t:x28autobib_x22Aggelos_Giantsiosx2c_Nikolaos_Papaspyroux2c_and_Konstantinos_SagonasConcolic_Testing_for_Functional_LanguagesIn_Procx2e_ACM_International_Conference_on_Principles_and_Practice_of_Declarative_Programmingx2c_ppx2e_137x2dx2d1482015x22x29]{\AutobibLink{2015}})} for
Erlang, Pex\Autobibref{~(\hyperref[t:x28autobib_x22Nikolai_Tillmann_and_Jonathan_de_HalleuxPexx3a_White_Box_Test_Generation_for_x2eNETIn_Procx2e_International_Conference_on_Tests_and_Proofsx2c_ppx2e_134x2dx2d1532008x22x29]{\AutobibLink{Tillmann and Halleux}} \hyperref[t:x28autobib_x22Nikolai_Tillmann_and_Jonathan_de_HalleuxPexx3a_White_Box_Test_Generation_for_x2eNETIn_Procx2e_International_Conference_on_Tests_and_Proofsx2c_ppx2e_134x2dx2d1532008x22x29]{\AutobibLink{2008}})} for .NET, KLEE for LLVM\Autobibref{~(\hyperref[t:x28autobib_x22Cristian_Cadarx2c_Daniel_Dunbarx2c_and_Dawson_EnglerKLEEx3a_Unassisted_and_Automatic_Generation_of_Highx2dcoverage_Tests_for_Complex_Systems_ProgramsIn_Procx2e_USENIX_Symposium_on_Operating_Systems_Design_and_Implementationx2c_ppx2e_209x2dx2d2242008x22x29]{\AutobibLink{Cadar et al\Sendabbrev{.}}} \hyperref[t:x28autobib_x22Cristian_Cadarx2c_Daniel_Dunbarx2c_and_Dawson_EnglerKLEEx3a_Unassisted_and_Automatic_Generation_of_Highx2dcoverage_Tests_for_Complex_Systems_ProgramsIn_Procx2e_USENIX_Symposium_on_Operating_Systems_Design_and_Implementationx2c_ppx2e_209x2dx2d2242008x22x29]{\AutobibLink{2008}})};
and (ii) different application domains: security\Autobibref{~(\hyperref[t:x28autobib_x22Patrice_Godefroidx2c_Michael_Yx2e_Levinx2c_and_David_MolnarAutomated_Whitebox_Fuzz_TestingIn_Procx2e_Network_and_Distributed_System_Security_Symposium2008x22x29]{\AutobibLink{Godefroid et al\Sendabbrev{.}}} \hyperref[t:x28autobib_x22Patrice_Godefroidx2c_Michael_Yx2e_Levinx2c_and_David_MolnarAutomated_Whitebox_Fuzz_TestingIn_Procx2e_Network_and_Distributed_System_Security_Symposium2008x22x29]{\AutobibLink{2008}}, \hyperref[t:x28autobib_x22Patrice_Godefroidx2c_Michael_Yx2e_Levinx2c_and_David_MolnarSAGEx3a_Whitebox_Fuzzing_for_Security_TestingACM_Queue_10x281x29x2c_ppx2e_20x3a20x2dx2d20x3a272012x22x29]{\AutobibLink{2012}})}, mobile apps\Autobibref{~(\hyperref[t:x28autobib_x22Saswat_Anandx2c_Mayur_Naikx2c_Mary_Jean_Harroldx2c_and_Hongseok_YangAutomated_Concolic_Testing_of_Smartphone_AppsIn_Procx2e_International_Symposium_on_on_the_Foundations_of_Software_Engineeringx2c_ppx2e_59x3a1x2dx2d59x3a112012x22x29]{\AutobibLink{Anand et al\Sendabbrev{.}}} \hyperref[t:x28autobib_x22Saswat_Anandx2c_Mayur_Naikx2c_Mary_Jean_Harroldx2c_and_Hongseok_YangAutomated_Concolic_Testing_of_Smartphone_AppsIn_Procx2e_International_Symposium_on_on_the_Foundations_of_Software_Engineeringx2c_ppx2e_59x3a1x2dx2d59x3a112012x22x29]{\AutobibLink{2012}})}, database
applications\Autobibref{~(\hyperref[t:x28autobib_x22Michael_Emmix2c_Rupak_Majumdarx2c_and_Koushik_SenDynamic_Test_Input_Generation_for_Database_ApplicationsIn_Procx2e_International_Symposium_on_Software_Testing_and_Analysisx2c_ppx2e_151x2dx2d1622007x22x29]{\AutobibLink{Emmi et al\Sendabbrev{.}}} \hyperref[t:x28autobib_x22Michael_Emmix2c_Rupak_Majumdarx2c_and_Koushik_SenDynamic_Test_Input_Generation_for_Database_ApplicationsIn_Procx2e_International_Symposium_on_Software_Testing_and_Analysisx2c_ppx2e_151x2dx2d1622007x22x29]{\AutobibLink{2007}})}, concurrent
programs\Autobibref{~(\hyperref[t:x28autobib_x22Azadeh_Farzanx2c_Andreas_Holzerx2c_Niloofar_Razavix2c_and_Helmut_VeithCon2Colic_TestingIn_Procx2e_International_Symposium_on_on_the_Foundations_of_Software_Engineeringx2c_ppx2e_37x2dx2d472013x22x29]{\AutobibLink{Farzan et al\Sendabbrev{.}}} \hyperref[t:x28autobib_x22Azadeh_Farzanx2c_Andreas_Holzerx2c_Niloofar_Razavix2c_and_Helmut_VeithCon2Colic_TestingIn_Procx2e_International_Symposium_on_on_the_Foundations_of_Software_Engineeringx2c_ppx2e_37x2dx2d472013x22x29]{\AutobibLink{2013}}; \hyperref[t:x28autobib_x22Niloofar_Razavix2c_Franjo_Ivanx10dix107x2c_Vineet_Kahlonx2c_and_Aarti_GuptaConcurrent_Test_Generation_Using_Concolic_Multix2dtrace_AnalysisIn_Procx2e_Asian_Symposium_on_Programming_Languages_and_Systemsx2c_ppx2e_239x2dx2d2552012x22x29]{\AutobibLink{Razavi et al\Sendabbrev{.}}} \hyperref[t:x28autobib_x22Niloofar_Razavix2c_Franjo_Ivanx10dix107x2c_Vineet_Kahlonx2c_and_Aarti_GuptaConcurrent_Test_Generation_Using_Concolic_Multix2dtrace_AnalysisIn_Procx2e_Asian_Symposium_on_Programming_Languages_and_Systemsx2c_ppx2e_239x2dx2d2552012x22x29]{\AutobibLink{2012}})}, embedded
systems\Autobibref{~(\hyperref[t:x28autobib_x22Yunho_Kim_and_Moonzoo_KimSCOREx3a_A_Scalable_Concolic_Testing_Tool_for_Reliable_Embedded_SoftwareIn_Procx2e_International_Symposium_on_on_the_Foundations_of_Software_Engineeringx2c_ppx2e_420x2dx2d4232011x22x29]{\AutobibLink{Kim and Kim}} \hyperref[t:x28autobib_x22Yunho_Kim_and_Moonzoo_KimSCOREx3a_A_Scalable_Concolic_Testing_Tool_for_Reliable_Embedded_SoftwareIn_Procx2e_International_Symposium_on_on_the_Foundations_of_Software_Engineeringx2c_ppx2e_420x2dx2d4232011x22x29]{\AutobibLink{2011}})}, GPU programming\Autobibref{~(\hyperref[t:x28autobib_x22Guodong_Lix2c_Peng_Lix2c_Geof_Sawayax2c_Ganesh_Gopalakrishnanx2c_Indradeep_Ghoshx2c_and_Sreeranga_Px2e_RajanGKLEEx3a_Concolic_Verification_and_Test_Generation_for_GPUsIn_Procx2e_Symposium_on_Principles_and_Practice_of_Parallel_Programmingx2c_ppx2e_215x2dx2d2242012x22x29]{\AutobibLink{Li et al\Sendabbrev{.}}} \hyperref[t:x28autobib_x22Guodong_Lix2c_Peng_Lix2c_Geof_Sawayax2c_Ganesh_Gopalakrishnanx2c_Indradeep_Ghoshx2c_and_Sreeranga_Px2e_RajanGKLEEx3a_Concolic_Verification_and_Test_Generation_for_GPUsIn_Procx2e_Symposium_on_Principles_and_Practice_of_Parallel_Programmingx2c_ppx2e_215x2dx2d2242012x22x29]{\AutobibLink{2012}})} and deep
learning\Autobibref{~(\hyperref[t:x28autobib_x22Youcheng_Sunx2c_Min_Wux2c_Wenjie_Ruanx2c_Xiaowei_Huangx2c_Marta_Kwiatkowskax2c_and_Daniel_KroeningConcolic_Testing_for_Deep_Neural_NetworksIn_Procx2e_ACMx2fIEEE_International_Conference_on_Automated_Software_Engineeringx2c_ppx2e_109x2dx2d1192018x22x29]{\AutobibLink{Sun et al\Sendabbrev{.}}} \hyperref[t:x28autobib_x22Youcheng_Sunx2c_Min_Wux2c_Wenjie_Ruanx2c_Xiaowei_Huangx2c_Marta_Kwiatkowskax2c_and_Daniel_KroeningConcolic_Testing_for_Deep_Neural_NetworksIn_Procx2e_ACMx2fIEEE_International_Conference_on_Automated_Software_Engineeringx2c_ppx2e_109x2dx2d1192018x22x29]{\AutobibLink{2018}})}.  However, not all is coming up roses. When the size of
a program increases, so does the number of control{-}flow paths that the concolic
tester needs to explore.  As a result concolic testing of large programs
becomes ineffective and misses bugs.

An alternative to applying concolic testing to a whole program is to test
program components. Unfortunately, existing concolic testers are not
prepared to effectively test components of programs written in
modern higher{-}order languages like JavaScript, Python and Racket.  When concolic
testers deal with programs that consume functions or objects, they fall back
to heuristics\Autobibref{~(\hyperref[t:x28autobib_x22Pieter_Koopman_and_Rinus_PlasmeijerAutomatic_Testing_of_Higher_Order_FunctionsIn_Procx2e_Asian_Symposium_on_Programming_Languages_and_Systemsx2c_ppx2e_148x2dx2d1642006x22x29]{\AutobibLink{Koopman and Plasmeijer}} \hyperref[t:x28autobib_x22Pieter_Koopman_and_Rinus_PlasmeijerAutomatic_Testing_of_Higher_Order_FunctionsIn_Procx2e_Asian_Symposium_on_Programming_Languages_and_Systemsx2c_ppx2e_148x2dx2d1642006x22x29]{\AutobibLink{2006}}; \hyperref[t:x28autobib_x22Lian_Lix2c_Yi_Lux2c_and_Jingling_XueDynamic_Symbolic_Execution_for_PolymorphismIn_Procx2e_International_Conference_on_Compiler_Constructionx2c_ppx2e_120x2dx2d1302017x22x29]{\AutobibLink{Li et al\Sendabbrev{.}}} \hyperref[t:x28autobib_x22Lian_Lix2c_Yi_Lux2c_and_Jingling_XueDynamic_Symbolic_Execution_for_PolymorphismIn_Procx2e_International_Conference_on_Compiler_Constructionx2c_ppx2e_120x2dx2d1302017x22x29]{\AutobibLink{2017}}; \hyperref[t:x28autobib_x22Marija_Selakovicx2c_Michael_Pradelx2c_Rezwana_Karimx2c_and_Frank_TipTest_Generation_for_Higherx2dorder_Functions_in_Dynamic_LanguagesProceedings_of_the_ACM_on_Programming_Languages_x28OOPSLAx29_2x2c_ppx2e_161x3a1x2dx2d161x3a272018x22x29]{\AutobibLink{Selakovic et al\Sendabbrev{.}}} \hyperref[t:x28autobib_x22Marija_Selakovicx2c_Michael_Pradelx2c_Rezwana_Karimx2c_and_Frank_TipTest_Generation_for_Higherx2dorder_Functions_in_Dynamic_LanguagesProceedings_of_the_ACM_on_Programming_Languages_x28OOPSLAx29_2x2c_ppx2e_161x3a1x2dx2d161x3a272018x22x29]{\AutobibLink{2018}})} or are unable to call back into higher{-}order
values\Autobibref{~(\hyperref[t:x28autobib_x22Aggelos_Giantsiosx2c_Nikolaos_Papaspyroux2c_and_Konstantinos_SagonasConcolic_Testing_for_Functional_LanguagesIn_Procx2e_ACM_International_Conference_on_Principles_and_Practice_of_Declarative_Programmingx2c_ppx2e_137x2dx2d1482015x22x29]{\AutobibLink{Giantsios et al\Sendabbrev{.}}} \hyperref[t:x28autobib_x22Aggelos_Giantsiosx2c_Nikolaos_Papaspyroux2c_and_Konstantinos_SagonasConcolic_Testing_for_Functional_LanguagesIn_Procx2e_ACM_International_Conference_on_Principles_and_Practice_of_Declarative_Programmingx2c_ppx2e_137x2dx2d1482015x22x29]{\AutobibLink{2015}})}. In general, there are inputs that existing concolic testers cannot generate
and these inputs are necessary to explore all of the behavior of a
higher{-}order program.

The reason for this limitation is fundamental;
concolic testers rely on SMT solvers that can only deal with first{-}order formulas.
Specifically, in a higher{-}order setting, inputs affect not only the control{-}flow
of the user program but the
user program itself may affect the control{-}flow of its inputs. This
inter{-}dependency implies that a control{-}flow path of a higher{-}order program
cannot be described as a symbolic formula of constraints based solely on the
first{-}order properties of the program{'}s inputs. The control flow of the program also
depends on the behavior of the inputs.

Starting from this observation, our \textbf{contribution}
is the design of a concolic tester for
higher{-}order programs in a functional, dynamically{-}typed setting.
Our key insight is that it is possible to split the search for
a bug in two levels: at one level we describe entire classes of behavior
of the input, and at the second level we exploit the solver to search within
a specific pattern of behavior.

In
essence, the formulas capture the first{-}order properties of the values that flow between
the user program and its input, for a fixed control structure of the input.
Thus, just like for first{-}order code, an SMT solver can produce
 inputs that explore
different control{-}flow paths in the user program.  Intuitively, the
concolic tester and the user program are in a
conversation and we use the SMT solver to search the space of the
first{-}order properties of the values they exchange after deciding on the
{``}strategy{''} the concolic tester follows to produce its answers.

To explore all of the behavior of the user program, however, the concolic tester must also explore different patterns of behavior
of the input. That is, it must explore all the different strategies the concolic tester
can employ in the conversation with the user program.
Specifically, this search level varies the higher{-}order control structure of values
that the concolic tester sends to the user program.
Building on work on higher{-}order symbolic execution \Autobibref{~(\hyperref[t:x28autobib_x22Phxfac_Nguyx1ec5nx2c_Sam_Tobinx2dHochstadtx2c_and_David_Van_HornRelatively_complete_counterexamples_for_higherx2dorder_programsx2eIn_Procx2e_ACM_Conference_on_Programming_Language_Design_and_Implementationx2c_ppx2e_446x2dx2d4562015x22x29]{\AutobibLink{Nguy\~{\^{e}}n et al\Sendabbrev{.}}} \hyperref[t:x28autobib_x22Phxfac_Nguyx1ec5nx2c_Sam_Tobinx2dHochstadtx2c_and_David_Van_HornSoft_Contract_VerificationIn_Procx2e_ACM_International_Conference_on_Functional_Programmingx2c_ppx2e_139x2dx2d1522014x22x29]{\AutobibLink{2014}}, \hyperref[t:x28autobib_x22Phxfac_Nguyx1ec5nx2c_Sam_Tobinx2dHochstadtx2c_and_David_Van_HornRelatively_complete_counterexamples_for_higherx2dorder_programsx2eIn_Procx2e_ACM_Conference_on_Programming_Language_Design_and_Implementationx2c_ppx2e_446x2dx2d4562015x22x29]{\AutobibLink{2015}}, \hyperref[t:x28autobib_x22Phxfac_Nguyx1ec5nx2c_Sam_Tobinx2dHochstadtx2c_and_David_Van_HornHigher_order_symbolic_execution_for_contract_verification_and_refutationx2eJournal_of_Functional_Programmingx2827x29x2c_ppx2e_e3x3a1x2dx2de3x3a542017x22x29]{\AutobibLink{2017}})}, we design a canonical form for functions that captures all
possible patterns of control structure
as different syntactic shapes of the body of a function.
Starting from the simplest shape, the constant function, we
gradually evolve the structural complexity of the inputs. In sum,
the SMT{-}driven exploration of the first{-}order aspects of the
inputs and the orthogonal evolution of their shape
work hand{-}in{-}hand to explore the space of the
higher{-}order function inputs to the user program.

Our design comes with a formal
model that justifies it. The model specifies a concolic tester for a functional
language and describes how the tester can evolve an initial, random
input to explore different parts of the behavior of the user program.
We validate our design idea in two ways:

\begin{itemize}\atItemizeStart

\item We prove that if a higher{-}order program in our model has a bug, then our model
 eventually constructs an input that
triggers the bug;

\item We provide a prototype concolic tester based on the model and
use it to uncover bugs in examples from the literature on concolic
testing and symbolic execution.\end{itemize}

The remainder of the paper is organized as follows.
\ChapRefUC{\SectionNumberLink{t:x28part_x22secx3abackgroundx22x29}{2}}{A Refresher on First{-}Order Concolic Testing} and \ChapRef{\SectionNumberLink{t:x28part_x22secx3aapproximatex22x29}{3}}{Canonical Functions Are All We Need} revisit
the two foundations of our work:
first{-}order concolic testing and function forms with
canonical control structure.
\ChapRefUC{\SectionNumberLink{t:x28part_x22secx3ahowx2dgeneratex22x29}{4}}{Directed Evolution of Canonical Functions} builds on these two elements
to demonstrate how our higher{-}order concolic tester gradually evolves
inputs to discover a bug in a concrete higher{-}order example.
  \ChapRefUC{\SectionNumberLink{t:x28part_x22secx3amodelx22x29}{5}}{Formalizing Higher{-}Order Concolic Testing}  introduces the formal model for our concolic
  tester and  \ChapRef{\SectionNumberLink{t:x28part_x22secx3atheoryx22x29}{6}}{Correctness of Higher{-}Order Concolic Testing} establishes its formal correctness
  properties. \ChapRefUC{\SectionNumberLink{t:x28part_x22secx3aprototypex22x29}{7}}{Prototype Implementation} describes our prototype
  implementation and how we use it to discover bugs in
  a corpus of examples from the literature.
  Finally, \ChapRefUC{\SectionNumberLink{t:x28part_x22secx3arelatedx22x29}{8}}{Related Work} places our results in the context of related
   work.

\sectionNewpage

\Ssection{A Refresher on First{-}Order Concolic Testing}{A Refresher on First{-}Order Concolic Testing}\label{t:x28part_x22secx3abackgroundx22x29}

The central idea behind first{-}order concolic testing is beautifully simple.
The concolic tester generates a first set of random inputs for the user
program and uses them to run the program in two modes at the same time. The
concrete mode is the same as ordinary evaluation. The symbolic one constructs a
symbolic formula that represents the path of the control{-}flow of the
program that the generated inputs exercise. After the concolic tester obtains
the formula, it negates one of its pieces and asks an SMT solver to
produce, if possible, a model for the tweaked formula. The model
corresponds to a new set of inputs that cause the user program
to follow a different control{-}flow path than the one due to the original
inputs {---} exactly the control{-}flow path that corresponds to the tweaked
formula. The concolic tester repeats this process and systematically
examines the control{-}flow graph of the user program until it eventually
discovers a bug or crosses a pre{-}defined time or memory threshold.

To make the discussion concrete, consider the following example:

\noindent \begin{SCentered}\raisebox{-2.8429687499999954bp}{\makebox[129.19062499999998bp][l]{\includegraphics[trim=2.4000000000000004 2.4000000000000004 2.4000000000000004 2.4000000000000004]{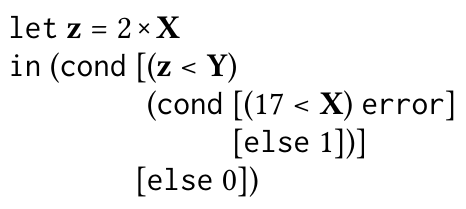}}}\end{SCentered}

\noindent Here, the inputs are represented
using the variables \raisebox{-3.1874999999999982bp}{\makebox[6.8921875bp][l]{\includegraphics[trim=2.4000000000000004 2.4000000000000004 2.4000000000000004 2.4000000000000004]{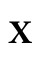}}} and \raisebox{-3.1874999999999982bp}{\makebox[5.98984375bp][l]{\includegraphics[trim=2.4000000000000004 2.4000000000000004 2.4000000000000004 2.4000000000000004]{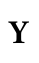}}}. The goal of the concolic tester is to generate \raisebox{-3.1874999999999982bp}{\makebox[6.8921875bp][l]{\includegraphics[trim=2.4000000000000004 2.4000000000000004 2.4000000000000004 2.4000000000000004]{pict_2.pdf}}} and \raisebox{-3.1874999999999982bp}{\makebox[5.98984375bp][l]{\includegraphics[trim=2.4000000000000004 2.4000000000000004 2.4000000000000004 2.4000000000000004]{pict_3.pdf}}} that
trigger the \raisebox{-3.1874999999999982bp}{\makebox[24.0bp][l]{\includegraphics[trim=2.4000000000000004 2.4000000000000004 2.4000000000000004 2.4000000000000004]{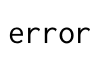}}}. In other words
the concolic tester has to cause both
\raisebox{-3.1874999999999982bp}{\makebox[34.30468750000001bp][l]{\includegraphics[trim=2.4000000000000004 2.4000000000000004 2.4000000000000004 2.4000000000000004]{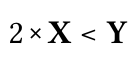}}} and \raisebox{-3.1874999999999982bp}{\makebox[25.8984375bp][l]{\includegraphics[trim=2.4000000000000004 2.4000000000000004 2.4000000000000004 2.4000000000000004]{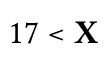}}} to be true.

Let{'}s assume that the concolic tester
generates an initial set of inputs where
\raisebox{-3.1874999999999982bp}{\makebox[6.8921875bp][l]{\includegraphics[trim=2.4000000000000004 2.4000000000000004 2.4000000000000004 2.4000000000000004]{pict_2.pdf}}} is \raisebox{-3.1874999999999982bp}{\makebox[4.46328125bp][l]{\includegraphics[trim=2.4000000000000004 2.4000000000000004 2.4000000000000004 2.4000000000000004]{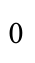}}} and \raisebox{-3.1874999999999982bp}{\makebox[5.98984375bp][l]{\includegraphics[trim=2.4000000000000004 2.4000000000000004 2.4000000000000004 2.4000000000000004]{pict_3.pdf}}} is \raisebox{-3.1874999999999982bp}{\makebox[4.46328125bp][l]{\includegraphics[trim=2.4000000000000004 2.4000000000000004 2.4000000000000004 2.4000000000000004]{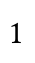}}}, which do not meet the above
conditions and fail to trigger the error.
In the concrete mode of the concolic tester, the initial inputs imply that
\raisebox{-3.1874999999999982bp}{\makebox[4.339062499999999bp][l]{\includegraphics[trim=2.4000000000000004 2.4000000000000004 2.4000000000000004 2.4000000000000004]{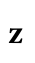}}} becomes $0$. At the same time, with its symbolic mode, the
concolic tester also tracks that
$0$ is the result of the expression \raisebox{-3.1874999999999982bp}{\makebox[18.235156250000003bp][l]{\includegraphics[trim=2.4000000000000004 2.4000000000000004 2.4000000000000004 2.4000000000000004]{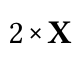}}}.
In general, any concrete value computed by primitive
operators comes with
an \emph{expression trace} indicating how the concrete value
relates to the inputs of the example.
For example, the expression \raisebox{-3.1874999999999982bp}{\makebox[20.40859375bp][l]{\includegraphics[trim=2.4000000000000004 2.4000000000000004 2.4000000000000004 2.4000000000000004]{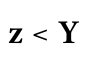}}} from
the outer \raisebox{-3.1874999999999982bp}{\makebox[19.200000000000003bp][l]{\includegraphics[trim=2.4000000000000004 2.4000000000000004 2.4000000000000004 2.4000000000000004]{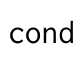}}} expression
evaluates to value $1$  with trace \raisebox{-3.1874999999999982bp}{\makebox[34.30468750000001bp][l]{\includegraphics[trim=2.4000000000000004 2.4000000000000004 2.4000000000000004 2.4000000000000004]{pict_5.pdf}}}.\NoteBox{\NoteContent{In
the example and the remainder of the paper, we use $1$
and $0$ to for true and false respectively.}}

Using the expression traces, the concolic
tester produces an ordered list of \emph{path constraints}
containing the concrete values and the associated expression traces
that determine whether the test of a branch in a conditional expression is
true or not.
For the above specific inputs and our example, the concolic tester records:

\begin{SCentered}\raisebox{-3.0164062499999975bp}{\makebox[65.50703125bp][l]{\includegraphics[trim=2.4000000000000004 2.4000000000000004 2.4000000000000004 2.4000000000000004]{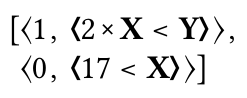}}}\end{SCentered}

In this list the first path constraint corresponds to
the outer \raisebox{-3.1874999999999982bp}{\makebox[19.200000000000003bp][l]{\includegraphics[trim=2.4000000000000004 2.4000000000000004 2.4000000000000004 2.4000000000000004]{pict_12.pdf}}} expression and captures that
the test of the  \raisebox{-3.1874999999999982bp}{\makebox[19.200000000000003bp][l]{\includegraphics[trim=2.4000000000000004 2.4000000000000004 2.4000000000000004 2.4000000000000004]{pict_12.pdf}}} evaluates to concrete value $1$
with expression trace
\raisebox{-3.1874999999999982bp}{\makebox[34.30468750000001bp][l]{\includegraphics[trim=2.4000000000000004 2.4000000000000004 2.4000000000000004 2.4000000000000004]{pict_5.pdf}}}, while the test of the inner \raisebox{-3.1874999999999982bp}{\makebox[19.200000000000003bp][l]{\includegraphics[trim=2.4000000000000004 2.4000000000000004 2.4000000000000004 2.4000000000000004]{pict_12.pdf}}} expression
produces  $0$ with trace \raisebox{-3.1874999999999982bp}{\makebox[25.8984375bp][l]{\includegraphics[trim=2.4000000000000004 2.4000000000000004 2.4000000000000004 2.4000000000000004]{pict_6.pdf}}}.
The order of the path constraints matches the evaluation of the example
and they induce the symbolic formula that represents the control{-}flow path
of the example for inputs where \raisebox{-3.1874999999999982bp}{\makebox[6.8921875bp][l]{\includegraphics[trim=2.4000000000000004 2.4000000000000004 2.4000000000000004 2.4000000000000004]{pict_2.pdf}}} is \raisebox{-3.1874999999999982bp}{\makebox[4.46328125bp][l]{\includegraphics[trim=2.4000000000000004 2.4000000000000004 2.4000000000000004 2.4000000000000004]{pict_7.pdf}}} and \raisebox{-3.1874999999999982bp}{\makebox[5.98984375bp][l]{\includegraphics[trim=2.4000000000000004 2.4000000000000004 2.4000000000000004 2.4000000000000004]{pict_3.pdf}}} is \raisebox{-3.1874999999999982bp}{\makebox[4.46328125bp][l]{\includegraphics[trim=2.4000000000000004 2.4000000000000004 2.4000000000000004 2.4000000000000004]{pict_8.pdf}}}:

\noindent \begin{SCentered}\raisebox{-3.1874999999999982bp}{\makebox[130.871875bp][l]{\includegraphics[trim=2.4000000000000004 2.4000000000000004 2.4000000000000004 2.4000000000000004]{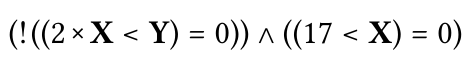}}}\end{SCentered}

\begin{Figure}\begin{Centerfigure}\begin{FigureInside}\raisebox{-0.14414062499999147bp}{\makebox[156.41875bp][l]{\includegraphics[trim=2.4000000000000004 2.4000000000000004 2.4000000000000004 2.4000000000000004]{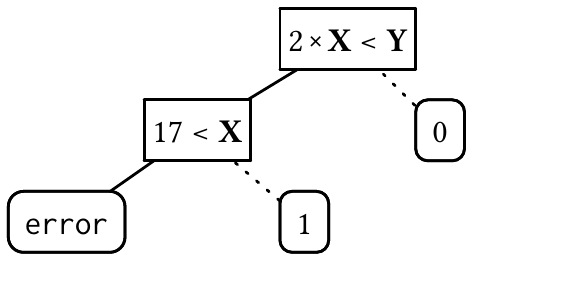}}}\end{FigureInside}\end{Centerfigure}

\Centertext{\Legend{\FigureTarget{\label{t:x28counter_x28x22figurex22_x22figx3apathx2dtreex22x29x29}\textsf{Fig.}~\textsf{1}. }{t:x28counter_x28x22figurex22_x22figx3apathx2dtreex22x29x29}\textsf{The Tree Representation of Control Paths}}}\end{Figure}

Put differently, the two inputs we examine reveal
two alternative control{-}flow paths
for the concolic tester
to explore {---} each corresponds to negating one of the two clauses of the
above formula.
Figure~\hyperref[t:x28counter_x28x22figurex22_x22figx3apathx2dtreex22x29x29]{\FigureRef{1}{t:x28counter_x28x22figurex22_x22figx3apathx2dtreex22x29x29}} depicts the current control{-}flow path
 and the two alternatives
as a tree with leaves containing the outcome of the example in each case.

As the concolic tester does not know which control{-}flow path
results in \raisebox{-3.1874999999999982bp}{\makebox[24.0bp][l]{\includegraphics[trim=2.4000000000000004 2.4000000000000004 2.4000000000000004 2.4000000000000004]{pict_4.pdf}}}, it explores all of them.
In order for the example to follow the right{-}most control{-}flow path,
the test of the outer \raisebox{-3.1874999999999982bp}{\makebox[19.200000000000003bp][l]{\includegraphics[trim=2.4000000000000004 2.4000000000000004 2.4000000000000004 2.4000000000000004]{pict_12.pdf}}} expression needs to
become \raisebox{-3.1874999999999982bp}{\makebox[4.46328125bp][l]{\includegraphics[trim=2.4000000000000004 2.4000000000000004 2.4000000000000004 2.4000000000000004]{pict_7.pdf}}}. Consequently, the concolic has to generate inputs
such that \raisebox{-3.1874999999999982bp}{\makebox[54.567968750000006bp][l]{\includegraphics[trim=2.4000000000000004 2.4000000000000004 2.4000000000000004 2.4000000000000004]{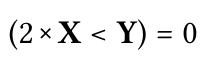}}}, i.e., the negation of the clause
from the first path constraint. The SMT solver
responds with the model where \raisebox{-3.1874999999999982bp}{\makebox[6.8921875bp][l]{\includegraphics[trim=2.4000000000000004 2.4000000000000004 2.4000000000000004 2.4000000000000004]{pict_2.pdf}}} is \raisebox{-3.1874999999999982bp}{\makebox[4.46328125bp][l]{\includegraphics[trim=2.4000000000000004 2.4000000000000004 2.4000000000000004 2.4000000000000004]{pict_8.pdf}}} and \raisebox{-3.1874999999999982bp}{\makebox[5.98984375bp][l]{\includegraphics[trim=2.4000000000000004 2.4000000000000004 2.4000000000000004 2.4000000000000004]{pict_3.pdf}}} is \raisebox{-3.1874999999999982bp}{\makebox[4.46328125bp][l]{\includegraphics[trim=2.4000000000000004 2.4000000000000004 2.4000000000000004 2.4000000000000004]{pict_7.pdf}}}, giving us our next
set of inputs.

Since these inputs do not result in an \raisebox{-3.1874999999999982bp}{\makebox[24.0bp][l]{\includegraphics[trim=2.4000000000000004 2.4000000000000004 2.4000000000000004 2.4000000000000004]{pict_4.pdf}}}, the search continues
and the concolic tester negates the second clause and issues a second query

\noindent \begin{SCentered}\raisebox{-3.1874999999999982bp}{\makebox[140.15625bp][l]{\includegraphics[trim=2.4000000000000004 2.4000000000000004 2.4000000000000004 2.4000000000000004]{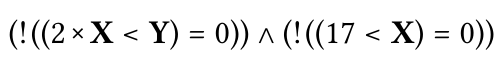}}}\end{SCentered}

\noindent The SMT solver responds by assigning \raisebox{-3.1874999999999982bp}{\makebox[6.8921875bp][l]{\includegraphics[trim=2.4000000000000004 2.4000000000000004 2.4000000000000004 2.4000000000000004]{pict_2.pdf}}} to \raisebox{-3.1874999999999982bp}{\makebox[8.9265625bp][l]{\includegraphics[trim=2.4000000000000004 2.4000000000000004 2.4000000000000004 2.4000000000000004]{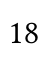}}} and \raisebox{-3.1874999999999982bp}{\makebox[5.98984375bp][l]{\includegraphics[trim=2.4000000000000004 2.4000000000000004 2.4000000000000004 2.4000000000000004]{pict_3.pdf}}} to \raisebox{-3.1874999999999982bp}{\makebox[8.9265625bp][l]{\includegraphics[trim=2.4000000000000004 2.4000000000000004 2.4000000000000004 2.4000000000000004]{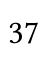}}},
causing the program to reach \raisebox{-3.1874999999999982bp}{\makebox[24.0bp][l]{\includegraphics[trim=2.4000000000000004 2.4000000000000004 2.4000000000000004 2.4000000000000004]{pict_4.pdf}}}.

In general terms, given a user program and a set of inputs,
the concolic tester runs the program and summarizes, in the form of path
constraints, the direction of the
control{-}flow branches the evaluation follows.
To explore a different control{-}flow path,
the  concolic tester
selects a prefix of the given list of path constraints,
negates its last element and consults the SMT solver
to generate new inputs for the next round of testing.
Consequently, with each iteration of the \emph{concolic loop}, the tester
explores an increasing portion of the control{-}flow graph of the user
program until it
discovers an error or hits a time or memory limit.

\sectionNewpage

\Ssection{Canonical Functions Are All We Need}{Canonical Functions Are All We Need}\label{t:x28part_x22secx3aapproximatex22x29}

In the setting of a higher{-}order functional language, a
concolic tester needs to generate not only numbers, but also
functions. Fortunately, it is possible to exhaustively
exercising all control paths of a higher{-}order user program,
with only a subset of the function terms of the language.

We can describe this subset with a small grammar of \emph{canonical
functions} that restrict the shape of function bodies. In essence,
each canonical shape of a functions translates to canonical pattern of
interaction between the input and the user program.  Intuitively, any
generated input first is invoked by the user program,
then inspects the results of calls it makes and finally
produces a reply.  We capture this pattern with a \raisebox{-3.1874999999999982bp}{\makebox[14.400000000000004bp][l]{\includegraphics[trim=2.4000000000000004 2.4000000000000004 2.4000000000000004 2.4000000000000004]{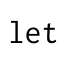}}} form followed by
a conditional. Of course, the body of the conditional can itself interact
(using the above pattern) or it might return a new canonical function.
This construction
generalizes to arbitrarily higher{-}order inputs\Autobibref{~(\hyperref[t:x28autobib_x22Phxfac_Nguyx1ec5nx2c_Sam_Tobinx2dHochstadtx2c_and_David_Van_HornRelatively_complete_counterexamples_for_higherx2dorder_programsx2eIn_Procx2e_ACM_Conference_on_Programming_Language_Design_and_Implementationx2c_ppx2e_446x2dx2d4562015x22x29]{\AutobibLink{Nguy\~{\^{e}}n et al\Sendabbrev{.}}} \hyperref[t:x28autobib_x22Phxfac_Nguyx1ec5nx2c_Sam_Tobinx2dHochstadtx2c_and_David_Van_HornRelatively_complete_counterexamples_for_higherx2dorder_programsx2eIn_Procx2e_ACM_Conference_on_Programming_Language_Design_and_Implementationx2c_ppx2e_446x2dx2d4562015x22x29]{\AutobibLink{2015}}, \hyperref[t:x28autobib_x22Phxfac_Nguyx1ec5nx2c_Sam_Tobinx2dHochstadtx2c_and_David_Van_HornHigher_order_symbolic_execution_for_contract_verification_and_refutationx2eJournal_of_Functional_Programmingx2827x29x2c_ppx2e_e3x3a1x2dx2de3x3a542017x22x29]{\AutobibLink{2017}})}.

In the remainder of this section we introduce our canonical
functions with a series of examples. The examples demonstrate
the key property of the canonical functions: if there is a function
that can trigger a bug in the user program, then there exists a
canonical function that also causes the user program to fail.
To get started, consider the following user program:

\noindent \begin{SCentered}\raisebox{-2.6742187499999996bp}{\makebox[97.68359375bp][l]{\includegraphics[trim=2.4000000000000004 2.4000000000000004 2.4000000000000004 2.4000000000000004]{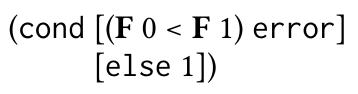}}}\end{SCentered}

\noindent It is apparent that \raisebox{-3.1874999999999982bp}{\makebox[5.231249999999999bp][l]{\includegraphics[trim=2.4000000000000004 2.4000000000000004 2.4000000000000004 2.4000000000000004]{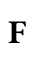}}} has to be
a first{-}order function from numbers to
numbers, but not a constant function, e.g.,
\raisebox{-3.1874999999999982bp}{\makebox[31.721875000000004bp][l]{\includegraphics[trim=2.4000000000000004 2.4000000000000004 2.4000000000000004 2.4000000000000004]{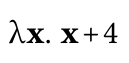}}}. However, to reproduce the behavior of this function
in this example, the concolic tester does not need to come up with the
arithmetic operation \raisebox{-3.1874999999999982bp}{\makebox[16.728125000000006bp][l]{\includegraphics[trim=2.4000000000000004 2.4000000000000004 2.4000000000000004 2.4000000000000004]{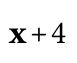}}}.
Rather, since the function \raisebox{-3.1874999999999982bp}{\makebox[5.231249999999999bp][l]{\includegraphics[trim=2.4000000000000004 2.4000000000000004 2.4000000000000004 2.4000000000000004]{pict_22.pdf}}} is only applied to 0 and 1, the concolic
tester can instead generate the function:

\noindent \begin{SCentered}\raisebox{-2.6742187499999996bp}{\makebox[81.66015625bp][l]{\includegraphics[trim=2.4000000000000004 2.4000000000000004 2.4000000000000004 2.4000000000000004]{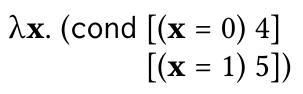}}}\end{SCentered}

More generally, if, like in our example, the evaluation of the user program terminates
for a particular input, there are only finitely many calls to that input.  Thus in the case where the
input is a first{-}order function, it suffices for the concolic
tester to produce a \raisebox{-3.1874999999999982bp}{\makebox[19.200000000000003bp][l]{\includegraphics[trim=2.4000000000000004 2.4000000000000004 2.4000000000000004 2.4000000000000004]{pict_12.pdf}}} expression that merely maps arguments from the code
under test to the corresponding results of the input without reconstructing
the actual computation performed by the input.

The situation is more involved when the input is higher{-}order itself.
Since the input can invoke its argument and then decide how to proceed
based on the result(s) of the call(s), the values that the input
provides to its argument act, in effect, as new inputs to the user
program.
The following variant of the previous example displays the issue:

\noindent \begin{SCentered}\raisebox{-2.6742187499999996bp}{\makebox[151.31875bp][l]{\includegraphics[trim=2.4000000000000004 2.4000000000000004 2.4000000000000004 2.4000000000000004]{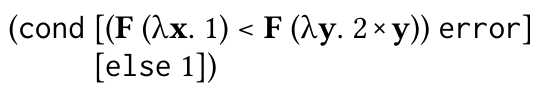}}}\end{SCentered}

\noindent A function that can cause this example to fail is  \raisebox{-3.1874999999999982bp}{\makebox[5.231249999999999bp][l]{\includegraphics[trim=2.4000000000000004 2.4000000000000004 2.4000000000000004 2.4000000000000004]{pict_22.pdf}}} is \raisebox{-3.1874999999999982bp}{\makebox[51.20156250000001bp][l]{\includegraphics[trim=2.4000000000000004 2.4000000000000004 2.4000000000000004 2.4000000000000004]{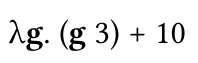}}}.
However, similar to the previous example, the concolic tester can achieve the same
outcome by generating a different function.
The necessary details of the shape of the function become evident
if we take a closer look at the interaction between the
input  \raisebox{-3.1874999999999982bp}{\makebox[51.20156250000001bp][l]{\includegraphics[trim=2.4000000000000004 2.4000000000000004 2.4000000000000004 2.4000000000000004]{pict_27.pdf}}} and the code under test:

\begin{SCentered}\begin{tikzpicture}
\node[rectangle, draw, minimum height=3.2cm](funF) at (-3.5,0) {\;\;\;\;
$\raisebox{-3.1874999999999982bp}{\makebox[5.231249999999999bp][l]{\includegraphics[trim=2.4000000000000004 2.4000000000000004 2.4000000000000004 2.4000000000000004]{pict_22.pdf}}} \equiv \raisebox{-3.1874999999999982bp}{\makebox[51.20156250000001bp][l]{\includegraphics[trim=2.4000000000000004 2.4000000000000004 2.4000000000000004 2.4000000000000004]{pict_27.pdf}}}$
\;\;\;\;};
\node[rectangle, draw, minimum height=3.2cm](prog) at (3.5,0) {\;\;\;\;
\raisebox{-3.1874999999999982bp}{\makebox[87.90390624999999bp][l]{\includegraphics[trim=2.4000000000000004 2.4000000000000004 2.4000000000000004 2.4000000000000004]{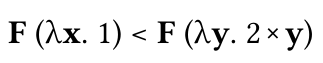}}}
\;\;\;\;};
\draw[->,thick] (1.23,1.4) -- node[midway,above]{\raisebox{-3.1874999999999982bp}{\makebox[19.45703125bp][l]{\includegraphics[trim=2.4000000000000004 2.4000000000000004 2.4000000000000004 2.4000000000000004]{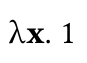}}}} (-1.67,1.4);
\node at (-2,0.9) {\raisebox{-3.1874999999999982bp}{\makebox[17.585156250000004bp][l]{\includegraphics[trim=2.4000000000000004 2.4000000000000004 2.4000000000000004 2.4000000000000004]{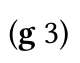}}}};
\draw[->] (-1.67,0.81) -- node[midway,above]{3} (1.23,0.81);
\draw[->] (1.23,0.39) -- node[midway,above]{1} (-1.67,0.39);
\node at (2,0.6) {\raisebox{-3.1874999999999982bp}{\makebox[5.0015625bp][l]{\includegraphics[trim=2.4000000000000004 2.4000000000000004 2.4000000000000004 2.4000000000000004]{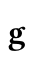}}} $\equiv$ \raisebox{-3.1874999999999982bp}{\makebox[19.45703125bp][l]{\includegraphics[trim=2.4000000000000004 2.4000000000000004 2.4000000000000004 2.4000000000000004]{pict_29.pdf}}}};
\draw[->,thick] (-1.67,-0.2) -- node[midway,above]{$1+10=2$} (1.23,-0.2);
\draw[->,thick] (1.23,-1.3) -- node[midway,above]{\raisebox{-3.1874999999999982bp}{\makebox[31.664062500000007bp][l]{\includegraphics[trim=2.4000000000000004 2.4000000000000004 2.4000000000000004 2.4000000000000004]{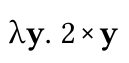}}}} (-1.67,-1.3);
\node at (-0.3,-1.5) {$\vdots$};
\end{tikzpicture}\end{SCentered}

\noindent To read the diagram, start from the top, considering each line to
be an interaction between the user program and it{'}s input. The first
interaction is sending \raisebox{-3.1874999999999982bp}{\makebox[19.45703125bp][l]{\includegraphics[trim=2.4000000000000004 2.4000000000000004 2.4000000000000004 2.4000000000000004]{pict_29.pdf}}} to \raisebox{-3.1874999999999982bp}{\makebox[5.231249999999999bp][l]{\includegraphics[trim=2.4000000000000004 2.4000000000000004 2.4000000000000004 2.4000000000000004]{pict_22.pdf}}}.
\raisebox{-3.1874999999999982bp}{\makebox[5.231249999999999bp][l]{\includegraphics[trim=2.4000000000000004 2.4000000000000004 2.4000000000000004 2.4000000000000004]{pict_22.pdf}}} reacts with  the application \raisebox{-3.1874999999999982bp}{\makebox[11.864843749999999bp][l]{\includegraphics[trim=2.4000000000000004 2.4000000000000004 2.4000000000000004 2.4000000000000004]{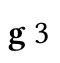}}}
that transmits \raisebox{-3.1874999999999982bp}{\makebox[4.46328125bp][l]{\includegraphics[trim=2.4000000000000004 2.4000000000000004 2.4000000000000004 2.4000000000000004]{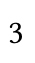}}} back to the user program which
replies with \raisebox{-3.1874999999999982bp}{\makebox[4.46328125bp][l]{\includegraphics[trim=2.4000000000000004 2.4000000000000004 2.4000000000000004 2.4000000000000004]{pict_8.pdf}}}. Subsequently \raisebox{-3.1874999999999982bp}{\makebox[5.231249999999999bp][l]{\includegraphics[trim=2.4000000000000004 2.4000000000000004 2.4000000000000004 2.4000000000000004]{pict_22.pdf}}} adds \raisebox{-3.1874999999999982bp}{\makebox[8.9265625bp][l]{\includegraphics[trim=2.4000000000000004 2.4000000000000004 2.4000000000000004 2.4000000000000004]{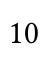}}} to \raisebox{-3.1874999999999982bp}{\makebox[4.46328125bp][l]{\includegraphics[trim=2.4000000000000004 2.4000000000000004 2.4000000000000004 2.4000000000000004]{pict_8.pdf}}} to obtain the
result \raisebox{-3.1874999999999982bp}{\makebox[8.9265625bp][l]{\includegraphics[trim=2.4000000000000004 2.4000000000000004 2.4000000000000004 2.4000000000000004]{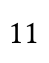}}} and ends its first interaction with the user program.
After this first round of interaction, the user program
initiates a second round by sending \raisebox{-3.1874999999999982bp}{\makebox[31.664062500000007bp][l]{\includegraphics[trim=2.4000000000000004 2.4000000000000004 2.4000000000000004 2.4000000000000004]{pict_32.pdf}}} to \raisebox{-3.1874999999999982bp}{\makebox[5.231249999999999bp][l]{\includegraphics[trim=2.4000000000000004 2.4000000000000004 2.4000000000000004 2.4000000000000004]{pict_22.pdf}}}.
While \raisebox{-3.1874999999999982bp}{\makebox[5.231249999999999bp][l]{\includegraphics[trim=2.4000000000000004 2.4000000000000004 2.4000000000000004 2.4000000000000004]{pict_22.pdf}}} sends again \raisebox{-3.1874999999999982bp}{\makebox[4.46328125bp][l]{\includegraphics[trim=2.4000000000000004 2.4000000000000004 2.4000000000000004 2.4000000000000004]{pict_34.pdf}}}
back to the user program, this time the user program multiplies it by \raisebox{-3.1874999999999982bp}{\makebox[4.46328125bp][l]{\includegraphics[trim=2.4000000000000004 2.4000000000000004 2.4000000000000004 2.4000000000000004]{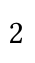}}}
and responds with \raisebox{-3.1874999999999982bp}{\makebox[4.46328125bp][l]{\includegraphics[trim=2.4000000000000004 2.4000000000000004 2.4000000000000004 2.4000000000000004]{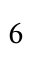}}}. Finally, \raisebox{-3.1874999999999982bp}{\makebox[5.231249999999999bp][l]{\includegraphics[trim=2.4000000000000004 2.4000000000000004 2.4000000000000004 2.4000000000000004]{pict_22.pdf}}} adds \raisebox{-3.1874999999999982bp}{\makebox[8.9265625bp][l]{\includegraphics[trim=2.4000000000000004 2.4000000000000004 2.4000000000000004 2.4000000000000004]{pict_35.pdf}}}
to \raisebox{-3.1874999999999982bp}{\makebox[4.46328125bp][l]{\includegraphics[trim=2.4000000000000004 2.4000000000000004 2.4000000000000004 2.4000000000000004]{pict_38.pdf}}} to produce \raisebox{-3.1874999999999982bp}{\makebox[8.9265625bp][l]{\includegraphics[trim=2.4000000000000004 2.4000000000000004 2.4000000000000004 2.4000000000000004]{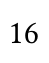}}}, its reply for the second and final interaction.

In this chain of interaction, \raisebox{-3.1874999999999982bp}{\makebox[5.231249999999999bp][l]{\includegraphics[trim=2.4000000000000004 2.4000000000000004 2.4000000000000004 2.4000000000000004]{pict_22.pdf}}} performs
two actions: first
it sends a number back to the user program, requesting
a new number, and then, similar to a first{-}order function, it maps that number to
another number. Thus we can rewrite
\raisebox{-3.1874999999999982bp}{\makebox[5.231249999999999bp][l]{\includegraphics[trim=2.4000000000000004 2.4000000000000004 2.4000000000000004 2.4000000000000004]{pict_22.pdf}}} to separate the two actions syntactically:

\begin{SCentered}\raisebox{-2.6742187499999996bp}{\makebox[57.69375000000001bp][l]{\includegraphics[trim=2.4000000000000004 2.4000000000000004 2.4000000000000004 2.4000000000000004]{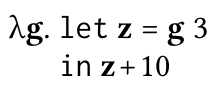}}}\end{SCentered}

This revamped shape suggests how we can replace the input
with a canonical function that behaves just like the input,
inducing the same interactions with the user program
as the input does. Building on our discussion about first{-}order
function inputs above, we construct a \raisebox{-3.1874999999999982bp}{\makebox[19.200000000000003bp][l]{\includegraphics[trim=2.4000000000000004 2.4000000000000004 2.4000000000000004 2.4000000000000004]{pict_12.pdf}}} expression that dispatches on \raisebox{-3.1874999999999982bp}{\makebox[4.339062499999999bp][l]{\includegraphics[trim=2.4000000000000004 2.4000000000000004 2.4000000000000004 2.4000000000000004]{pict_9.pdf}}}
and avoid the use of the arithmetic operator:

\begin{SCentered}\raisebox{-2.7304687499999982bp}{\makebox[96.69375bp][l]{\includegraphics[trim=2.4000000000000004 2.4000000000000004 2.4000000000000004 2.4000000000000004]{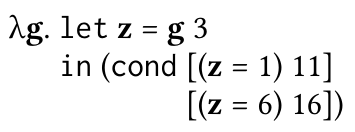}}}\end{SCentered}

In sum, canonical functions can capture all patterns of interaction
between the user program  and the input despite their strict structure.  We revisit
canonical functions and their properties formally in \ChapRef{\SectionNumberLink{t:x28part_x22secx3amodelx22x29}{5}}{Formalizing Higher{-}Order Concolic Testing} and
\ChapRef{\SectionNumberLink{t:x28part_x22secx3atheoryx22x29}{6}}{Correctness of Higher{-}Order Concolic Testing}. In the following section, we discuss how their
 structure helps generate inputs that uncover bugs in user programs.

\sectionNewpage

\Ssection{Directed Evolution of Canonical Functions}{Directed Evolution of Canonical Functions}\label{t:x28part_x22secx3ahowx2dgeneratex22x29}

Canonical functions allow the concolic tester to limit
the space of inputs it needs to search to trigger an error.  Even better,
the structure of our canonical functions makes it possible for the
concolic tester to search the space of canonical inputs in a targeted
fashion. The concolic loop starts with a constant
function. Then, based on the interaction between the input and the user program,
the input evolves into a more complex one,
at each evolution step hoping to explore a different path of the control{-}flow
of the user program. This process of evolution of the term is guaranteed to find
the fittest function: no matter where a bug might be in the program, a canonical input can trigger it.
This section explains the evolution of canonical functions by example.

To start, consider the following program where \raisebox{-3.1874999999999982bp}{\makebox[5.231249999999999bp][l]{\includegraphics[trim=2.4000000000000004 2.4000000000000004 2.4000000000000004 2.4000000000000004]{pict_22.pdf}}} is its higher{-}order input:

\noindent \begin{SCentered}\raisebox{-2.786718749999997bp}{\makebox[143.7125bp][l]{\includegraphics[trim=2.4000000000000004 2.4000000000000004 2.4000000000000004 2.4000000000000004]{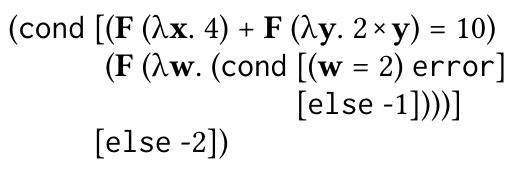}}}\end{SCentered}

\noindent The example applies \raisebox{-3.1874999999999982bp}{\makebox[5.231249999999999bp][l]{\includegraphics[trim=2.4000000000000004 2.4000000000000004 2.4000000000000004 2.4000000000000004]{pict_22.pdf}}} to two functions and
if the sum of the two applications is equal to \raisebox{-3.1874999999999982bp}{\makebox[8.9265625bp][l]{\includegraphics[trim=2.4000000000000004 2.4000000000000004 2.4000000000000004 2.4000000000000004]{pict_35.pdf}}},
it calls \raisebox{-3.1874999999999982bp}{\makebox[5.231249999999999bp][l]{\includegraphics[trim=2.4000000000000004 2.4000000000000004 2.4000000000000004 2.4000000000000004]{pict_22.pdf}}} with a third function which signals an
\raisebox{-3.1874999999999982bp}{\makebox[24.0bp][l]{\includegraphics[trim=2.4000000000000004 2.4000000000000004 2.4000000000000004 2.4000000000000004]{pict_4.pdf}}}.

The concolic tester starts with the simplest canonical function \raisebox{-3.2562499999999996bp}{\makebox[79.03828125bp][l]{\includegraphics[trim=2.4000000000000004 2.4000000000000004 2.4000000000000004 2.4000000000000004]{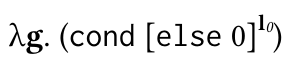}}}
as the first input for the example.  The peculiar shape of
the function is due to the shape of canonical functions that
requires that their bodies
consists of \raisebox{-3.1874999999999982bp}{\makebox[14.400000000000004bp][l]{\includegraphics[trim=2.4000000000000004 2.4000000000000004 2.4000000000000004 2.4000000000000004]{pict_20.pdf}}} and/or \raisebox{-3.1874999999999982bp}{\makebox[19.200000000000003bp][l]{\includegraphics[trim=2.4000000000000004 2.4000000000000004 2.4000000000000004 2.4000000000000004]{pict_12.pdf}}} expressions. The label
\raisebox{-2.3617187499999996bp}{\makebox[5.60546875bp][l]{\includegraphics[trim=2.4000000000000004 2.4000000000000004 2.4000000000000004 2.4000000000000004]{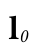}}} is a unique identifier for the conditional branch; we expand on the
role of branch labels further on.  With that input, both calls to
\raisebox{-3.1874999999999982bp}{\makebox[5.231249999999999bp][l]{\includegraphics[trim=2.4000000000000004 2.4000000000000004 2.4000000000000004 2.4000000000000004]{pict_22.pdf}}} return to \raisebox{-3.1874999999999982bp}{\makebox[4.46328125bp][l]{\includegraphics[trim=2.4000000000000004 2.4000000000000004 2.4000000000000004 2.4000000000000004]{pict_7.pdf}}} and the example returns \raisebox{-3.1874999999999982bp}{\makebox[7.707812499999999bp][l]{\includegraphics[trim=2.4000000000000004 2.4000000000000004 2.4000000000000004 2.4000000000000004]{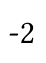}}}.

Following our discussion of first{-}order concolic testing in
\ChapRef{\SectionNumberLink{t:x28part_x22secx3abackgroundx22x29}{2}}{A Refresher on First{-}Order Concolic Testing}, the tester records
\emph{first{-}order} path constraints that encode the control{-}flow path of
the example (but not its inputs). The single constraint here is  \raisebox{-3.2890624999999982bp}{\makebox[59.68593750000001bp][l]{\includegraphics[trim=2.4000000000000004 2.4000000000000004 2.4000000000000004 2.4000000000000004]{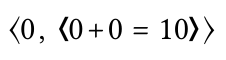}}}
and indicates that the test of the conditional of the
example evaluates to \raisebox{-3.1874999999999982bp}{\makebox[4.46328125bp][l]{\includegraphics[trim=2.4000000000000004 2.4000000000000004 2.4000000000000004 2.4000000000000004]{pict_7.pdf}}}  because the sum of the two constants \raisebox{-3.1874999999999982bp}{\makebox[4.46328125bp][l]{\includegraphics[trim=2.4000000000000004 2.4000000000000004 2.4000000000000004 2.4000000000000004]{pict_7.pdf}}} is
not equal to \raisebox{-3.1874999999999982bp}{\makebox[8.9265625bp][l]{\includegraphics[trim=2.4000000000000004 2.4000000000000004 2.4000000000000004 2.4000000000000004]{pict_35.pdf}}}.  On its own this constraint is not sufficient for the
concolic tester to generate a new input; it lacks any connection with
\raisebox{-3.1874999999999982bp}{\makebox[5.231249999999999bp][l]{\includegraphics[trim=2.4000000000000004 2.4000000000000004 2.4000000000000004 2.4000000000000004]{pict_22.pdf}}}.

To remedy this shortcoming, our higher{-}order concolic tester
records constraints about the evaluation of the input that
it generates. As each \raisebox{-3.1874999999999982bp}{\makebox[19.200000000000003bp][l]{\includegraphics[trim=2.4000000000000004 2.4000000000000004 2.4000000000000004 2.4000000000000004]{pict_12.pdf}}} expression inspects the properties of
one specific value (checking to see which integer it is, or
if it is a procedure), we record a constraint that captures
that value followed by a series of constraints that capture the
queries of the value. The former are called \emph{test}
constraints, and they record both a unique label for the
\raisebox{-3.1874999999999982bp}{\makebox[19.200000000000003bp][l]{\includegraphics[trim=2.4000000000000004 2.4000000000000004 2.4000000000000004 2.4000000000000004]{pict_12.pdf}}} expression and the value being inspected. The latter are
called \emph{branch} constraints, and they each record three
pieces of information: a unique label for the branch, the
outcome of its test (if the branch is taken or not), and the
expression trace of the outcome. Our running example results
into four input{-}related constraints in addition to the
first{-}order constraint from above:

\begin{SCentered}\raisebox{-2.743749999999997bp}{\makebox[172.571875bp][l]{\includegraphics[trim=2.4000000000000004 2.4000000000000004 2.4000000000000004 2.4000000000000004]{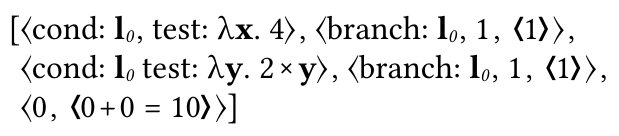}}}\end{SCentered}

The input{-}related constraints form two pairs (that show up on the same
line in the list above).
The first pair corresponds to the application,
\raisebox{-3.1874999999999982bp}{\makebox[32.80859375000001bp][l]{\includegraphics[trim=2.4000000000000004 2.4000000000000004 2.4000000000000004 2.4000000000000004]{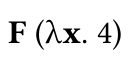}}}, and the resulting evaluation of \raisebox{-3.1874999999999982bp}{\makebox[5.231249999999999bp][l]{\includegraphics[trim=2.4000000000000004 2.4000000000000004 2.4000000000000004 2.4000000000000004]{pict_22.pdf}}}{'}s \raisebox{-3.1874999999999982bp}{\makebox[19.200000000000003bp][l]{\includegraphics[trim=2.4000000000000004 2.4000000000000004 2.4000000000000004 2.4000000000000004]{pict_12.pdf}}} expression,
while the second corresponds to the application \raisebox{-3.1874999999999982bp}{\makebox[45.01562500000001bp][l]{\includegraphics[trim=2.4000000000000004 2.4000000000000004 2.4000000000000004 2.4000000000000004]{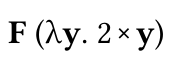}}}.
This grouping of input{-}related constraints is a general pattern; every
execution of a \raisebox{-3.1874999999999982bp}{\makebox[19.200000000000003bp][l]{\includegraphics[trim=2.4000000000000004 2.4000000000000004 2.4000000000000004 2.4000000000000004]{pict_12.pdf}}} expression in the input adds
a consecutive block of constraints that starts with a test constraint and
ends with a branch constraint whose test evaluates to \raisebox{-3.1874999999999982bp}{\makebox[4.46328125bp][l]{\includegraphics[trim=2.4000000000000004 2.4000000000000004 2.4000000000000004 2.4000000000000004]{pict_8.pdf}}}. If this last
constraint has the same label as the test constraint of the \raisebox{-3.1874999999999982bp}{\makebox[19.200000000000003bp][l]{\includegraphics[trim=2.4000000000000004 2.4000000000000004 2.4000000000000004 2.4000000000000004]{pict_12.pdf}}} expression,
like these two do, evaluation follows the \raisebox{-3.1874999999999982bp}{\makebox[19.200000000000003bp][l]{\includegraphics[trim=2.4000000000000004 2.4000000000000004 2.4000000000000004 2.4000000000000004]{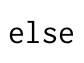}}} branch of the \raisebox{-3.1874999999999982bp}{\makebox[19.200000000000003bp][l]{\includegraphics[trim=2.4000000000000004 2.4000000000000004 2.4000000000000004 2.4000000000000004]{pict_12.pdf}}} expression.

Because of the pattern of the input{-}related constraints, the concolic tester can infer
that the default input ignores its argument \raisebox{-3.1874999999999982bp}{\makebox[5.0015625bp][l]{\includegraphics[trim=2.4000000000000004 2.4000000000000004 2.4000000000000004 2.4000000000000004]{pict_31.pdf}}} for both applications. This opens up a possibility for how the concolic
  tester can generate a new input that influences the example in a
  different way. That is, the input can evolve to one that interacts with \raisebox{-3.1874999999999982bp}{\makebox[5.0015625bp][l]{\includegraphics[trim=2.4000000000000004 2.4000000000000004 2.4000000000000004 2.4000000000000004]{pict_31.pdf}}}.
 The simplest one simply checks if \raisebox{-3.1874999999999982bp}{\makebox[5.0015625bp][l]{\includegraphics[trim=2.4000000000000004 2.4000000000000004 2.4000000000000004 2.4000000000000004]{pict_31.pdf}}} is
  a function:

\begin{SCentered}\raisebox{-2.950781249999997bp}{\makebox[118.33046875bp][l]{\includegraphics[trim=2.4000000000000004 2.4000000000000004 2.4000000000000004 2.4000000000000004]{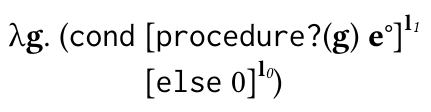}}}\end{SCentered}

\noindent where \raisebox{-3.1874999999999982bp}{\makebox[7.25625bp][l]{\includegraphics[trim=2.4000000000000004 2.4000000000000004 2.4000000000000004 2.4000000000000004]{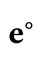}}} stands for the various ways the input proceeds
after establishing that \raisebox{-3.1874999999999982bp}{\makebox[5.0015625bp][l]{\includegraphics[trim=2.4000000000000004 2.4000000000000004 2.4000000000000004 2.4000000000000004]{pict_31.pdf}}} is indeed a function.

There are two general options for \raisebox{-3.1874999999999982bp}{\makebox[7.25625bp][l]{\includegraphics[trim=2.4000000000000004 2.4000000000000004 2.4000000000000004 2.4000000000000004]{pict_52.pdf}}} based on the shape of canonical
functions. The first is for the input to
ignore what it has learned about \raisebox{-3.1874999999999982bp}{\makebox[5.0015625bp][l]{\includegraphics[trim=2.4000000000000004 2.4000000000000004 2.4000000000000004 2.4000000000000004]{pict_31.pdf}}}, and
simply return a constant number, a
constant function, or a variable in scope
(in this case, the only such variable is \raisebox{-3.1874999999999982bp}{\makebox[5.0015625bp][l]{\includegraphics[trim=2.4000000000000004 2.4000000000000004 2.4000000000000004 2.4000000000000004]{pict_31.pdf}}} itself). Alternatively, the input
can call \raisebox{-3.1874999999999982bp}{\makebox[5.0015625bp][l]{\includegraphics[trim=2.4000000000000004 2.4000000000000004 2.4000000000000004 2.4000000000000004]{pict_31.pdf}}} (or in fact any other function in scope) which translates to an
\raisebox{-3.1874999999999982bp}{\makebox[7.25625bp][l]{\includegraphics[trim=2.4000000000000004 2.4000000000000004 2.4000000000000004 2.4000000000000004]{pict_52.pdf}}} of the form:

\noindent \begin{SCentered}\raisebox{-3.2562499999999996bp}{\makebox[124.92656249999999bp][l]{\includegraphics[trim=2.4000000000000004 2.4000000000000004 2.4000000000000004 2.4000000000000004]{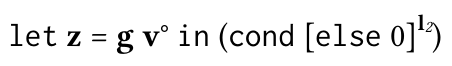}}}\end{SCentered}

Any of the above options could trigger the error in the example. Thus the
concolic tester may have to ultimately explore all of them. However, for the sake
of the discussion, we assume that the concolic tester picks a constant function
and stashes away the rest for future consideration. Consequently the concrete input for the next
iteration of the concolic loop is

\noindent \begin{SCentered}\raisebox{-3.0070312499999954bp}{\makebox[180.17578125bp][l]{\includegraphics[trim=2.4000000000000004 2.4000000000000004 2.4000000000000004 2.4000000000000004]{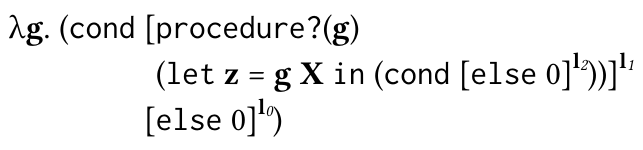}}}\end{SCentered}

\noindent Here, instead of some concrete value, \raisebox{-3.1874999999999982bp}{\makebox[5.0015625bp][l]{\includegraphics[trim=2.4000000000000004 2.4000000000000004 2.4000000000000004 2.4000000000000004]{pict_31.pdf}}} is applied to
 \raisebox{-3.1874999999999982bp}{\makebox[6.8921875bp][l]{\includegraphics[trim=2.4000000000000004 2.4000000000000004 2.4000000000000004 2.4000000000000004]{pict_2.pdf}}}. This is an important piece of the design of our higher{-}order concolic tester.
 In essence, \raisebox{-3.1874999999999982bp}{\makebox[6.8921875bp][l]{\includegraphics[trim=2.4000000000000004 2.4000000000000004 2.4000000000000004 2.4000000000000004]{pict_2.pdf}}} is an additional input that the concolic
 tester controls and can use to further affect the evaluation of the example.
 In other words, the evolution of the shape of the input aims to offer to the
 concolic tester new channels through which it can influence the example.

If we assume that the concolic tester randomly picks the initial value of  \raisebox{-3.1874999999999982bp}{\makebox[6.8921875bp][l]{\includegraphics[trim=2.4000000000000004 2.4000000000000004 2.4000000000000004 2.4000000000000004]{pict_2.pdf}}} to be \raisebox{-3.1874999999999982bp}{\makebox[4.46328125bp][l]{\includegraphics[trim=2.4000000000000004 2.4000000000000004 2.4000000000000004 2.4000000000000004]{pict_34.pdf}}},
 the application \raisebox{-3.1874999999999982bp}{\makebox[32.80859375000001bp][l]{\includegraphics[trim=2.4000000000000004 2.4000000000000004 2.4000000000000004 2.4000000000000004]{pict_48.pdf}}} reduces to the \raisebox{-2.3617187499999996bp}{\makebox[5.60546875bp][l]{\includegraphics[trim=2.4000000000000004 2.4000000000000004 2.4000000000000004 2.4000000000000004]{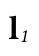}}} branch of
 \raisebox{-3.1874999999999982bp}{\makebox[5.231249999999999bp][l]{\includegraphics[trim=2.4000000000000004 2.4000000000000004 2.4000000000000004 2.4000000000000004]{pict_22.pdf}}}, producing the input{-}related constraints
\raisebox{-2.3617187499999996bp}{\makebox[79.55156249999999bp][l]{\includegraphics[trim=2.4000000000000004 2.4000000000000004 2.4000000000000004 2.4000000000000004]{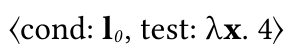}}}
and \raisebox{-3.2890624999999982bp}{\makebox[72.084375bp][l]{\includegraphics[trim=2.4000000000000004 2.4000000000000004 2.4000000000000004 2.4000000000000004]{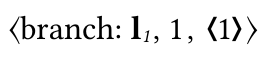}}}. The latter
indicates that the \raisebox{-3.1874999999999982bp}{\makebox[48.00000000000001bp][l]{\includegraphics[trim=2.4000000000000004 2.4000000000000004 2.4000000000000004 2.4000000000000004]{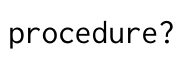}}}\raisebox{-3.1874999999999982bp}{\makebox[10.721875bp][l]{\includegraphics[trim=2.4000000000000004 2.4000000000000004 2.4000000000000004 2.4000000000000004]{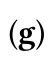}}} test succeeds, i.e., its result is \raisebox{-3.1874999999999982bp}{\makebox[4.46328125bp][l]{\includegraphics[trim=2.4000000000000004 2.4000000000000004 2.4000000000000004 2.4000000000000004]{pict_8.pdf}}}.
 Next, the application \raisebox{-3.1874999999999982bp}{\makebox[14.293750000000001bp][l]{\includegraphics[trim=2.4000000000000004 2.4000000000000004 2.4000000000000004 2.4000000000000004]{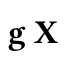}}} inside branch \raisebox{-2.3617187499999996bp}{\makebox[5.60546875bp][l]{\includegraphics[trim=2.4000000000000004 2.4000000000000004 2.4000000000000004 2.4000000000000004]{pict_55.pdf}}} returns \raisebox{-3.1874999999999982bp}{\makebox[4.46328125bp][l]{\includegraphics[trim=2.4000000000000004 2.4000000000000004 2.4000000000000004 2.4000000000000004]{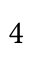}}} and
 the nested \raisebox{-3.1874999999999982bp}{\makebox[19.200000000000003bp][l]{\includegraphics[trim=2.4000000000000004 2.4000000000000004 2.4000000000000004 2.4000000000000004]{pict_12.pdf}}} expression follows the \raisebox{-3.1874999999999982bp}{\makebox[19.200000000000003bp][l]{\includegraphics[trim=2.4000000000000004 2.4000000000000004 2.4000000000000004 2.4000000000000004]{pict_50.pdf}}} branch with label
 \raisebox{-2.3617187499999996bp}{\makebox[5.60546875bp][l]{\includegraphics[trim=2.4000000000000004 2.4000000000000004 2.4000000000000004 2.4000000000000004]{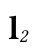}}} which induces
 the input{-}related constraints
\raisebox{-3.2890624999999982bp}{\makebox[72.9171875bp][l]{\includegraphics[trim=2.4000000000000004 2.4000000000000004 2.4000000000000004 2.4000000000000004]{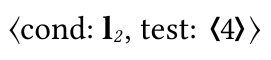}}}
and \raisebox{-3.2890624999999982bp}{\makebox[72.084375bp][l]{\includegraphics[trim=2.4000000000000004 2.4000000000000004 2.4000000000000004 2.4000000000000004]{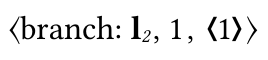}}}.
 The second application of \raisebox{-3.1874999999999982bp}{\makebox[5.231249999999999bp][l]{\includegraphics[trim=2.4000000000000004 2.4000000000000004 2.4000000000000004 2.4000000000000004]{pict_22.pdf}}}, \raisebox{-3.1874999999999982bp}{\makebox[45.01562500000001bp][l]{\includegraphics[trim=2.4000000000000004 2.4000000000000004 2.4000000000000004 2.4000000000000004]{pict_49.pdf}}}, proceeds in
 a similar manner.  In summary, we obtain the following new list of path
 constraints:

\begin{SCentered}\raisebox{-3.1984374999999954bp}{\makebox[181.20390625bp][l]{\includegraphics[trim=2.4000000000000004 2.4000000000000004 2.4000000000000004 2.4000000000000004]{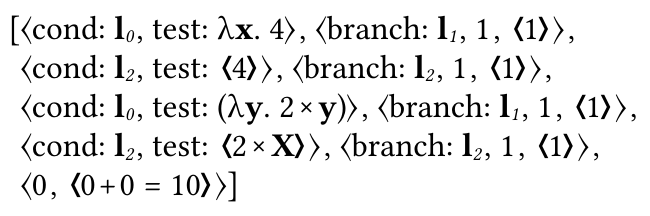}}}\end{SCentered}

It is worth noting that, in contrast to all the other test
constraints so far,  the second test constraints for the nested
\raisebox{-3.1874999999999982bp}{\makebox[19.200000000000003bp][l]{\includegraphics[trim=2.4000000000000004 2.4000000000000004 2.4000000000000004 2.4000000000000004]{pict_12.pdf}}} expression (with label \raisebox{-2.3617187499999996bp}{\makebox[5.60546875bp][l]{\includegraphics[trim=2.4000000000000004 2.4000000000000004 2.4000000000000004 2.4000000000000004]{pict_62.pdf}}}) includes an expression
trace (\raisebox{-3.2890624999999982bp}{\makebox[26.594531250000003bp][l]{\includegraphics[trim=2.4000000000000004 2.4000000000000004 2.4000000000000004 2.4000000000000004]{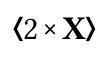}}}) instead of the actual value \raisebox{-3.1874999999999982bp}{\makebox[4.46328125bp][l]{\includegraphics[trim=2.4000000000000004 2.4000000000000004 2.4000000000000004 2.4000000000000004]{pict_38.pdf}}}.
This is because the tests in this \raisebox{-3.1874999999999982bp}{\makebox[19.200000000000003bp][l]{\includegraphics[trim=2.4000000000000004 2.4000000000000004 2.4000000000000004 2.4000000000000004]{pict_12.pdf}}} expression inspect a first{-}order value.
Hence the concolic tester keeps track of how the value depends on
first{-}order inputs so that it can
adjust those inputs to affect what branch of
the \raisebox{-3.1874999999999982bp}{\makebox[19.200000000000003bp][l]{\includegraphics[trim=2.4000000000000004 2.4000000000000004 2.4000000000000004 2.4000000000000004]{pict_12.pdf}}} expression that the evaluation follows.

With this second list of constraints, the input can evolve
further. For example, the input constraints for the nested
\raisebox{-3.1874999999999982bp}{\makebox[19.200000000000003bp][l]{\includegraphics[trim=2.4000000000000004 2.4000000000000004 2.4000000000000004 2.4000000000000004]{pict_12.pdf}}} expression of \raisebox{-3.1874999999999982bp}{\makebox[5.231249999999999bp][l]{\includegraphics[trim=2.4000000000000004 2.4000000000000004 2.4000000000000004 2.4000000000000004]{pict_22.pdf}}} indicate that \raisebox{-3.1874999999999982bp}{\makebox[5.231249999999999bp][l]{\includegraphics[trim=2.4000000000000004 2.4000000000000004 2.4000000000000004 2.4000000000000004]{pict_22.pdf}}} ignores the result of
the application \raisebox{-3.1874999999999982bp}{\makebox[14.293750000000001bp][l]{\includegraphics[trim=2.4000000000000004 2.4000000000000004 2.4000000000000004 2.4000000000000004]{pict_60.pdf}}} and goes straight to the \raisebox{-3.1874999999999982bp}{\makebox[19.200000000000003bp][l]{\includegraphics[trim=2.4000000000000004 2.4000000000000004 2.4000000000000004 2.4000000000000004]{pict_50.pdf}}} branch
with label  \raisebox{-2.3617187499999996bp}{\makebox[5.60546875bp][l]{\includegraphics[trim=2.4000000000000004 2.4000000000000004 2.4000000000000004 2.4000000000000004]{pict_62.pdf}}}. Hence, as above, the concolic tester can tweak
the shape of \raisebox{-3.1874999999999982bp}{\makebox[5.231249999999999bp][l]{\includegraphics[trim=2.4000000000000004 2.4000000000000004 2.4000000000000004 2.4000000000000004]{pict_22.pdf}}} so that its nested \raisebox{-3.1874999999999982bp}{\makebox[19.200000000000003bp][l]{\includegraphics[trim=2.4000000000000004 2.4000000000000004 2.4000000000000004 2.4000000000000004]{pict_12.pdf}}} expression inspects the
result of \raisebox{-3.1874999999999982bp}{\makebox[14.293750000000001bp][l]{\includegraphics[trim=2.4000000000000004 2.4000000000000004 2.4000000000000004 2.4000000000000004]{pict_60.pdf}}}.  The test constraint
\raisebox{-3.2890624999999982bp}{\makebox[72.9171875bp][l]{\includegraphics[trim=2.4000000000000004 2.4000000000000004 2.4000000000000004 2.4000000000000004]{pict_63.pdf}}}
shows one way to do so; the concolic tester can add a new clause to the nested \raisebox{-3.1874999999999982bp}{\makebox[19.200000000000003bp][l]{\includegraphics[trim=2.4000000000000004 2.4000000000000004 2.4000000000000004 2.4000000000000004]{pict_12.pdf}}} expression
that tests if \raisebox{-3.1874999999999982bp}{\makebox[18.882031250000004bp][l]{\includegraphics[trim=2.4000000000000004 2.4000000000000004 2.4000000000000004 2.4000000000000004]{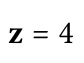}}} and forces the evaluation to avoid the \raisebox{-3.1874999999999982bp}{\makebox[19.200000000000003bp][l]{\includegraphics[trim=2.4000000000000004 2.4000000000000004 2.4000000000000004 2.4000000000000004]{pict_50.pdf}}}
branch.  As before, the concolic tester has a
number of choices for the new clause of the nested \raisebox{-3.1874999999999982bp}{\makebox[19.200000000000003bp][l]{\includegraphics[trim=2.4000000000000004 2.4000000000000004 2.4000000000000004 2.4000000000000004]{pict_12.pdf}}} expression
and eventually it may have to explore them all. This time we opt to
continue the discussion by having the new clause
return a random number that we represent as yet a new input
\raisebox{-2.3617187499999996bp}{\makebox[8.47578125bp][l]{\includegraphics[trim=2.4000000000000004 2.4000000000000004 2.4000000000000004 2.4000000000000004]{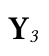}}}:

\begin{SCentered}\raisebox{-2.757812499999991bp}{\makebox[132.26328125bp][l]{\includegraphics[trim=2.4000000000000004 2.4000000000000004 2.4000000000000004 2.4000000000000004]{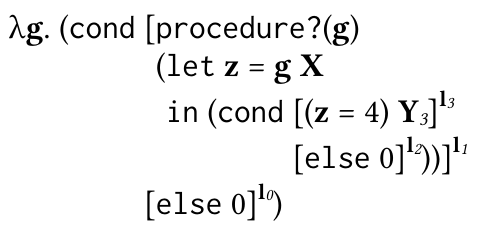}}}\end{SCentered}

Assuming that \raisebox{-3.1874999999999982bp}{\makebox[6.8921875bp][l]{\includegraphics[trim=2.4000000000000004 2.4000000000000004 2.4000000000000004 2.4000000000000004]{pict_2.pdf}}} is equal to \raisebox{-3.1874999999999982bp}{\makebox[4.46328125bp][l]{\includegraphics[trim=2.4000000000000004 2.4000000000000004 2.4000000000000004 2.4000000000000004]{pict_34.pdf}}}, as above, and \raisebox{-2.3617187499999996bp}{\makebox[8.47578125bp][l]{\includegraphics[trim=2.4000000000000004 2.4000000000000004 2.4000000000000004 2.4000000000000004]{pict_68.pdf}}} is randomly
set to  \raisebox{-3.1874999999999982bp}{\makebox[4.46328125bp][l]{\includegraphics[trim=2.4000000000000004 2.4000000000000004 2.4000000000000004 2.4000000000000004]{pict_7.pdf}}},  we obtain this third list of path constraints from the evaluation of the
example:

\begin{SCentered}\raisebox{-2.925781249999991bp}{\makebox[276.12500000000006bp][l]{\includegraphics[trim=2.4000000000000004 2.4000000000000004 2.4000000000000004 2.4000000000000004]{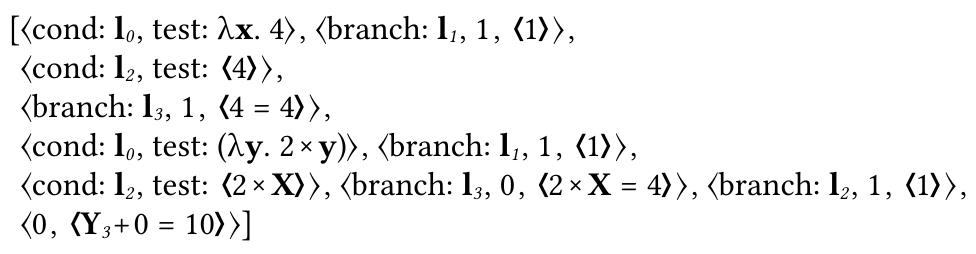}}}\end{SCentered}

\noindent The list is roughly the same as the one from the previous input
except for three constraints. The first difference is that
the fourth constraint is a test constraint that indicates that the nested \raisebox{-3.1874999999999982bp}{\makebox[19.200000000000003bp][l]{\includegraphics[trim=2.4000000000000004 2.4000000000000004 2.4000000000000004 2.4000000000000004]{pict_12.pdf}}} expressions of \raisebox{-3.1874999999999982bp}{\makebox[5.231249999999999bp][l]{\includegraphics[trim=2.4000000000000004 2.4000000000000004 2.4000000000000004 2.4000000000000004]{pict_22.pdf}}}
follows branch \raisebox{-2.3617187499999996bp}{\makebox[5.60546875bp][l]{\includegraphics[trim=2.4000000000000004 2.4000000000000004 2.4000000000000004 2.4000000000000004]{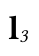}}} for the first application of \raisebox{-3.1874999999999982bp}{\makebox[5.231249999999999bp][l]{\includegraphics[trim=2.4000000000000004 2.4000000000000004 2.4000000000000004 2.4000000000000004]{pict_22.pdf}}}. Next, the
second to last constraint is a new
branch constraint that indicates that the nested \raisebox{-3.1874999999999982bp}{\makebox[19.200000000000003bp][l]{\includegraphics[trim=2.4000000000000004 2.4000000000000004 2.4000000000000004 2.4000000000000004]{pict_12.pdf}}} expressions of \raisebox{-3.1874999999999982bp}{\makebox[5.231249999999999bp][l]{\includegraphics[trim=2.4000000000000004 2.4000000000000004 2.4000000000000004 2.4000000000000004]{pict_22.pdf}}}
fails to follow branch \raisebox{-2.3617187499999996bp}{\makebox[5.60546875bp][l]{\includegraphics[trim=2.4000000000000004 2.4000000000000004 2.4000000000000004 2.4000000000000004]{pict_62.pdf}}} for the first application of \raisebox{-3.1874999999999982bp}{\makebox[5.231249999999999bp][l]{\includegraphics[trim=2.4000000000000004 2.4000000000000004 2.4000000000000004 2.4000000000000004]{pict_22.pdf}}}, since
\raisebox{-3.1874999999999982bp}{\makebox[18.235156250000003bp][l]{\includegraphics[trim=2.4000000000000004 2.4000000000000004 2.4000000000000004 2.4000000000000004]{pict_10.pdf}}} is not equal to \raisebox{-3.1874999999999982bp}{\makebox[4.46328125bp][l]{\includegraphics[trim=2.4000000000000004 2.4000000000000004 2.4000000000000004 2.4000000000000004]{pict_61.pdf}}}. Finally, the last constraint records that the
the test of the \raisebox{-3.1874999999999982bp}{\makebox[19.200000000000003bp][l]{\includegraphics[trim=2.4000000000000004 2.4000000000000004 2.4000000000000004 2.4000000000000004]{pict_12.pdf}}} expression of the example still evaluates to
\raisebox{-3.1874999999999982bp}{\makebox[4.46328125bp][l]{\includegraphics[trim=2.4000000000000004 2.4000000000000004 2.4000000000000004 2.4000000000000004]{pict_7.pdf}}} but that the result depends on the input \raisebox{-2.3617187499999996bp}{\makebox[8.47578125bp][l]{\includegraphics[trim=2.4000000000000004 2.4000000000000004 2.4000000000000004 2.4000000000000004]{pict_68.pdf}}},
specifically that \raisebox{-2.3617187499999996bp}{\makebox[38.825bp][l]{\includegraphics[trim=2.4000000000000004 2.4000000000000004 2.4000000000000004 2.4000000000000004]{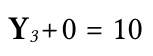}}}.

At this point, the concolic tester could focus on
the first{-}order constraint
\raisebox{-2.3617187499999996bp}{\makebox[61.059375bp][l]{\includegraphics[trim=2.4000000000000004 2.4000000000000004 2.4000000000000004 2.4000000000000004]{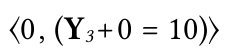}}} and ask the SMT solver to
come up with a value for \raisebox{-2.3617187499999996bp}{\makebox[8.47578125bp][l]{\includegraphics[trim=2.4000000000000004 2.4000000000000004 2.4000000000000004 2.4000000000000004]{pict_68.pdf}}} that flips the test of the
\raisebox{-3.1874999999999982bp}{\makebox[19.200000000000003bp][l]{\includegraphics[trim=2.4000000000000004 2.4000000000000004 2.4000000000000004 2.4000000000000004]{pict_12.pdf}}} expression. Instead, we opt for a different choice
that focuses on the input{-}related constraints and, pushing the
evolution of the shape of the input further.

In particular, the concolic tester can use the test
constraint
\raisebox{-3.2890624999999982bp}{\makebox[86.68906249999999bp][l]{\includegraphics[trim=2.4000000000000004 2.4000000000000004 2.4000000000000004 2.4000000000000004]{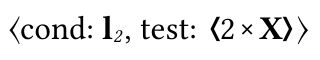}}}
to further refine the shape of \raisebox{-3.1874999999999982bp}{\makebox[5.231249999999999bp][l]{\includegraphics[trim=2.4000000000000004 2.4000000000000004 2.4000000000000004 2.4000000000000004]{pict_22.pdf}}}. This constraint is part of the group
of constraints that correspond to the second application of \raisebox{-3.1874999999999982bp}{\makebox[5.231249999999999bp][l]{\includegraphics[trim=2.4000000000000004 2.4000000000000004 2.4000000000000004 2.4000000000000004]{pict_22.pdf}}} which
imply that for the second application of \raisebox{-3.1874999999999982bp}{\makebox[5.231249999999999bp][l]{\includegraphics[trim=2.4000000000000004 2.4000000000000004 2.4000000000000004 2.4000000000000004]{pict_22.pdf}}}, evaluation again follows the
\raisebox{-3.1874999999999982bp}{\makebox[19.200000000000003bp][l]{\includegraphics[trim=2.4000000000000004 2.4000000000000004 2.4000000000000004 2.4000000000000004]{pict_50.pdf}}} branch of the nested \raisebox{-3.1874999999999982bp}{\makebox[19.200000000000003bp][l]{\includegraphics[trim=2.4000000000000004 2.4000000000000004 2.4000000000000004 2.4000000000000004]{pict_12.pdf}}} expression of the input. Thus the concolic tester can extend
the nested \raisebox{-3.1874999999999982bp}{\makebox[19.200000000000003bp][l]{\includegraphics[trim=2.4000000000000004 2.4000000000000004 2.4000000000000004 2.4000000000000004]{pict_12.pdf}}} expression
with a new branch with the test \raisebox{-3.1874999999999982bp}{\makebox[32.653906250000006bp][l]{\includegraphics[trim=2.4000000000000004 2.4000000000000004 2.4000000000000004 2.4000000000000004]{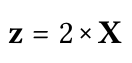}}}:

\noindent \begin{SCentered}\raisebox{-2.4523437499999883bp}{\makebox[142.10234375bp][l]{\includegraphics[trim=2.4000000000000004 2.4000000000000004 2.4000000000000004 2.4000000000000004]{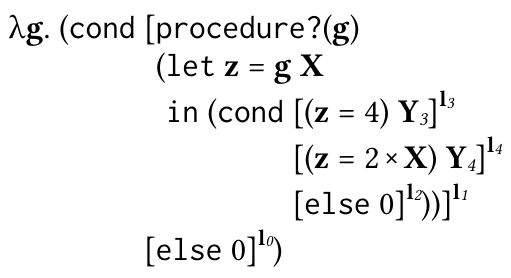}}}\end{SCentered}

\noindent where both \raisebox{-2.3617187499999996bp}{\makebox[8.47578125bp][l]{\includegraphics[trim=2.4000000000000004 2.4000000000000004 2.4000000000000004 2.4000000000000004]{pict_68.pdf}}} and \raisebox{-2.3617187499999996bp}{\makebox[8.47578125bp][l]{\includegraphics[trim=2.4000000000000004 2.4000000000000004 2.4000000000000004 2.4000000000000004]{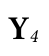}}} are the constant number \raisebox{-3.1874999999999982bp}{\makebox[4.46328125bp][l]{\includegraphics[trim=2.4000000000000004 2.4000000000000004 2.4000000000000004 2.4000000000000004]{pict_7.pdf}}}.

Note that this step of input evolution is that the test of the new branch uses
the expression trace \raisebox{-3.1874999999999982bp}{\makebox[18.235156250000003bp][l]{\includegraphics[trim=2.4000000000000004 2.4000000000000004 2.4000000000000004 2.4000000000000004]{pict_10.pdf}}} instead of an actual number.
This is necessary because this branch should handle the result of
the application \raisebox{-3.1874999999999982bp}{\makebox[14.293750000000001bp][l]{\includegraphics[trim=2.4000000000000004 2.4000000000000004 2.4000000000000004 2.4000000000000004]{pict_60.pdf}}} when \raisebox{-3.1874999999999982bp}{\makebox[5.0015625bp][l]{\includegraphics[trim=2.4000000000000004 2.4000000000000004 2.4000000000000004 2.4000000000000004]{pict_31.pdf}}} is \raisebox{-3.1874999999999982bp}{\makebox[31.664062500000007bp][l]{\includegraphics[trim=2.4000000000000004 2.4000000000000004 2.4000000000000004 2.4000000000000004]{pict_32.pdf}}} for all
values of \raisebox{-3.1874999999999982bp}{\makebox[6.8921875bp][l]{\includegraphics[trim=2.4000000000000004 2.4000000000000004 2.4000000000000004 2.4000000000000004]{pict_2.pdf}}} and not just when \raisebox{-3.1874999999999982bp}{\makebox[6.8921875bp][l]{\includegraphics[trim=2.4000000000000004 2.4000000000000004 2.4000000000000004 2.4000000000000004]{pict_2.pdf}}} is \raisebox{-3.1874999999999982bp}{\makebox[4.46328125bp][l]{\includegraphics[trim=2.4000000000000004 2.4000000000000004 2.4000000000000004 2.4000000000000004]{pict_34.pdf}}}.

At the same time, this evolution establishes a connection between the input \raisebox{-3.1874999999999982bp}{\makebox[6.8921875bp][l]{\includegraphics[trim=2.4000000000000004 2.4000000000000004 2.4000000000000004 2.4000000000000004]{pict_2.pdf}}} and the
test of the \raisebox{-3.1874999999999982bp}{\makebox[19.200000000000003bp][l]{\includegraphics[trim=2.4000000000000004 2.4000000000000004 2.4000000000000004 2.4000000000000004]{pict_12.pdf}}} expression through the new input \raisebox{-2.3617187499999996bp}{\makebox[8.47578125bp][l]{\includegraphics[trim=2.4000000000000004 2.4000000000000004 2.4000000000000004 2.4000000000000004]{pict_77.pdf}}}, which
has important repercussions for the concolic loop.
Specifically, running the example with the new input
produces the list of path constraints:

\noindent \begin{SCentered}\raisebox{-3.1984374999999954bp}{\makebox[318.21171875000005bp][l]{\includegraphics[trim=2.4000000000000004 2.4000000000000004 2.4000000000000004 2.4000000000000004]{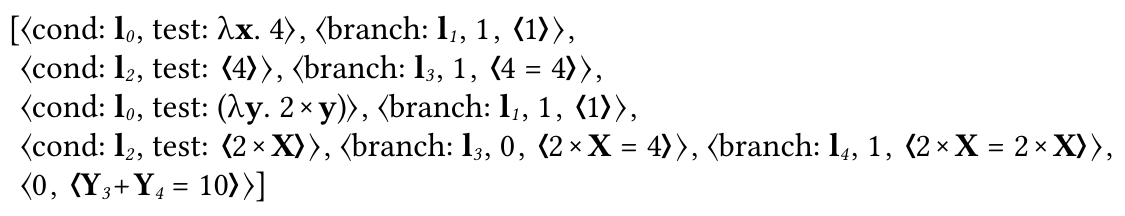}}}\end{SCentered}

\noindent Hence, the concolic tester has two ways (\raisebox{-2.3617187499999996bp}{\makebox[8.47578125bp][l]{\includegraphics[trim=2.4000000000000004 2.4000000000000004 2.4000000000000004 2.4000000000000004]{pict_68.pdf}}} and
\raisebox{-2.3617187499999996bp}{\makebox[8.47578125bp][l]{\includegraphics[trim=2.4000000000000004 2.4000000000000004 2.4000000000000004 2.4000000000000004]{pict_77.pdf}}}) to affect the \raisebox{-3.1874999999999982bp}{\makebox[19.200000000000003bp][l]{\includegraphics[trim=2.4000000000000004 2.4000000000000004 2.4000000000000004 2.4000000000000004]{pict_12.pdf}}} expression in the example.
It can ask the SMT solver to come
up with values for \raisebox{-2.3617187499999996bp}{\makebox[8.47578125bp][l]{\includegraphics[trim=2.4000000000000004 2.4000000000000004 2.4000000000000004 2.4000000000000004]{pict_68.pdf}}}, \raisebox{-2.3617187499999996bp}{\makebox[8.47578125bp][l]{\includegraphics[trim=2.4000000000000004 2.4000000000000004 2.4000000000000004 2.4000000000000004]{pict_77.pdf}}} and \raisebox{-3.1874999999999982bp}{\makebox[6.8921875bp][l]{\includegraphics[trim=2.4000000000000004 2.4000000000000004 2.4000000000000004 2.4000000000000004]{pict_2.pdf}}}, using the negation of the only first{-}order
constraint, i.e., \raisebox{-2.3617187499999996bp}{\makebox[42.837500000000006bp][l]{\includegraphics[trim=2.4000000000000004 2.4000000000000004 2.4000000000000004 2.4000000000000004]{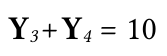}}}. However, due to the
second to last branch constraint it has to also
constrain the SMT solver, ensuring that \raisebox{-3.1874999999999982bp}{\makebox[38.49843750000001bp][l]{\includegraphics[trim=2.4000000000000004 2.4000000000000004 2.4000000000000004 2.4000000000000004]{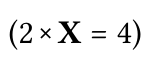}}} is false. Unfortunately,
this constraint makes it impossible for the SMT solver to derive \raisebox{-3.1874999999999982bp}{\makebox[6.8921875bp][l]{\includegraphics[trim=2.4000000000000004 2.4000000000000004 2.4000000000000004 2.4000000000000004]{pict_2.pdf}}}\raisebox{-3.1874999999999982bp}{\makebox[4.800000000000001bp][l]{\includegraphics[trim=2.4000000000000004 2.4000000000000004 2.4000000000000004 2.4000000000000004]{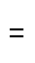}}}\raisebox{-3.1874999999999982bp}{\makebox[4.46328125bp][l]{\includegraphics[trim=2.4000000000000004 2.4000000000000004 2.4000000000000004 2.4000000000000004]{pict_37.pdf}}} which is necessary to
trigger the bug.  In fact, exactly for this reason, this shape of input
is a dead{-}end.

As a result, the
concolic server is better off backtracking to the second input from above:

\noindent \begin{SCentered}\raisebox{-2.757812499999991bp}{\makebox[132.26328125bp][l]{\includegraphics[trim=2.4000000000000004 2.4000000000000004 2.4000000000000004 2.4000000000000004]{pict_69.pdf}}}\end{SCentered}

\noindent where \raisebox{-3.1874999999999982bp}{\makebox[6.8921875bp][l]{\includegraphics[trim=2.4000000000000004 2.4000000000000004 2.4000000000000004 2.4000000000000004]{pict_2.pdf}}} and \raisebox{-2.3617187499999996bp}{\makebox[8.47578125bp][l]{\includegraphics[trim=2.4000000000000004 2.4000000000000004 2.4000000000000004 2.4000000000000004]{pict_68.pdf}}} are equal to \raisebox{-3.1874999999999982bp}{\makebox[4.46328125bp][l]{\includegraphics[trim=2.4000000000000004 2.4000000000000004 2.4000000000000004 2.4000000000000004]{pict_34.pdf}}} and \raisebox{-3.1874999999999982bp}{\makebox[4.46328125bp][l]{\includegraphics[trim=2.4000000000000004 2.4000000000000004 2.4000000000000004 2.4000000000000004]{pict_7.pdf}}} respectively.
With this input, the evaluation of the example produces these path constraints (duplicated
from earlier):

\noindent \begin{SCentered}\raisebox{-3.1984374999999954bp}{\makebox[276.12500000000006bp][l]{\includegraphics[trim=2.4000000000000004 2.4000000000000004 2.4000000000000004 2.4000000000000004]{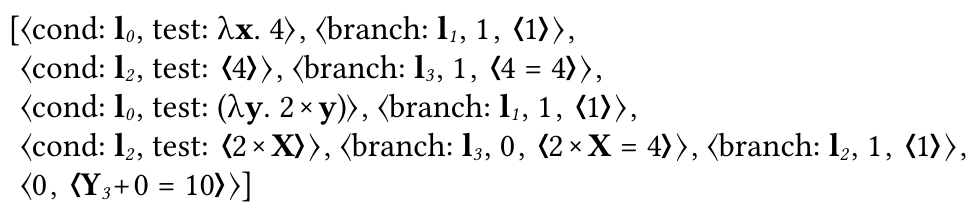}}}\end{SCentered}

\noindent Thus, the concolic tester can ask the SMT solver for help with determining
a new value for \raisebox{-2.3617187499999996bp}{\makebox[8.47578125bp][l]{\includegraphics[trim=2.4000000000000004 2.4000000000000004 2.4000000000000004 2.4000000000000004]{pict_68.pdf}}}, the input it has a first{-}order
constraint for. Subsequently the SMT solver derives that when \raisebox{-2.3617187499999996bp}{\makebox[8.47578125bp][l]{\includegraphics[trim=2.4000000000000004 2.4000000000000004 2.4000000000000004 2.4000000000000004]{pict_68.pdf}}} is
\raisebox{-3.1874999999999982bp}{\makebox[8.9265625bp][l]{\includegraphics[trim=2.4000000000000004 2.4000000000000004 2.4000000000000004 2.4000000000000004]{pict_35.pdf}}}, the result of \raisebox{-2.3617187499999996bp}{\makebox[38.825bp][l]{\includegraphics[trim=2.4000000000000004 2.4000000000000004 2.4000000000000004 2.4000000000000004]{pict_72.pdf}}} flips.

With a new value for \raisebox{-2.3617187499999996bp}{\makebox[8.47578125bp][l]{\includegraphics[trim=2.4000000000000004 2.4000000000000004 2.4000000000000004 2.4000000000000004]{pict_68.pdf}}} in hand and the same shape for \raisebox{-3.1874999999999982bp}{\makebox[5.231249999999999bp][l]{\includegraphics[trim=2.4000000000000004 2.4000000000000004 2.4000000000000004 2.4000000000000004]{pict_22.pdf}}},
the concolic tester evaluates the example and, as usual, records a new
list of path constraints:

\noindent \begin{SCentered}\raisebox{-2.7093749999999996bp}{\makebox[293.42578125000006bp][l]{\includegraphics[trim=2.4000000000000004 2.4000000000000004 2.4000000000000004 2.4000000000000004]{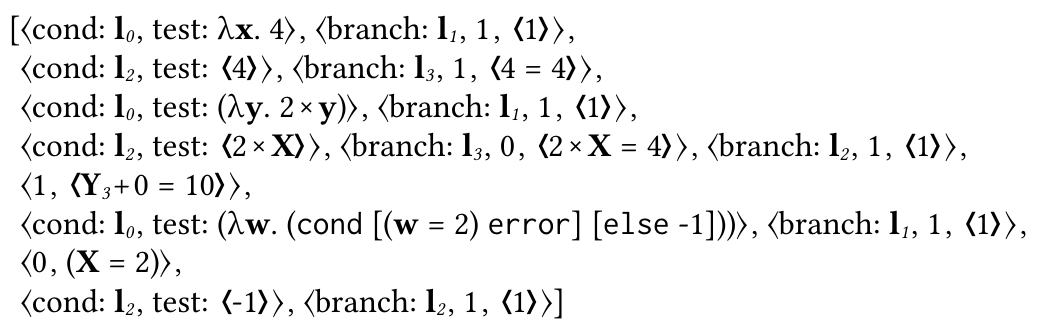}}}\end{SCentered}

\noindent Since the value of \raisebox{-2.3617187499999996bp}{\makebox[8.47578125bp][l]{\includegraphics[trim=2.4000000000000004 2.4000000000000004 2.4000000000000004 2.4000000000000004]{pict_68.pdf}}} makes the test of the \raisebox{-3.1874999999999982bp}{\makebox[19.200000000000003bp][l]{\includegraphics[trim=2.4000000000000004 2.4000000000000004 2.4000000000000004 2.4000000000000004]{pict_12.pdf}}} expression
true, the evaluation reaches the third application of \raisebox{-3.1874999999999982bp}{\makebox[5.231249999999999bp][l]{\includegraphics[trim=2.4000000000000004 2.4000000000000004 2.4000000000000004 2.4000000000000004]{pict_22.pdf}}}
and through that the test of the \raisebox{-3.1874999999999982bp}{\makebox[19.200000000000003bp][l]{\includegraphics[trim=2.4000000000000004 2.4000000000000004 2.4000000000000004 2.4000000000000004]{pict_12.pdf}}} expression of \raisebox{-3.1874999999999982bp}{\makebox[149.60625bp][l]{\includegraphics[trim=2.4000000000000004 2.4000000000000004 2.4000000000000004 2.4000000000000004]{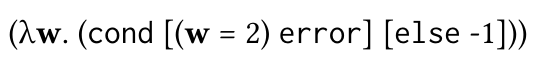}}}.
This test contributes the new first{-}order constraint \raisebox{-3.1874999999999982bp}{\makebox[43.669531250000006bp][l]{\includegraphics[trim=2.4000000000000004 2.4000000000000004 2.4000000000000004 2.4000000000000004]{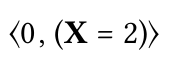}}}.
Unfortunately, if the concolic tester asks the
the SMT solver to tweak the input so that \raisebox{-3.1874999999999982bp}{\makebox[6.8921875bp][l]{\includegraphics[trim=2.4000000000000004 2.4000000000000004 2.4000000000000004 2.4000000000000004]{pict_2.pdf}}} is equal to \raisebox{-3.1874999999999982bp}{\makebox[4.46328125bp][l]{\includegraphics[trim=2.4000000000000004 2.4000000000000004 2.4000000000000004 2.4000000000000004]{pict_37.pdf}}},
the solver reports that the formula is unsatisfiable since
  \raisebox{-3.1874999999999982bp}{\makebox[21.435156250000006bp][l]{\includegraphics[trim=2.4000000000000004 2.4000000000000004 2.4000000000000004 2.4000000000000004]{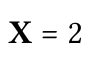}}} from the new first{-}order constraint clashes with
  the fact that \raisebox{-3.1874999999999982bp}{\makebox[38.49843750000001bp][l]{\includegraphics[trim=2.4000000000000004 2.4000000000000004 2.4000000000000004 2.4000000000000004]{pict_80.pdf}}} is false from the second to last branch constraint.

To make progress here, the concolic tester can take a step back and
and force the evaluation
of the second application of \raisebox{-3.1874999999999982bp}{\makebox[5.231249999999999bp][l]{\includegraphics[trim=2.4000000000000004 2.4000000000000004 2.4000000000000004 2.4000000000000004]{pict_22.pdf}}} in the example to follow branch
\raisebox{-2.3617187499999996bp}{\makebox[5.60546875bp][l]{\includegraphics[trim=2.4000000000000004 2.4000000000000004 2.4000000000000004 2.4000000000000004]{pict_71.pdf}}} instead of the \raisebox{-3.1874999999999982bp}{\makebox[19.200000000000003bp][l]{\includegraphics[trim=2.4000000000000004 2.4000000000000004 2.4000000000000004 2.4000000000000004]{pict_50.pdf}}} branch. To do so, the tester truncates its list of path constraints and
inserts at the end a new branch constraint as if the test of branch
\raisebox{-2.3617187499999996bp}{\makebox[5.60546875bp][l]{\includegraphics[trim=2.4000000000000004 2.4000000000000004 2.4000000000000004 2.4000000000000004]{pict_71.pdf}}} would have succeeded for the second application of \raisebox{-3.1874999999999982bp}{\makebox[5.231249999999999bp][l]{\includegraphics[trim=2.4000000000000004 2.4000000000000004 2.4000000000000004 2.4000000000000004]{pict_22.pdf}}}. This
targeted branch constraint modification produces:

\noindent \begin{SCentered}\raisebox{-2.4710937499999996bp}{\makebox[199.234375bp][l]{\includegraphics[trim=2.4000000000000004 2.4000000000000004 2.4000000000000004 2.4000000000000004]{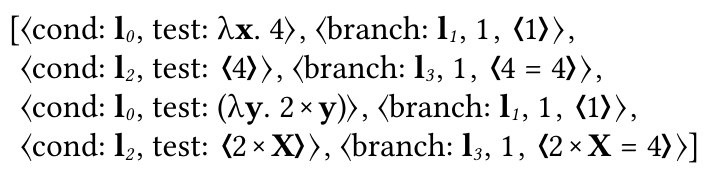}}}\end{SCentered}

\noindent The concolic tester translates the modified constraints into a formula,
asks the SMT solver for a model and the solver replies that \raisebox{-3.1874999999999982bp}{\makebox[6.8921875bp][l]{\includegraphics[trim=2.4000000000000004 2.4000000000000004 2.4000000000000004 2.4000000000000004]{pict_2.pdf}}} is equal to
\raisebox{-3.1874999999999982bp}{\makebox[4.46328125bp][l]{\includegraphics[trim=2.4000000000000004 2.4000000000000004 2.4000000000000004 2.4000000000000004]{pict_37.pdf}}}.

Consequently, the concolic tester tries again the previous input

\noindent \begin{SCentered}\raisebox{-2.757812499999991bp}{\makebox[132.26328125bp][l]{\includegraphics[trim=2.4000000000000004 2.4000000000000004 2.4000000000000004 2.4000000000000004]{pict_69.pdf}}}\end{SCentered}

\noindent but where \raisebox{-3.1874999999999982bp}{\makebox[6.8921875bp][l]{\includegraphics[trim=2.4000000000000004 2.4000000000000004 2.4000000000000004 2.4000000000000004]{pict_2.pdf}}} is equal to \raisebox{-3.1874999999999982bp}{\makebox[4.46328125bp][l]{\includegraphics[trim=2.4000000000000004 2.4000000000000004 2.4000000000000004 2.4000000000000004]{pict_37.pdf}}} and \raisebox{-2.3617187499999996bp}{\makebox[8.47578125bp][l]{\includegraphics[trim=2.4000000000000004 2.4000000000000004 2.4000000000000004 2.4000000000000004]{pict_68.pdf}}} remains equal to \raisebox{-3.1874999999999982bp}{\makebox[4.46328125bp][l]{\includegraphics[trim=2.4000000000000004 2.4000000000000004 2.4000000000000004 2.4000000000000004]{pict_7.pdf}}}, as before.
The resulting path constraints are:

\noindent \begin{SCentered}\raisebox{-3.1984374999999954bp}{\makebox[198.72890625bp][l]{\includegraphics[trim=2.4000000000000004 2.4000000000000004 2.4000000000000004 2.4000000000000004]{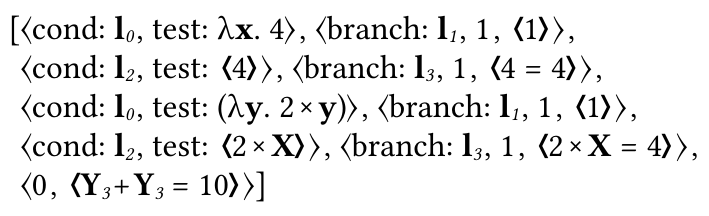}}}\end{SCentered}

With the next iteration of the concolic loop, the concolic tester can
perform one more SMT query. The query adjusts \raisebox{-2.3617187499999996bp}{\makebox[8.47578125bp][l]{\includegraphics[trim=2.4000000000000004 2.4000000000000004 2.4000000000000004 2.4000000000000004]{pict_68.pdf}}} to \raisebox{-3.1874999999999982bp}{\makebox[4.46328125bp][l]{\includegraphics[trim=2.4000000000000004 2.4000000000000004 2.4000000000000004 2.4000000000000004]{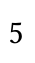}}} to achieve
\raisebox{-2.3617187499999996bp}{\makebox[42.837500000000006bp][l]{\includegraphics[trim=2.4000000000000004 2.4000000000000004 2.4000000000000004 2.4000000000000004]{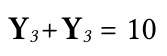}}}. As
a result the evaluation reaches again the \raisebox{-3.1874999999999982bp}{\makebox[19.200000000000003bp][l]{\includegraphics[trim=2.4000000000000004 2.4000000000000004 2.4000000000000004 2.4000000000000004]{pict_12.pdf}}} expression in \raisebox{-3.1874999999999982bp}{\makebox[149.60625bp][l]{\includegraphics[trim=2.4000000000000004 2.4000000000000004 2.4000000000000004 2.4000000000000004]{pict_84.pdf}}} and
since \raisebox{-3.1874999999999982bp}{\makebox[6.8921875bp][l]{\includegraphics[trim=2.4000000000000004 2.4000000000000004 2.4000000000000004 2.4000000000000004]{pict_2.pdf}}} is already equal to \raisebox{-3.1874999999999982bp}{\makebox[4.46328125bp][l]{\includegraphics[trim=2.4000000000000004 2.4000000000000004 2.4000000000000004 2.4000000000000004]{pict_37.pdf}}},
the input finally triggers the error.

To sum up, the example shows how our higher{-}order concolic tester
evaluates a program under test and records both first{-}order and
input{-}related constraints. It feeds the first{-}order ones to an
SMT solver the same way a first{-}order concolic tester uses the
constraints it records to explore the control{-}flow graph of the program under
test. The input{-}related ones, which are a
distinguishing feature of our design, help the concolic tester to evolve the
input iteratively and
introduce further ways
the input can influence the evaluation of the user program. The
subsequent section makes these insights precise with a formal model.

\paragraph{Remark on Completeness}
Before though we conclude this section, we briefly discuss
a subtle point about the inputs that our concolic tester
produces through evolution. The astute reader may have observed that
\raisebox{-3.1874999999999982bp}{\makebox[79.78046875bp][l]{\includegraphics[trim=2.4000000000000004 2.4000000000000004 2.4000000000000004 2.4000000000000004]{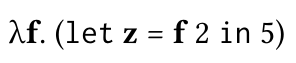}}}, a simpler input than the one our
concolic tester generates, can also reveal the error in our running
example. The ability of the concolic tester to produce inputs
like the one for our running example {---} that call their arguments (maybe more than once),
inspect the result of
the call(s) and, then based on that, decide on the value to return {---} is
critical and, without it, we would not be able to explore all possible behavior
of the user program.
For instance,
consider  the following variant of our example:

\noindent \begin{SCentered}\raisebox{-2.786718749999997bp}{\makebox[155.3140625bp][l]{\includegraphics[trim=2.4000000000000004 2.4000000000000004 2.4000000000000004 2.4000000000000004]{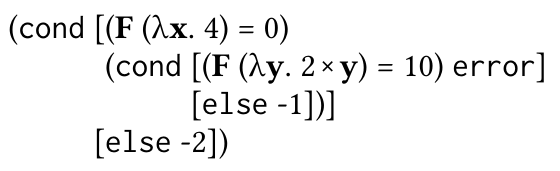}}}\end{SCentered}

\noindent Here, the input cannot trigger the error unless it can distinguish the
two call sites of \raisebox{-3.1874999999999982bp}{\makebox[5.231249999999999bp][l]{\includegraphics[trim=2.4000000000000004 2.4000000000000004 2.4000000000000004 2.4000000000000004]{pict_22.pdf}}} and thus it must be able to distinguish the two
functions that are passed to \raisebox{-3.1874999999999982bp}{\makebox[5.231249999999999bp][l]{\includegraphics[trim=2.4000000000000004 2.4000000000000004 2.4000000000000004 2.4000000000000004]{pict_22.pdf}}}. And, clearly, the only way to distinguish
them is to call them. Indeed
this is exactly what concolic{-}generated input does:

\noindent \begin{SCentered}\raisebox{-2.4523437499999883bp}{\makebox[142.10234375bp][l]{\includegraphics[trim=2.4000000000000004 2.4000000000000004 2.4000000000000004 2.4000000000000004]{pict_76.pdf}}}\end{SCentered}

Put differently, it is not sufficient that a higher{-}order concolic
tester call user program functions simply to discover errors in them.
It must also be able to call given functions in order to force the user
program to expose all of its possible behavior.

\sectionNewpage

\Ssection{Formalizing Higher{-}Order Concolic Testing}{Formalizing Higher{-}Order Concolic Testing}\label{t:x28part_x22secx3amodelx22x29}

The formal model of our higher{-}order concolic tester consists of three
pieces: (i) the language of user programs, (ii) a concolic
machine that evaluates a user program for a given input and produces
the corresponding list of path constraints, and (iii) the input evolution
process that uses the list of path constraints to construct a new input for
the user program for the next iteration of the concolic loop.

Figure~\hyperref[t:x28counter_x28x22figurex22_x22figx3amodelx2dsummaryx22x29x29]{\FigureRef{2}{t:x28counter_x28x22figurex22_x22figx3amodelx2dsummaryx22x29x29}} puts these three pieces together and
shows how they form the concolic loop.  First, a user program
\raisebox{-3.1874999999999982bp}{\makebox[4.69375bp][l]{\includegraphics[trim=2.4000000000000004 2.4000000000000004 2.4000000000000004 2.4000000000000004]{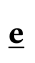}}} goes though an instrumentation step that translates it
into the corresponding concolic program  \raisebox{-3.1874999999999982bp}{\makebox[4.69375bp][l]{\includegraphics[trim=2.4000000000000004 2.4000000000000004 2.4000000000000004 2.4000000000000004]{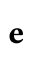}}}.  Instrumentation
requires also the construction of a store \raisebox{-3.1874999999999982bp}{\makebox[5.471875000000001bp][l]{\includegraphics[trim=2.4000000000000004 2.4000000000000004 2.4000000000000004 2.4000000000000004]{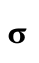}}} that maps each free variable
\raisebox{-3.1874999999999982bp}{\makebox[6.8921875bp][l]{\includegraphics[trim=2.4000000000000004 2.4000000000000004 2.4000000000000004 2.4000000000000004]{pict_2.pdf}}} of \raisebox{-3.1874999999999982bp}{\makebox[4.69375bp][l]{\includegraphics[trim=2.4000000000000004 2.4000000000000004 2.4000000000000004 2.4000000000000004]{pict_93.pdf}}} to a number \raisebox{-3.1874999999999982bp}{\makebox[5.913281250000001bp][l]{\includegraphics[trim=2.4000000000000004 2.4000000000000004 2.4000000000000004 2.4000000000000004]{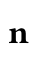}}} or a canonical function \raisebox{-2.3617187499999996bp}{\makebox[44.544531250000006bp][l]{\includegraphics[trim=2.4000000000000004 2.4000000000000004 2.4000000000000004 2.4000000000000004]{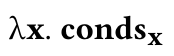}}}.
 In the initial store, concolic variables map either to a random number or
 the default canonical function \raisebox{-3.2562499999999996bp}{\makebox[77.29140625bp][l]{\includegraphics[trim=2.4000000000000004 2.4000000000000004 2.4000000000000004 2.4000000000000004]{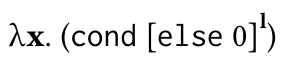}}}.  Put differently the store codifies the
 input for the user program.  Given the store, the concolic machine
 evaluates the result of the instrumentation,  \raisebox{-3.1874999999999982bp}{\makebox[4.69375bp][l]{\includegraphics[trim=2.4000000000000004 2.4000000000000004 2.4000000000000004 2.4000000000000004]{pict_94.pdf}}}, using
three registers: the store \raisebox{-3.1874999999999982bp}{\makebox[5.471875000000001bp][l]{\includegraphics[trim=2.4000000000000004 2.4000000000000004 2.4000000000000004 2.4000000000000004]{pict_95.pdf}}}, the list of
 path constraints \raisebox{-3.1874999999999982bp}{\makebox[5.682812500000001bp][l]{\includegraphics[trim=2.4000000000000004 2.4000000000000004 2.4000000000000004 2.4000000000000004]{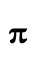}}} (that is initially empty) and the concolic program
 \raisebox{-3.1874999999999982bp}{\makebox[4.69375bp][l]{\includegraphics[trim=2.4000000000000004 2.4000000000000004 2.4000000000000004 2.4000000000000004]{pict_94.pdf}}}. If the result of the evaluation is not an error, the final content
 \raisebox{-3.1874999999999982bp}{\makebox[5.682812500000001bp][l]{\includegraphics[trim=2.4000000000000004 2.4000000000000004 2.4000000000000004 2.4000000000000004]{pict_99.pdf}}} of the middle register  together with the store \raisebox{-3.1874999999999982bp}{\makebox[5.471875000000001bp][l]{\includegraphics[trim=2.4000000000000004 2.4000000000000004 2.4000000000000004 2.4000000000000004]{pict_95.pdf}}} become the
 seed for the evolution of the input.  Specifically, the \raisebox{-3.1874999999999982bp}{\makebox[23.3828125bp][l]{\includegraphics[trim=2.4000000000000004 2.4000000000000004 2.4000000000000004 2.4000000000000004]{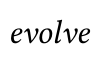}}}
 metafunction uses them to compute a new store  \raisebox{-3.1874999999999982bp}{\makebox[5.471875000000001bp][l]{\includegraphics[trim=2.4000000000000004 2.4000000000000004 2.4000000000000004 2.4000000000000004]{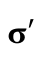}}} and the concolic
 loop proceeds with this new input to the next iteration until it discovers
 an error in the user program.

\begin{Figure}\begin{Centerfigure}\begin{FigureInside}\begin{tikzpicture}
\node[rectangle, draw](source) at (-3.5,0.8) {\begin{minipage}[t]{3.5cm}\begin{center}The User Program\end{center}
\begin{center}\raisebox{-3.1874999999999982bp}{\makebox[4.69375bp][l]{\includegraphics[trim=2.4000000000000004 2.4000000000000004 2.4000000000000004 2.4000000000000004]{pict_93.pdf}}}\end{center}\end{minipage}};
\node[rectangle, draw](store) at (-3.5,-0.8) {\begin{minipage}[t]{5.5cm}\begin{center}The New Input\end{center}
\begin{center}$\raisebox{-3.1874999999999982bp}{\makebox[5.471875000000001bp][l]{\includegraphics[trim=2.4000000000000004 2.4000000000000004 2.4000000000000004 2.4000000000000004]{pict_95.pdf}}} : \raisebox{-3.1874999999999982bp}{\makebox[6.8921875bp][l]{\includegraphics[trim=2.4000000000000004 2.4000000000000004 2.4000000000000004 2.4000000000000004]{pict_2.pdf}}} \longmapsto \raisebox{-3.1874999999999982bp}{\makebox[5.913281250000001bp][l]{\includegraphics[trim=2.4000000000000004 2.4000000000000004 2.4000000000000004 2.4000000000000004]{pict_96.pdf}}} \text{ or } \raisebox{-2.3617187499999996bp}{\makebox[44.544531250000006bp][l]{\includegraphics[trim=2.4000000000000004 2.4000000000000004 2.4000000000000004 2.4000000000000004]{pict_97.pdf}}}$\end{center}\end{minipage}};
\node[rectangle, draw](concolic) at (4,0) {\begin{minipage}[t]{4cm}\begin{center}The Concolic Program\end{center}
\begin{center}\raisebox{-3.1874999999999982bp}{\makebox[4.69375bp][l]{\includegraphics[trim=2.4000000000000004 2.4000000000000004 2.4000000000000004 2.4000000000000004]{pict_94.pdf}}}\end{center}\end{minipage}};
\node[rectangle, draw](pcs) at (0,-3.5) {\begin{minipage}[t]{4.5cm}\begin{center}The List of Path Constraints\end{center}
\begin{center}\raisebox{-3.1874999999999982bp}{\makebox[5.682812500000001bp][l]{\includegraphics[trim=2.4000000000000004 2.4000000000000004 2.4000000000000004 2.4000000000000004]{pict_99.pdf}}}\end{center}\end{minipage}};
\draw[thick] (source) -- (store);
\draw[->,thick] (-3.5,0) -- node[text width=2.5cm,pos=0.72,above]{\begin{minipage}[t]{2.5cm}\begin{center}Instrumentation\end{center}
\begin{center}\raisebox{-3.1874999999999982bp}{\makebox[30.430468750000003bp][l]{\includegraphics[trim=2.4000000000000004 2.4000000000000004 2.4000000000000004 2.4000000000000004]{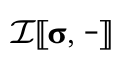}}}\end{center}\end{minipage}} (concolic);

\draw[->,thick] (concolic) -- node[midway,right]{\begin{minipage}[t]{4cm}\begin{center}Concolic Evaluation\end{center}
\begin{center}\raisebox{-2.6179687499999993bp}{\makebox[95.08125bp][l]{\includegraphics[trim=2.4000000000000004 2.4000000000000004 2.4000000000000004 2.4000000000000004]{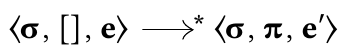}}}\end{center}\end{minipage}}(pcs);
\draw[->,thick] (pcs) -- node[midway,left]{\begin{minipage}[t]{4.75cm}\begin{center}Evolving New Inputs\end{center}
\begin{center}\raisebox{-2.9390624999999986bp}{\makebox[90.08281249999999bp][l]{\includegraphics[trim=2.4000000000000004 2.4000000000000004 2.4000000000000004 2.4000000000000004]{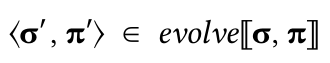}}}\end{center}\end{minipage}}(store);

\end{tikzpicture}\end{FigureInside}\end{Centerfigure}

\Centertext{\Legend{\FigureTarget{\label{t:x28counter_x28x22figurex22_x22figx3amodelx2dsummaryx22x29x29}\textsf{Fig.}~\textsf{2}. }{t:x28counter_x28x22figurex22_x22figx3amodelx2dsummaryx22x29x29}\textsf{The Concolic Loop, Formally}}}\end{Figure}

\SecRefUC{\SectionNumberLink{t:x28part_x22secx3amodelx2dlanguagex22x29}{5.1}}{From User Programs to Concolic Evaluation} details the syntax of user programs,
the syntax of concolic programs, the instrumentation step that
translates the former to the latter and the operation of the concolic
machine.
\SecRefUC{\SectionNumberLink{t:x28part_x22secx3amodelx2devolvex22x29}{5.2}}{Evolving  New Inputs} concludes this section
with a formal description of the evolution of new inputs.

\Ssubsection{From User Programs to Concolic Evaluation}{From User Programs to Concolic Evaluation}\label{t:x28part_x22secx3amodelx2dlanguagex22x29}

\begin{Figure}\begin{Centerfigure}\begin{FigureInside}\raisebox{-2.6742187499999996bp}{\makebox[293.07968750000003bp][l]{\includegraphics[trim=2.4000000000000004 2.4000000000000004 2.4000000000000004 2.4000000000000004]{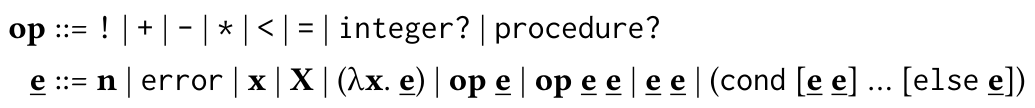}}}\end{FigureInside}\end{Centerfigure}

\Centertext{\Legend{\FigureTarget{\label{t:x28counter_x28x22figurex22_x22figx3asourcex2dlanguagex22x29x29}\textsf{Fig.}~\textsf{3}. }{t:x28counter_x28x22figurex22_x22figx3asourcex2dlanguagex22x29x29}\textsf{The Syntax of User Programs}}}\end{Figure}

The language of user programs is a typical dynamic functional language based on the  call{-}by{-}value
$\lambda${-}calculus. Figure~\hyperref[t:x28counter_x28x22figurex22_x22figx3asourcex2dlanguagex22x29x29]{\FigureRef{3}{t:x28counter_x28x22figurex22_x22figx3asourcex2dlanguagex22x29x29}} collects the constructs of the
 language that include numbers \raisebox{-3.1874999999999982bp}{\makebox[5.913281250000001bp][l]{\includegraphics[trim=2.4000000000000004 2.4000000000000004 2.4000000000000004 2.4000000000000004]{pict_96.pdf}}}, \raisebox{-3.1874999999999982bp}{\makebox[24.0bp][l]{\includegraphics[trim=2.4000000000000004 2.4000000000000004 2.4000000000000004 2.4000000000000004]{pict_4.pdf}}},
 multi{-}way conditional expressions \raisebox{-3.1874999999999982bp}{\makebox[19.200000000000003bp][l]{\includegraphics[trim=2.4000000000000004 2.4000000000000004 2.4000000000000004 2.4000000000000004]{pict_12.pdf}}},
 and concolic variables \raisebox{-3.1874999999999982bp}{\makebox[6.8921875bp][l]{\includegraphics[trim=2.4000000000000004 2.4000000000000004 2.4000000000000004 2.4000000000000004]{pict_2.pdf}}} which as we discuss
 in the previous sections correspond to
 the inputs of a program. For closed user programs, i.e., those without concolic
 or other free variables, we define a standard  call{-}by{-}value reduction
 semantics with reduction relation \raisebox{-3.3617187499999996bp}{\makebox[19.7734375bp][l]{\includegraphics[trim=2.4000000000000004 2.4000000000000004 2.4000000000000004 2.4000000000000004]{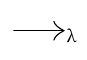}}}. The complete definition of the
 language of the user programs is in
 the supplementary material.

\begin{Figure}\begin{Centerfigure}\begin{FigureInside}\raisebox{-2.6703124999999996bp}{\makebox[332.9312500000001bp][l]{\includegraphics[trim=2.4000000000000004 2.4000000000000004 2.4000000000000004 2.4000000000000004]{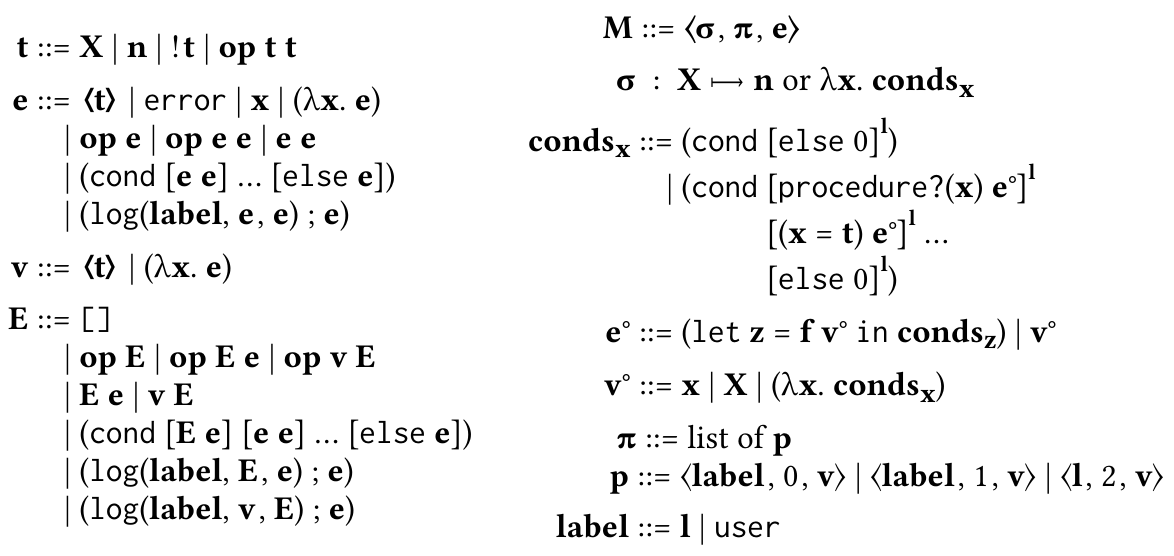}}}\end{FigureInside}\end{Centerfigure}

\Centertext{\Legend{\FigureTarget{\label{t:x28counter_x28x22figurex22_x22figx3amachinex2dlanguagex22x29x29}\textsf{Fig.}~\textsf{4}. }{t:x28counter_x28x22figurex22_x22figx3amachinex2dlanguagex22x29x29}\textsf{The Syntax of the Concolic Language and the Definition of the Concolic Machine}}}\end{Figure}

Figure~\hyperref[t:x28counter_x28x22figurex22_x22figx3amachinex2dlanguagex22x29x29]{\FigureRef{4}{t:x28counter_x28x22figurex22_x22figx3amachinex2dlanguagex22x29x29}}
shows the definition of the concolic language.
The syntax of concolic programs (in the left{-}hand side of the figure)
deviates from that of
user programs in two points: (i)
it combines and generalizes
numbers and concolic variables
into \emph{traced values} \raisebox{-3.2890624999999982bp}{\makebox[11.79609375bp][l]{\includegraphics[trim=2.4000000000000004 2.4000000000000004 2.4000000000000004 2.4000000000000004]{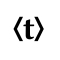}}}
and (ii) it adds an additional log{-}and{-}sequence construct
\raisebox{-2.3617187499999996bp}{\makebox[85.35624999999999bp][l]{\includegraphics[trim=2.4000000000000004 2.4000000000000004 2.4000000000000004 2.4000000000000004]{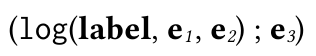}}}.

As we discuss in \ChapRef{\SectionNumberLink{t:x28part_x22secx3abackgroundx22x29}{2}}{A Refresher on First{-}Order Concolic Testing} and \ChapRef{\SectionNumberLink{t:x28part_x22secx3ahowx2dgeneratex22x29}{4}}{Directed Evolution of Canonical Functions}, an
expression trace \raisebox{-3.1874999999999982bp}{\makebox[3.436718749999999bp][l]{\includegraphics[trim=2.4000000000000004 2.4000000000000004 2.4000000000000004 2.4000000000000004]{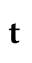}}} describes how  a number depends on the inputs
of the program. That is a number may be directly the value of a concolic
variable \raisebox{-3.1874999999999982bp}{\makebox[6.8921875bp][l]{\includegraphics[trim=2.4000000000000004 2.4000000000000004 2.4000000000000004 2.4000000000000004]{pict_2.pdf}}}, or a constant number \raisebox{-3.1874999999999982bp}{\makebox[5.913281250000001bp][l]{\includegraphics[trim=2.4000000000000004 2.4000000000000004 2.4000000000000004 2.4000000000000004]{pict_96.pdf}}}, or the negation of a trace
\raisebox{-3.1874999999999982bp}{\makebox[7.000781249999999bp][l]{\includegraphics[trim=2.4000000000000004 2.4000000000000004 2.4000000000000004 2.4000000000000004]{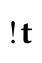}}} (since we use \raisebox{-3.1874999999999982bp}{\makebox[4.46328125bp][l]{\includegraphics[trim=2.4000000000000004 2.4000000000000004 2.4000000000000004 2.4000000000000004]{pict_7.pdf}}} and \raisebox{-3.1874999999999982bp}{\makebox[4.46328125bp][l]{\includegraphics[trim=2.4000000000000004 2.4000000000000004 2.4000000000000004 2.4000000000000004]{pict_8.pdf}}} as booleans), or the result of
an operation \raisebox{-2.3617187499999996bp}{\makebox[27.511718750000007bp][l]{\includegraphics[trim=2.4000000000000004 2.4000000000000004 2.4000000000000004 2.4000000000000004]{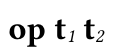}}}.
In a concolic program we represent a number with trace \raisebox{-3.1874999999999982bp}{\makebox[3.436718749999999bp][l]{\includegraphics[trim=2.4000000000000004 2.4000000000000004 2.4000000000000004 2.4000000000000004]{pict_110.pdf}}}  as \raisebox{-3.2890624999999982bp}{\makebox[11.79609375bp][l]{\includegraphics[trim=2.4000000000000004 2.4000000000000004 2.4000000000000004 2.4000000000000004]{pict_108.pdf}}},
and we recompute the actual numerical value when needed. In effect,
traced values allow us to {``}concolicly{''} multiplex the concrete and the symbolic
evaluation of programs.

The log{-}and{-}sequence construct
\raisebox{-2.3617187499999996bp}{\makebox[85.35624999999999bp][l]{\includegraphics[trim=2.4000000000000004 2.4000000000000004 2.4000000000000004 2.4000000000000004]{pict_109.pdf}}}
instructs the concolic machine
to extend  its current list of path constraints with a new constraint.
For the new constraint the machine uses the first three arguments of the
construct and we revisit the details of how it does so
further on. After the  machine adds the new constraint to its list of
path constraints, concolic evaluation continues with  \raisebox{-2.3617187499999996bp}{\makebox[7.1796875bp][l]{\includegraphics[trim=2.4000000000000004 2.4000000000000004 2.4000000000000004 2.4000000000000004]{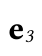}}}.

Before delving into the details of how the instrumentation translates
user programs to concolic ones, we discuss  the concolic machine \raisebox{-3.1874999999999982bp}{\makebox[8.6296875bp][l]{\includegraphics[trim=2.4000000000000004 2.4000000000000004 2.4000000000000004 2.4000000000000004]{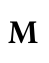}}}.
The right{-}hand side of figure~\hyperref[t:x28counter_x28x22figurex22_x22figx3amachinex2dlanguagex22x29x29]{\FigureRef{4}{t:x28counter_x28x22figurex22_x22figx3amachinex2dlanguagex22x29x29}} defines
the states of the machine as a triplet of a store
\raisebox{-3.1874999999999982bp}{\makebox[5.471875000000001bp][l]{\includegraphics[trim=2.4000000000000004 2.4000000000000004 2.4000000000000004 2.4000000000000004]{pict_95.pdf}}}, a list of path constraints  \raisebox{-3.1874999999999982bp}{\makebox[5.682812500000001bp][l]{\includegraphics[trim=2.4000000000000004 2.4000000000000004 2.4000000000000004 2.4000000000000004]{pict_99.pdf}}} and a concolic program
\raisebox{-3.1874999999999982bp}{\makebox[4.69375bp][l]{\includegraphics[trim=2.4000000000000004 2.4000000000000004 2.4000000000000004 2.4000000000000004]{pict_94.pdf}}}. As we discuss above \raisebox{-3.1874999999999982bp}{\makebox[5.471875000000001bp][l]{\includegraphics[trim=2.4000000000000004 2.4000000000000004 2.4000000000000004 2.4000000000000004]{pict_95.pdf}}} encodes the inputs of the user program
and maps the program{'}s concolic variables to either numbers
\raisebox{-3.1874999999999982bp}{\makebox[5.913281250000001bp][l]{\includegraphics[trim=2.4000000000000004 2.4000000000000004 2.4000000000000004 2.4000000000000004]{pict_96.pdf}}} or concolic functions \raisebox{-2.3617187499999996bp}{\makebox[44.544531250000006bp][l]{\includegraphics[trim=2.4000000000000004 2.4000000000000004 2.4000000000000004 2.4000000000000004]{pict_97.pdf}}}.
The body of the latter, \raisebox{-2.3617187499999996bp}{\makebox[29.550781250000007bp][l]{\includegraphics[trim=2.4000000000000004 2.4000000000000004 2.4000000000000004 2.4000000000000004]{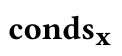}}}, follows the informal description from
  \ChapRef{\SectionNumberLink{t:x28part_x22secx3aapproximatex22x29}{3}}{Canonical Functions Are All We Need} and \ChapRef{\SectionNumberLink{t:x28part_x22secx3ahowx2dgeneratex22x29}{4}}{Directed Evolution of Canonical Functions}.
Formally  \raisebox{-2.3617187499999996bp}{\makebox[29.550781250000007bp][l]{\includegraphics[trim=2.4000000000000004 2.4000000000000004 2.4000000000000004 2.4000000000000004]{pict_115.pdf}}} is a \raisebox{-3.1874999999999982bp}{\makebox[19.200000000000003bp][l]{\includegraphics[trim=2.4000000000000004 2.4000000000000004 2.4000000000000004 2.4000000000000004]{pict_12.pdf}}} expression
that has either just an \raisebox{-3.1874999999999982bp}{\makebox[19.200000000000003bp][l]{\includegraphics[trim=2.4000000000000004 2.4000000000000004 2.4000000000000004 2.4000000000000004]{pict_50.pdf}}} branch that returns the default value \raisebox{-3.1874999999999982bp}{\makebox[4.46328125bp][l]{\includegraphics[trim=2.4000000000000004 2.4000000000000004 2.4000000000000004 2.4000000000000004]{pict_7.pdf}}}
or, in addition to the \raisebox{-3.1874999999999982bp}{\makebox[19.200000000000003bp][l]{\includegraphics[trim=2.4000000000000004 2.4000000000000004 2.4000000000000004 2.4000000000000004]{pict_50.pdf}}} branch it also has branches that inspect
variable \raisebox{-3.1874999999999982bp}{\makebox[5.38515625bp][l]{\includegraphics[trim=2.4000000000000004 2.4000000000000004 2.4000000000000004 2.4000000000000004]{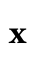}}}. In particular, the latter branches  test whether \raisebox{-3.1874999999999982bp}{\makebox[5.38515625bp][l]{\includegraphics[trim=2.4000000000000004 2.4000000000000004 2.4000000000000004 2.4000000000000004]{pict_116.pdf}}} is
a procedure or whether it is equal to a number. Each branch, including the
\raisebox{-3.1874999999999982bp}{\makebox[19.200000000000003bp][l]{\includegraphics[trim=2.4000000000000004 2.4000000000000004 2.4000000000000004 2.4000000000000004]{pict_50.pdf}}} branch, has a unique identifying label \raisebox{-3.1874999999999982bp}{\makebox[3.1195312499999996bp][l]{\includegraphics[trim=2.4000000000000004 2.4000000000000004 2.4000000000000004 2.4000000000000004]{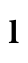}}}. The body of a
branch other than the \raisebox{-3.1874999999999982bp}{\makebox[19.200000000000003bp][l]{\includegraphics[trim=2.4000000000000004 2.4000000000000004 2.4000000000000004 2.4000000000000004]{pict_50.pdf}}} branch is either a variable or a concolic function or a \raisebox{-3.1874999999999982bp}{\makebox[14.400000000000004bp][l]{\includegraphics[trim=2.4000000000000004 2.4000000000000004 2.4000000000000004 2.4000000000000004]{pict_20.pdf}}}
expression that calls some function \raisebox{-3.1874999999999982bp}{\makebox[3.753124999999999bp][l]{\includegraphics[trim=2.4000000000000004 2.4000000000000004 2.4000000000000004 2.4000000000000004]{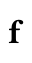}}} in scope and  makes the result
\raisebox{-3.1874999999999982bp}{\makebox[4.339062499999999bp][l]{\includegraphics[trim=2.4000000000000004 2.4000000000000004 2.4000000000000004 2.4000000000000004]{pict_9.pdf}}} available to a nested \raisebox{-2.3617187499999996bp}{\makebox[28.678906250000004bp][l]{\includegraphics[trim=2.4000000000000004 2.4000000000000004 2.4000000000000004 2.4000000000000004]{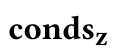}}} expression.\NoteBox{\NoteContent{We use \raisebox{-2.3617187499999996bp}{\makebox[61.02421875bp][l]{\includegraphics[trim=2.4000000000000004 2.4000000000000004 2.4000000000000004 2.4000000000000004]{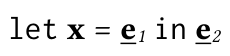}}}
as shorthand for \raisebox{-2.3617187499999996bp}{\makebox[43.19375000000001bp][l]{\includegraphics[trim=2.4000000000000004 2.4000000000000004 2.4000000000000004 2.4000000000000004]{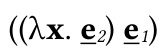}}}.}}

The list of path constraints  \raisebox{-3.1874999999999982bp}{\makebox[5.682812500000001bp][l]{\includegraphics[trim=2.4000000000000004 2.4000000000000004 2.4000000000000004 2.4000000000000004]{pict_99.pdf}}} consists of three kinds of constraints
\raisebox{-3.1874999999999982bp}{\makebox[5.577343750000001bp][l]{\includegraphics[trim=2.4000000000000004 2.4000000000000004 2.4000000000000004 2.4000000000000004]{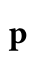}}}. The first are of the form
\raisebox{-3.1874999999999982bp}{\makebox[47.553906250000004bp][l]{\includegraphics[trim=2.4000000000000004 2.4000000000000004 2.4000000000000004 2.4000000000000004]{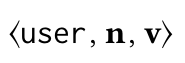}}} and record the result \raisebox{-3.1874999999999982bp}{\makebox[5.913281250000001bp][l]{\includegraphics[trim=2.4000000000000004 2.4000000000000004 2.4000000000000004 2.4000000000000004]{pict_96.pdf}}} (\raisebox{-3.1874999999999982bp}{\makebox[4.46328125bp][l]{\includegraphics[trim=2.4000000000000004 2.4000000000000004 2.4000000000000004 2.4000000000000004]{pict_7.pdf}}} or \raisebox{-3.1874999999999982bp}{\makebox[4.46328125bp][l]{\includegraphics[trim=2.4000000000000004 2.4000000000000004 2.4000000000000004 2.4000000000000004]{pict_8.pdf}}}) of a test
in a \raisebox{-3.1874999999999982bp}{\makebox[19.200000000000003bp][l]{\includegraphics[trim=2.4000000000000004 2.4000000000000004 2.4000000000000004 2.4000000000000004]{pict_12.pdf}}} expression of the user program together with its expression
trace {---} recall that traces are embedded in the representation of
first{-}order values in the concolic language. These constraints are akin to the
first{-}order constraints from \ChapRef{\SectionNumberLink{t:x28part_x22secx3ahowx2dgeneratex22x29}{4}}{Directed Evolution of Canonical Functions} such as
\raisebox{-3.2890624999999982bp}{\makebox[59.68593750000001bp][l]{\includegraphics[trim=2.4000000000000004 2.4000000000000004 2.4000000000000004 2.4000000000000004]{pict_46.pdf}}}. The other two kinds correspond to
the input{-}related constraints from \ChapRef{\SectionNumberLink{t:x28part_x22secx3ahowx2dgeneratex22x29}{4}}{Directed Evolution of Canonical Functions}.
In the model, branch constraints
are of the form \raisebox{-3.1874999999999982bp}{\makebox[31.473437500000003bp][l]{\includegraphics[trim=2.4000000000000004 2.4000000000000004 2.4000000000000004 2.4000000000000004]{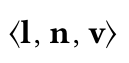}}} and, as before, record whether the
test from branch \raisebox{-3.1874999999999982bp}{\makebox[3.1195312499999996bp][l]{\includegraphics[trim=2.4000000000000004 2.4000000000000004 2.4000000000000004 2.4000000000000004]{pict_117.pdf}}} of a \raisebox{-3.1874999999999982bp}{\makebox[19.200000000000003bp][l]{\includegraphics[trim=2.4000000000000004 2.4000000000000004 2.4000000000000004 2.4000000000000004]{pict_12.pdf}}} expression in an input succeeds or fails
together with the (traced) result of the test. Test constraints
take the form \raisebox{-3.1874999999999982bp}{\makebox[30.023437500000007bp][l]{\includegraphics[trim=2.4000000000000004 2.4000000000000004 2.4000000000000004 2.4000000000000004]{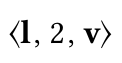}}} and they record that all the tests of a
\raisebox{-3.1874999999999982bp}{\makebox[19.200000000000003bp][l]{\includegraphics[trim=2.4000000000000004 2.4000000000000004 2.4000000000000004 2.4000000000000004]{pict_12.pdf}}} expression whose \raisebox{-3.1874999999999982bp}{\makebox[19.200000000000003bp][l]{\includegraphics[trim=2.4000000000000004 2.4000000000000004 2.4000000000000004 2.4000000000000004]{pict_50.pdf}}} branch has label  \raisebox{-3.1874999999999982bp}{\makebox[3.1195312499999996bp][l]{\includegraphics[trim=2.4000000000000004 2.4000000000000004 2.4000000000000004 2.4000000000000004]{pict_117.pdf}}}  may inspect value \raisebox{-3.1874999999999982bp}{\makebox[5.078125bp][l]{\includegraphics[trim=2.4000000000000004 2.4000000000000004 2.4000000000000004 2.4000000000000004]{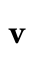}}}.

\begin{Figure}\begin{Centerfigure}\begin{FigureInside}\raisebox{-2.949218749999977bp}{\makebox[354.1070312500001bp][l]{\includegraphics[trim=2.4000000000000004 2.4000000000000004 2.4000000000000004 2.4000000000000004]{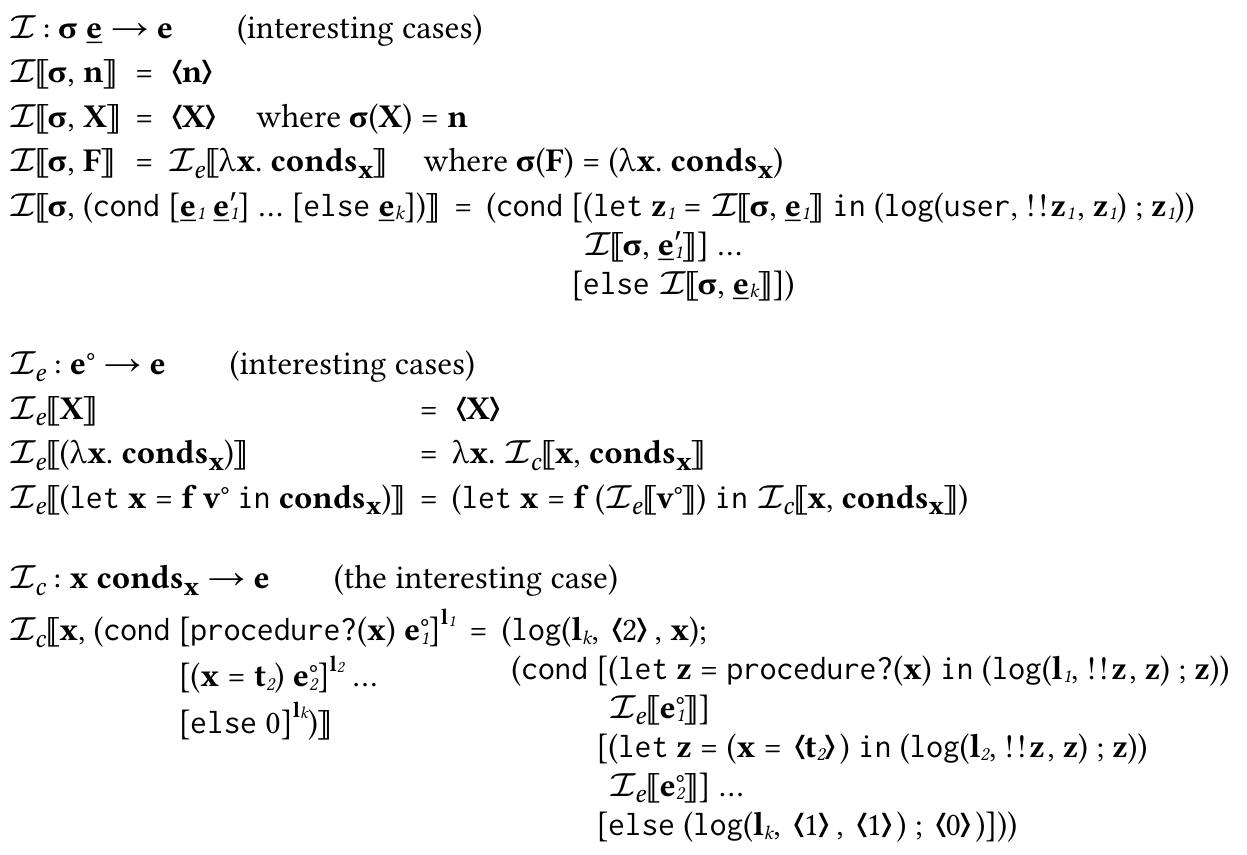}}}\end{FigureInside}\end{Centerfigure}

\Centertext{\Legend{\FigureTarget{\label{t:x28counter_x28x22figurex22_x22figx3ainstrumentationx22x29x29}\textsf{Fig.}~\textsf{5}. }{t:x28counter_x28x22figurex22_x22figx3ainstrumentationx22x29x29}\textsf{Interesting Cases of The Instrumentation Functions}}}\end{Figure}

The instrumentation meta{-}function \raisebox{-3.1874999999999982bp}{\makebox[7.487500000000001bp][l]{\includegraphics[trim=2.4000000000000004 2.4000000000000004 2.4000000000000004 2.4000000000000004]{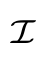}}}
consumes a store \raisebox{-3.1874999999999982bp}{\makebox[5.471875000000001bp][l]{\includegraphics[trim=2.4000000000000004 2.4000000000000004 2.4000000000000004 2.4000000000000004]{pict_95.pdf}}} and a user program
\raisebox{-3.1874999999999982bp}{\makebox[4.69375bp][l]{\includegraphics[trim=2.4000000000000004 2.4000000000000004 2.4000000000000004 2.4000000000000004]{pict_93.pdf}}} and produces an equivalent concolic program
\raisebox{-3.1874999999999982bp}{\makebox[4.69375bp][l]{\includegraphics[trim=2.4000000000000004 2.4000000000000004 2.4000000000000004 2.4000000000000004]{pict_94.pdf}}} that has \raisebox{-3.1874999999999982bp}{\makebox[14.400000000000004bp][l]{\includegraphics[trim=2.4000000000000004 2.4000000000000004 2.4000000000000004 2.4000000000000004]{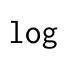}}} expressions at the appropriate places
to record path constraints.
Figure~\hyperref[t:x28counter_x28x22figurex22_x22figx3ainstrumentationx22x29x29]{\FigureRef{5}{t:x28counter_x28x22figurex22_x22figx3ainstrumentationx22x29x29}} shows the interesting cases of
\raisebox{-3.1874999999999982bp}{\makebox[30.324218750000007bp][l]{\includegraphics[trim=2.4000000000000004 2.4000000000000004 2.4000000000000004 2.4000000000000004]{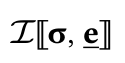}}} while the remaining ones
recur structurally on \raisebox{-3.1874999999999982bp}{\makebox[4.69375bp][l]{\includegraphics[trim=2.4000000000000004 2.4000000000000004 2.4000000000000004 2.4000000000000004]{pict_93.pdf}}}.
The interesting cases are
those concerning numbers, \raisebox{-3.1874999999999982bp}{\makebox[19.200000000000003bp][l]{\includegraphics[trim=2.4000000000000004 2.4000000000000004 2.4000000000000004 2.4000000000000004]{pict_12.pdf}}} expressions and concolic variables.
For a number \raisebox{-3.1874999999999982bp}{\makebox[5.913281250000001bp][l]{\includegraphics[trim=2.4000000000000004 2.4000000000000004 2.4000000000000004 2.4000000000000004]{pict_96.pdf}}}, the instrumentation embeds \raisebox{-3.1874999999999982bp}{\makebox[5.913281250000001bp][l]{\includegraphics[trim=2.4000000000000004 2.4000000000000004 2.4000000000000004 2.4000000000000004]{pict_96.pdf}}} in a traced value.
For a \raisebox{-3.1874999999999982bp}{\makebox[19.200000000000003bp][l]{\includegraphics[trim=2.4000000000000004 2.4000000000000004 2.4000000000000004 2.4000000000000004]{pict_12.pdf}}} expression, it transforms the expression in a mostly
recursive manner except that it injects a \raisebox{-3.1874999999999982bp}{\makebox[14.400000000000004bp][l]{\includegraphics[trim=2.4000000000000004 2.4000000000000004 2.4000000000000004 2.4000000000000004]{pict_129.pdf}}} expression for the
result of the test of each
branch that records it together with its expression
trace.\NoteBox{\NoteContent{\raisebox{-3.1874999999999982bp}{\makebox[4.800000000000001bp][l]{\includegraphics[trim=2.4000000000000004 2.4000000000000004 2.4000000000000004 2.4000000000000004]{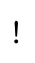}}}\raisebox{-3.1874999999999982bp}{\makebox[4.800000000000001bp][l]{\includegraphics[trim=2.4000000000000004 2.4000000000000004 2.4000000000000004 2.4000000000000004]{pict_131.pdf}}}\raisebox{-2.3617187499999996bp}{\makebox[6.824999999999999bp][l]{\includegraphics[trim=2.4000000000000004 2.4000000000000004 2.4000000000000004 2.4000000000000004]{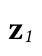}}} is a double negation that turns \raisebox{-2.3617187499999996bp}{\makebox[6.824999999999999bp][l]{\includegraphics[trim=2.4000000000000004 2.4000000000000004 2.4000000000000004 2.4000000000000004]{pict_132.pdf}}} to \raisebox{-3.1874999999999982bp}{\makebox[4.46328125bp][l]{\includegraphics[trim=2.4000000000000004 2.4000000000000004 2.4000000000000004 2.4000000000000004]{pict_8.pdf}}} if
it is a number other than \raisebox{-3.1874999999999982bp}{\makebox[4.46328125bp][l]{\includegraphics[trim=2.4000000000000004 2.4000000000000004 2.4000000000000004 2.4000000000000004]{pict_7.pdf}}}.}} The injected \raisebox{-3.1874999999999982bp}{\makebox[14.400000000000004bp][l]{\includegraphics[trim=2.4000000000000004 2.4000000000000004 2.4000000000000004 2.4000000000000004]{pict_129.pdf}}} expression
generates a first{-}order constraint as we discuss above.

The cases for a concolic variable \raisebox{-3.1874999999999982bp}{\makebox[6.8921875bp][l]{\includegraphics[trim=2.4000000000000004 2.4000000000000004 2.4000000000000004 2.4000000000000004]{pict_2.pdf}}} are the most involved.  If the
store \raisebox{-3.1874999999999982bp}{\makebox[5.471875000000001bp][l]{\includegraphics[trim=2.4000000000000004 2.4000000000000004 2.4000000000000004 2.4000000000000004]{pict_95.pdf}}} maps \raisebox{-3.1874999999999982bp}{\makebox[6.8921875bp][l]{\includegraphics[trim=2.4000000000000004 2.4000000000000004 2.4000000000000004 2.4000000000000004]{pict_2.pdf}}} to a number then, the instrumentation simply embeds
\raisebox{-3.1874999999999982bp}{\makebox[6.8921875bp][l]{\includegraphics[trim=2.4000000000000004 2.4000000000000004 2.4000000000000004 2.4000000000000004]{pict_2.pdf}}} in a traced value similar to the case for numbers.  However, if \raisebox{-3.1874999999999982bp}{\makebox[5.471875000000001bp][l]{\includegraphics[trim=2.4000000000000004 2.4000000000000004 2.4000000000000004 2.4000000000000004]{pict_95.pdf}}}
maps \raisebox{-3.1874999999999982bp}{\makebox[6.8921875bp][l]{\includegraphics[trim=2.4000000000000004 2.4000000000000004 2.4000000000000004 2.4000000000000004]{pict_2.pdf}}} to a canonical function, the instrumentation cannot turn \raisebox{-3.1874999999999982bp}{\makebox[6.8921875bp][l]{\includegraphics[trim=2.4000000000000004 2.4000000000000004 2.4000000000000004 2.4000000000000004]{pict_2.pdf}}}
to the corresponding traced value. After all, expression traces aim to capture
 formulas that the concolic tester records in path constraints
so that it can use them to
issue queries to the SMT solver. Thus they have to refer strictly to
variables that hold first{-}order data.  Consequently, \raisebox{-3.1874999999999982bp}{\makebox[7.487500000000001bp][l]{\includegraphics[trim=2.4000000000000004 2.4000000000000004 2.4000000000000004 2.4000000000000004]{pict_128.pdf}}}
delegates to two further meta{-}functions \raisebox{-2.3617187499999996bp}{\makebox[10.6953125bp][l]{\includegraphics[trim=2.4000000000000004 2.4000000000000004 2.4000000000000004 2.4000000000000004]{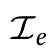}}} and
\raisebox{-2.3617187499999996bp}{\makebox[10.59921875bp][l]{\includegraphics[trim=2.4000000000000004 2.4000000000000004 2.4000000000000004 2.4000000000000004]{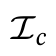}}} that produce an instrumented function that gets
substituted for the concolic variable in the concolic
program.

The \raisebox{-2.3617187499999996bp}{\makebox[10.6953125bp][l]{\includegraphics[trim=2.4000000000000004 2.4000000000000004 2.4000000000000004 2.4000000000000004]{pict_133.pdf}}} meta{-}function is similar to
\raisebox{-3.1874999999999982bp}{\makebox[7.487500000000001bp][l]{\includegraphics[trim=2.4000000000000004 2.4000000000000004 2.4000000000000004 2.4000000000000004]{pict_128.pdf}}}\NoteBox{\NoteContent{In contrast to \raisebox{-3.1874999999999982bp}{\makebox[7.487500000000001bp][l]{\includegraphics[trim=2.4000000000000004 2.4000000000000004 2.4000000000000004 2.4000000000000004]{pict_128.pdf}}},
\raisebox{-2.3617187499999996bp}{\makebox[10.6953125bp][l]{\includegraphics[trim=2.4000000000000004 2.4000000000000004 2.4000000000000004 2.4000000000000004]{pict_133.pdf}}} has a single case for concolic variables \raisebox{-3.1874999999999982bp}{\makebox[6.8921875bp][l]{\includegraphics[trim=2.4000000000000004 2.4000000000000004 2.4000000000000004 2.4000000000000004]{pict_2.pdf}}}. This is
because by construction concolic variables in canonical functions point to
first{-}order data. We return to that point further on when we discuss the
evolution of inputs.}}
except that it calls  \raisebox{-2.3617187499999996bp}{\makebox[10.59921875bp][l]{\includegraphics[trim=2.4000000000000004 2.4000000000000004 2.4000000000000004 2.4000000000000004]{pict_134.pdf}}} for the instrumentation of
a \raisebox{-3.1874999999999982bp}{\makebox[19.200000000000003bp][l]{\includegraphics[trim=2.4000000000000004 2.4000000000000004 2.4000000000000004 2.4000000000000004]{pict_12.pdf}}} expression and passes along to it the value its tests may
inspect. In turn, \raisebox{-2.3617187499999996bp}{\makebox[10.59921875bp][l]{\includegraphics[trim=2.4000000000000004 2.4000000000000004 2.4000000000000004 2.4000000000000004]{pict_134.pdf}}} adds at the beginning of the
\raisebox{-3.1874999999999982bp}{\makebox[19.200000000000003bp][l]{\includegraphics[trim=2.4000000000000004 2.4000000000000004 2.4000000000000004 2.4000000000000004]{pict_12.pdf}}} expression a \raisebox{-3.1874999999999982bp}{\makebox[14.400000000000004bp][l]{\includegraphics[trim=2.4000000000000004 2.4000000000000004 2.4000000000000004 2.4000000000000004]{pict_129.pdf}}} expression that records a test path
constraint for the conditional expression and then similar to \raisebox{-3.1874999999999982bp}{\makebox[7.487500000000001bp][l]{\includegraphics[trim=2.4000000000000004 2.4000000000000004 2.4000000000000004 2.4000000000000004]{pict_128.pdf}}}, it
injects a \raisebox{-3.1874999999999982bp}{\makebox[14.400000000000004bp][l]{\includegraphics[trim=2.4000000000000004 2.4000000000000004 2.4000000000000004 2.4000000000000004]{pict_129.pdf}}} expression in each branch of the \raisebox{-3.1874999999999982bp}{\makebox[19.200000000000003bp][l]{\includegraphics[trim=2.4000000000000004 2.4000000000000004 2.4000000000000004 2.4000000000000004]{pict_12.pdf}}} expression to
record branch constraints.

\begin{Figure}\begin{Centerfigure}\begin{FigureInside}\raisebox{-0.9421874999999886bp}{\makebox[354.38671875000006bp][l]{\includegraphics[trim=2.4000000000000004 2.4000000000000004 2.4000000000000004 2.4000000000000004]{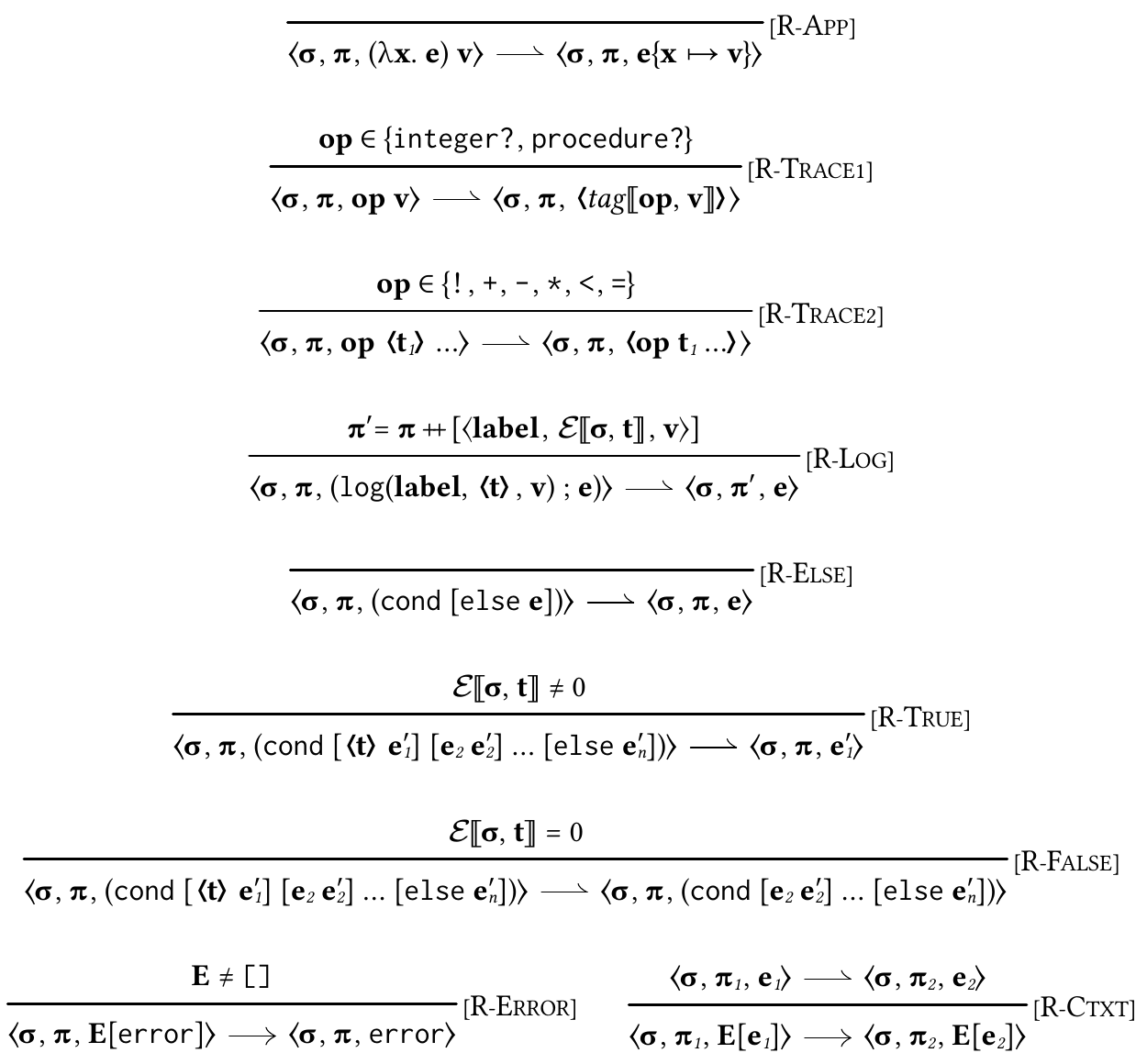}}}\hspace*{\fill}\\\hspace*{\fill}\\\raisebox{-2.386718749999991bp}{\makebox[349.20000000000005bp][l]{\includegraphics[trim=2.4000000000000004 2.4000000000000004 2.4000000000000004 2.4000000000000004]{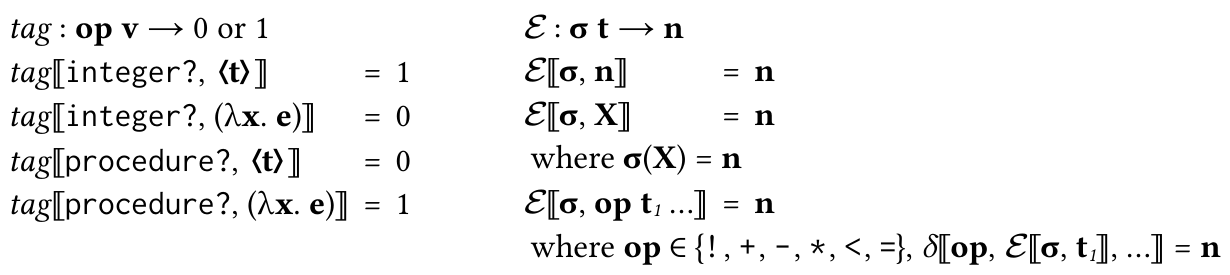}}}\end{FigureInside}\end{Centerfigure}

\Centertext{\Legend{\FigureTarget{\label{t:x28counter_x28x22figurex22_x22figx3aconcolicx2dexecutionx22x29x29}\textsf{Fig.}~\textsf{6}. }{t:x28counter_x28x22figurex22_x22figx3aconcolicx2dexecutionx22x29x29}\textsf{The Reduction Relation of Concolic Execution}}}\end{Figure}

We conclude this section with a discussion of the reduction rules for concolic evaluation.
from figure~\hyperref[t:x28counter_x28x22figurex22_x22figx3aconcolicx2dexecutionx22x29x29]{\FigureRef{6}{t:x28counter_x28x22figurex22_x22figx3aconcolicx2dexecutionx22x29x29}}.
Rule \raisebox{-2.6179687499999993bp}{\makebox[27.6546875bp][l]{\includegraphics[trim=2.4000000000000004 2.4000000000000004 2.4000000000000004 2.4000000000000004]{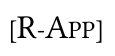}}} is the standard call{-}by{-}value
$\beta${-}reduction.
The next two rules, \raisebox{-2.6179687499999993bp}{\makebox[39.917187500000004bp][l]{\includegraphics[trim=2.4000000000000004 2.4000000000000004 2.4000000000000004 2.4000000000000004]{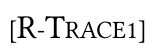}}} and \raisebox{-2.6179687499999993bp}{\makebox[39.917187500000004bp][l]{\includegraphics[trim=2.4000000000000004 2.4000000000000004 2.4000000000000004 2.4000000000000004]{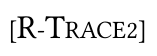}}},
reduce primitive operators and produce appropriate expression traces.
When the given operator is one of \raisebox{-3.1874999999999982bp}{\makebox[38.400000000000006bp][l]{\includegraphics[trim=2.4000000000000004 2.4000000000000004 2.4000000000000004 2.4000000000000004]{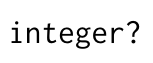}}} or \raisebox{-3.1874999999999982bp}{\makebox[48.00000000000001bp][l]{\includegraphics[trim=2.4000000000000004 2.4000000000000004 2.4000000000000004 2.4000000000000004]{pict_58.pdf}}},
\raisebox{-2.6179687499999993bp}{\makebox[39.917187500000004bp][l]{\includegraphics[trim=2.4000000000000004 2.4000000000000004 2.4000000000000004 2.4000000000000004]{pict_138.pdf}}} inspects the tag of its argument and produces \raisebox{-3.2890624999999982bp}{\makebox[12.822656250000001bp][l]{\includegraphics[trim=2.4000000000000004 2.4000000000000004 2.4000000000000004 2.4000000000000004]{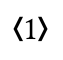}}}
if the tag is the right one and \raisebox{-3.2890624999999982bp}{\makebox[12.822656250000001bp][l]{\includegraphics[trim=2.4000000000000004 2.4000000000000004 2.4000000000000004 2.4000000000000004]{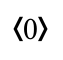}}} otherwise.
When the given operator is not a predicate,
\raisebox{-2.6179687499999993bp}{\makebox[39.917187500000004bp][l]{\includegraphics[trim=2.4000000000000004 2.4000000000000004 2.4000000000000004 2.4000000000000004]{pict_139.pdf}}} constructs an expression trace
from the given operator and the traces of the operands.
Rule \raisebox{-2.6179687499999993bp}{\makebox[28.247656250000006bp][l]{\includegraphics[trim=2.4000000000000004 2.4000000000000004 2.4000000000000004 2.4000000000000004]{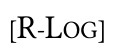}}} appends a new path constraint to the
list of path constraints \raisebox{-3.1874999999999982bp}{\makebox[5.682812500000001bp][l]{\includegraphics[trim=2.4000000000000004 2.4000000000000004 2.4000000000000004 2.4000000000000004]{pict_99.pdf}}} of the machine and then proceeds with the
\raisebox{-3.1874999999999982bp}{\makebox[4.69375bp][l]{\includegraphics[trim=2.4000000000000004 2.4000000000000004 2.4000000000000004 2.4000000000000004]{pict_94.pdf}}} expression.
The new path constraint contains the label from the first argument
of \raisebox{-3.1874999999999982bp}{\makebox[14.400000000000004bp][l]{\includegraphics[trim=2.4000000000000004 2.4000000000000004 2.4000000000000004 2.4000000000000004]{pict_129.pdf}}}, the number that corresponds to the traced value of the second
argument and the traced value of the third argument.
The next three rules, \raisebox{-2.6179687499999993bp}{\makebox[29.843750000000007bp][l]{\includegraphics[trim=2.4000000000000004 2.4000000000000004 2.4000000000000004 2.4000000000000004]{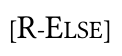}}}, \raisebox{-2.6179687499999993bp}{\makebox[31.918750000000003bp][l]{\includegraphics[trim=2.4000000000000004 2.4000000000000004 2.4000000000000004 2.4000000000000004]{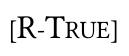}}} and \raisebox{-2.6179687499999993bp}{\makebox[33.732031250000006bp][l]{\includegraphics[trim=2.4000000000000004 2.4000000000000004 2.4000000000000004 2.4000000000000004]{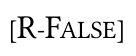}}}
govern the evaluation of \raisebox{-3.1874999999999982bp}{\makebox[19.200000000000003bp][l]{\includegraphics[trim=2.4000000000000004 2.4000000000000004 2.4000000000000004 2.4000000000000004]{pict_12.pdf}}} expressions .
If the number that corresponds  to the \raisebox{-3.2890624999999982bp}{\makebox[11.79609375bp][l]{\includegraphics[trim=2.4000000000000004 2.4000000000000004 2.4000000000000004 2.4000000000000004]{pict_108.pdf}}} test expression of
the first branch of the \raisebox{-3.1874999999999982bp}{\makebox[19.200000000000003bp][l]{\includegraphics[trim=2.4000000000000004 2.4000000000000004 2.4000000000000004 2.4000000000000004]{pict_12.pdf}}} expression is non{-}zero,
\raisebox{-2.6179687499999993bp}{\makebox[31.918750000000003bp][l]{\includegraphics[trim=2.4000000000000004 2.4000000000000004 2.4000000000000004 2.4000000000000004]{pict_145.pdf}}} reduces the \raisebox{-3.1874999999999982bp}{\makebox[19.200000000000003bp][l]{\includegraphics[trim=2.4000000000000004 2.4000000000000004 2.4000000000000004 2.4000000000000004]{pict_12.pdf}}} expression to  the
body of the first branch. Otherwise, \raisebox{-2.6179687499999993bp}{\makebox[33.732031250000006bp][l]{\includegraphics[trim=2.4000000000000004 2.4000000000000004 2.4000000000000004 2.4000000000000004]{pict_146.pdf}}} drops the first branch.
When the given \raisebox{-3.1874999999999982bp}{\makebox[19.200000000000003bp][l]{\includegraphics[trim=2.4000000000000004 2.4000000000000004 2.4000000000000004 2.4000000000000004]{pict_12.pdf}}} expression only contains the \raisebox{-3.1874999999999982bp}{\makebox[19.200000000000003bp][l]{\includegraphics[trim=2.4000000000000004 2.4000000000000004 2.4000000000000004 2.4000000000000004]{pict_50.pdf}}} branch,
\raisebox{-2.6179687499999993bp}{\makebox[29.843750000000007bp][l]{\includegraphics[trim=2.4000000000000004 2.4000000000000004 2.4000000000000004 2.4000000000000004]{pict_144.pdf}}} proceeds with its body.
Finally, rules \raisebox{-2.6179687499999993bp}{\makebox[36.27187500000001bp][l]{\includegraphics[trim=2.4000000000000004 2.4000000000000004 2.4000000000000004 2.4000000000000004]{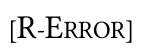}}} and \raisebox{-2.6179687499999993bp}{\makebox[32.74218750000001bp][l]{\includegraphics[trim=2.4000000000000004 2.4000000000000004 2.4000000000000004 2.4000000000000004]{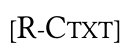}}} introduce the
compatible closure of the reduction relation (modulo errors).

\Ssubsection{Evolving  New Inputs}{Evolving  New Inputs}\label{t:x28part_x22secx3amodelx2devolvex22x29}

\begin{Figure}\begin{Centerfigure}\begin{FigureInside}\fbox{\parbox{0.96\textwidth}{\begin{quote}\noindent
\raisebox{-2.6179687499999993bp}{\makebox[67.65624999999999bp][l]{\includegraphics[trim=2.4000000000000004 2.4000000000000004 2.4000000000000004 2.4000000000000004]{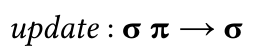}}}

\begin{quote}Given a store \raisebox{-3.1874999999999982bp}{\makebox[5.471875000000001bp][l]{\includegraphics[trim=2.4000000000000004 2.4000000000000004 2.4000000000000004 2.4000000000000004]{pict_95.pdf}}} and a list of path constraints \raisebox{-3.1874999999999982bp}{\makebox[5.682812500000001bp][l]{\includegraphics[trim=2.4000000000000004 2.4000000000000004 2.4000000000000004 2.4000000000000004]{pict_99.pdf}}},
encodes the path constraints as an SMT query,
invokes the SMT solver and updates the store with the solution.
The encoding of the list of path constraints into the query involves
(i) asserting all first-order constraints in the list;
(ii) asserting the expression traces from the branch constraints
      in the list if the branches' tests succeed or their negation if they fail and
(iii) asserting a constraint per conditional expression in the store that entails
the branches of the conditional are disjoint.\end{quote}
\vspace{0.7em}

\noindent
\raisebox{-2.5304687499999954bp}{\makebox[142.2578125bp][l]{\includegraphics[trim=2.4000000000000004 2.4000000000000004 2.4000000000000004 2.4000000000000004]{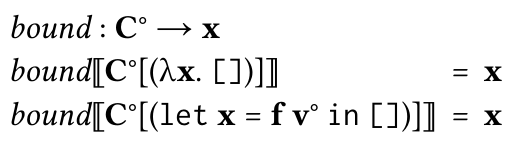}}}

\begin{quote}Given a canonical context s
that ends with a \textlambda-abstraction
or a \raisebox{-3.1874999999999982bp}{\makebox[14.400000000000004bp][l]{\includegraphics[trim=2.4000000000000004 2.4000000000000004 2.4000000000000004 2.4000000000000004]{pict_20.pdf}}}-binding,
extracts the variable introduced by the innermost binder.\end{quote}
\vspace{0.7em}

\noindent
$\raisebox{-3.1874999999999982bp}{\makebox[6.64296875bp][l]{\includegraphics[trim=2.4000000000000004 2.4000000000000004 2.4000000000000004 2.4000000000000004]{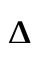}}} \subset \{ \raisebox{-3.1874999999999982bp}{\makebox[5.38515625bp][l]{\includegraphics[trim=2.4000000000000004 2.4000000000000004 2.4000000000000004 2.4000000000000004]{pict_116.pdf}}}, \raisebox{-3.1874999999999982bp}{\makebox[5.356249999999999bp][l]{\includegraphics[trim=2.4000000000000004 2.4000000000000004 2.4000000000000004 2.4000000000000004]{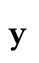}}}, \raisebox{-3.1874999999999982bp}{\makebox[4.339062499999999bp][l]{\includegraphics[trim=2.4000000000000004 2.4000000000000004 2.4000000000000004 2.4000000000000004]{pict_9.pdf}}}, \dots \}$
   stands for any finite subset of non-concolic variables.
\vspace{0.7em}

\noindent
\raisebox{-2.3617187499999996bp}{\makebox[147.99375bp][l]{\includegraphics[trim=2.4000000000000004 2.4000000000000004 2.4000000000000004 2.4000000000000004]{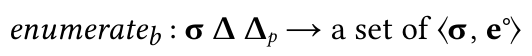}}}
\begin{quote}Given a store, two sets of variables ,
computes the set of new branch bodies, creating fresh concolic variables
and mapping them to numbers in the resulting store as needed.
The bodies of the new branches can refer to any
variables in the first set while let expressions
can apply any the
of variables in the second set.\end{quote}
\vspace{0.7em}

\noindent
\raisebox{-2.6179687499999993bp}{\makebox[59.946875000000006bp][l]{\includegraphics[trim=2.4000000000000004 2.4000000000000004 2.4000000000000004 2.4000000000000004]{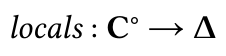}}}
\begin{quote}Given a canonical context d,
computes the set of all variables in scope in the hole.\end{quote}
\vspace{0.7em}

\noindent
\raisebox{-2.3617187499999996bp}{\makebox[63.858593750000004bp][l]{\includegraphics[trim=2.4000000000000004 2.4000000000000004 2.4000000000000004 2.4000000000000004]{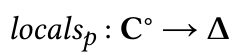}}}
\begin{quote}Given a canonical context of variables in scope in the whole that are bound to functions.\end{quote}\end{quote}}}\end{FigureInside}\end{Centerfigure}

\Centertext{\Legend{\FigureTarget{\label{t:x28counter_x28x22figurex22_x22figx3amutatex2dhelperx22x29x29}\textsf{Fig.}~\textsf{7}. }{t:x28counter_x28x22figurex22_x22figx3amutatex2dhelperx22x29x29}\textsf{Metafunctions Used by \raisebox{-3.1874999999999982bp}{\makebox[23.3828125bp][l]{\includegraphics[trim=2.4000000000000004 2.4000000000000004 2.4000000000000004 2.4000000000000004]{pict_100.pdf}}}
to Update Inputs and Compute Auxiliary Information}}}\end{Figure}

Once it obtains a list of path constraints \raisebox{-3.1874999999999982bp}{\makebox[5.682812500000001bp][l]{\includegraphics[trim=2.4000000000000004 2.4000000000000004 2.4000000000000004 2.4000000000000004]{pict_99.pdf}}} from the concolic evaluation of
a user program (after instrumentation), the concolic tester uses
meta{-}function \raisebox{-3.1874999999999982bp}{\makebox[23.3828125bp][l]{\includegraphics[trim=2.4000000000000004 2.4000000000000004 2.4000000000000004 2.4000000000000004]{pict_100.pdf}}} to modify the current input \raisebox{-3.1874999999999982bp}{\makebox[5.471875000000001bp][l]{\includegraphics[trim=2.4000000000000004 2.4000000000000004 2.4000000000000004 2.4000000000000004]{pict_95.pdf}}}.
Figure~\hyperref[t:x28counter_x28x22figurex22_x22figx3amutate1x22x29x29]{\FigureRef{8}{t:x28counter_x28x22figurex22_x22figx3amutate1x22x29x29}} through figure~\hyperref[t:x28counter_x28x22figurex22_x22figx3amutate3x22x29x29]{\FigureRef{10}{t:x28counter_x28x22figurex22_x22figx3amutate3x22x29x29}} present the
formal definition of \raisebox{-3.1874999999999982bp}{\makebox[47.208593750000006bp][l]{\includegraphics[trim=2.4000000000000004 2.4000000000000004 2.4000000000000004 2.4000000000000004]{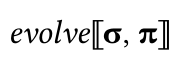}}} as a non{-}deterministic process
that describes all the valid next inputs after
an iteration of the concolic loop. The meta{-}function
also returns a new list of path constraints \raisebox{-3.1874999999999982bp}{\makebox[5.682812500000001bp][l]{\includegraphics[trim=2.4000000000000004 2.4000000000000004 2.4000000000000004 2.4000000000000004]{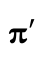}}} that, as we show in
\ChapRef{\SectionNumberLink{t:x28part_x22secx3atheoryx22x29}{6}}{Correctness of Higher{-}Order Concolic Testing}, predicts the path
constraints produced by the evaluation of  the user program with input
\raisebox{-3.1874999999999982bp}{\makebox[5.471875000000001bp][l]{\includegraphics[trim=2.4000000000000004 2.4000000000000004 2.4000000000000004 2.4000000000000004]{pict_101.pdf}}}.
Figure~\hyperref[t:x28counter_x28x22figurex22_x22figx3amutatex2dhelperx22x29x29]{\FigureRef{7}{t:x28counter_x28x22figurex22_x22figx3amutatex2dhelperx22x29x29}} provides a summary of the auxiliary
metafunctions that \raisebox{-3.1874999999999982bp}{\makebox[23.3828125bp][l]{\includegraphics[trim=2.4000000000000004 2.4000000000000004 2.4000000000000004 2.4000000000000004]{pict_100.pdf}}} employs. The
\raisebox{-3.1874999999999982bp}{\makebox[25.917187500000004bp][l]{\includegraphics[trim=2.4000000000000004 2.4000000000000004 2.4000000000000004 2.4000000000000004]{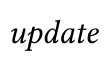}}} meta{-}function offers an interface to the SMT solver while
the rest compute information about the contexts of the store \raisebox{-3.1874999999999982bp}{\makebox[5.471875000000001bp][l]{\includegraphics[trim=2.4000000000000004 2.4000000000000004 2.4000000000000004 2.4000000000000004]{pict_95.pdf}}}.
The supplementary material contains their complete definitions.

\begin{Figure}\begin{Centerfigure}\begin{FigureInside}\raisebox{-3.4320312499999943bp}{\makebox[324.42234375000004bp][l]{\includegraphics[trim=2.4000000000000004 2.4000000000000004 2.4000000000000004 2.4000000000000004]{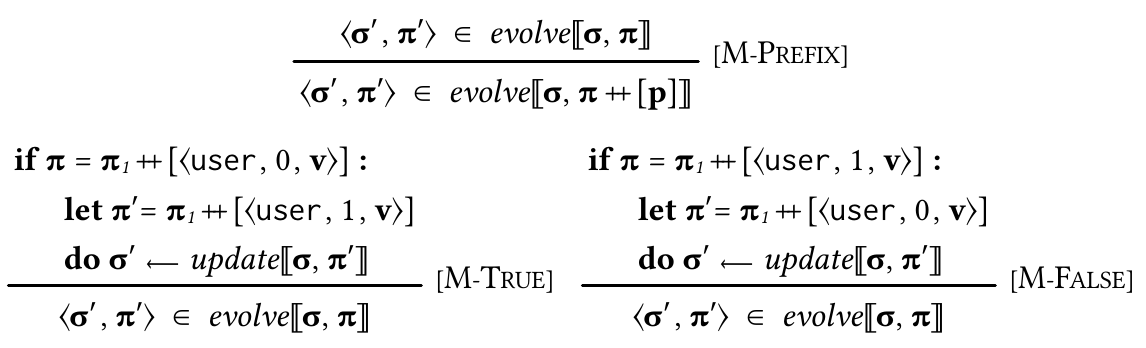}}}\end{FigureInside}\end{Centerfigure}

\Centertext{\Legend{\FigureTarget{\label{t:x28counter_x28x22figurex22_x22figx3amutate1x22x29x29}\textsf{Fig.}~\textsf{8}. }{t:x28counter_x28x22figurex22_x22figx3amutate1x22x29x29}\textsf{Evolving New Inputs (i): Negating Branches in User Programs}}}\end{Figure}

\begin{Figure}\begin{Centerfigure}\begin{FigureInside}\raisebox{-3.162499999999966bp}{\makebox[347.81250000000006bp][l]{\includegraphics[trim=2.4000000000000004 2.4000000000000004 2.4000000000000004 2.4000000000000004]{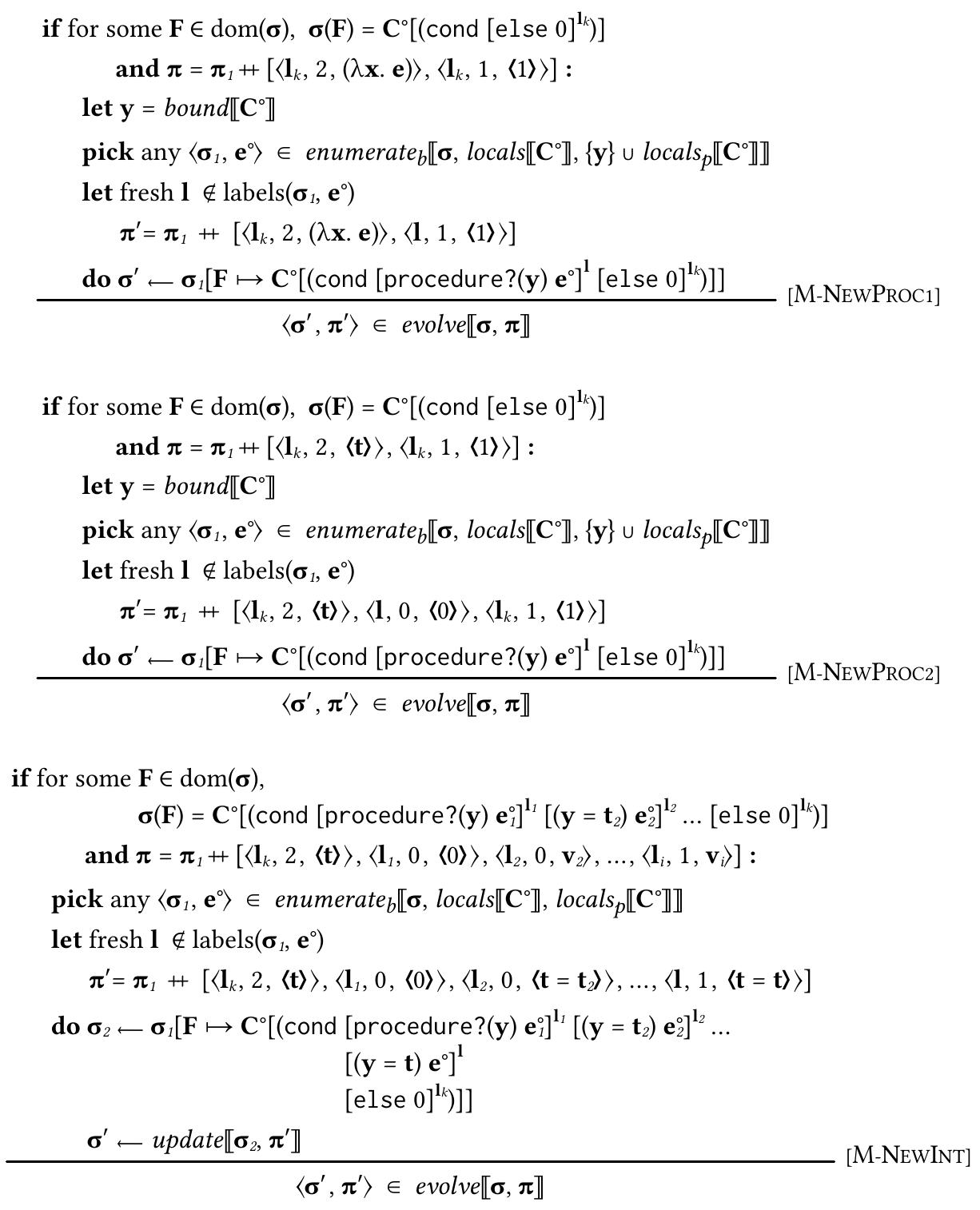}}}\end{FigureInside}\end{Centerfigure}

\Centertext{\Legend{\FigureTarget{\label{t:x28counter_x28x22figurex22_x22figx3amutate2x22x29x29}\textsf{Fig.}~\textsf{9}. }{t:x28counter_x28x22figurex22_x22figx3amutate2x22x29x29}\textsf{Evolving New Inputs (ii): Adding New Branches in Canonical Functions}}}\end{Figure}

\begin{Figure}\begin{Centerfigure}\begin{FigureInside}\raisebox{-3.1624999999999943bp}{\makebox[348.2328125000001bp][l]{\includegraphics[trim=2.4000000000000004 2.4000000000000004 2.4000000000000004 2.4000000000000004]{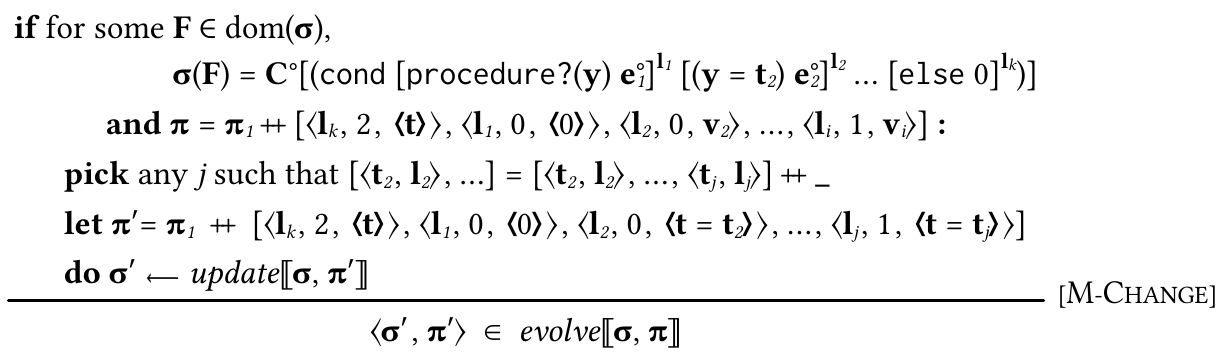}}}\end{FigureInside}\end{Centerfigure}

\Centertext{\Legend{\FigureTarget{\label{t:x28counter_x28x22figurex22_x22figx3amutate3x22x29x29}\textsf{Fig.}~\textsf{10}. }{t:x28counter_x28x22figurex22_x22figx3amutate3x22x29x29}\textsf{Evolving New Inputs (iii): Targeted Branch Constraint Modification}}}\end{Figure}

Figure~\hyperref[t:x28counter_x28x22figurex22_x22figx3amutate1x22x29x29]{\FigureRef{8}{t:x28counter_x28x22figurex22_x22figx3amutate1x22x29x29}} presents the first three rules of
\raisebox{-3.1874999999999982bp}{\makebox[23.3828125bp][l]{\includegraphics[trim=2.4000000000000004 2.4000000000000004 2.4000000000000004 2.4000000000000004]{pict_100.pdf}}}. The first rule, \raisebox{-2.6179687499999993bp}{\makebox[39.42265625bp][l]{\includegraphics[trim=2.4000000000000004 2.4000000000000004 2.4000000000000004 2.4000000000000004]{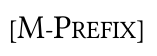}}}, allows the
concolic tester to cut off any suffix from the list of path
constraints and focus on the remaining prefix.
The next two rules, \raisebox{-2.6179687499999993bp}{\makebox[36.15078125000001bp][l]{\includegraphics[trim=2.4000000000000004 2.4000000000000004 2.4000000000000004 2.4000000000000004]{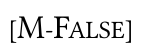}}} and \raisebox{-2.6179687499999993bp}{\makebox[34.337500000000006bp][l]{\includegraphics[trim=2.4000000000000004 2.4000000000000004 2.4000000000000004 2.4000000000000004]{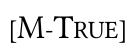}}},
handle the situation where the remaining list, \raisebox{-3.1874999999999982bp}{\makebox[5.682812500000001bp][l]{\includegraphics[trim=2.4000000000000004 2.4000000000000004 2.4000000000000004 2.4000000000000004]{pict_99.pdf}}}, ends with a
(first{-}order) path constraint logged from a conditional expression in the
user program. These two rules negate the last path constraint
and, with \raisebox{-3.1874999999999982bp}{\makebox[25.917187500000004bp][l]{\includegraphics[trim=2.4000000000000004 2.4000000000000004 2.4000000000000004 2.4000000000000004]{pict_158.pdf}}}, consult the SMT solver for a
solution that satisfies
the negated list of path constraints \raisebox{-3.1874999999999982bp}{\makebox[5.682812500000001bp][l]{\includegraphics[trim=2.4000000000000004 2.4000000000000004 2.4000000000000004 2.4000000000000004]{pict_157.pdf}}}.

The next three rules, \raisebox{-2.6179687499999993bp}{\makebox[55.65078125000001bp][l]{\includegraphics[trim=2.4000000000000004 2.4000000000000004 2.4000000000000004 2.4000000000000004]{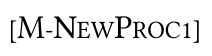}}}, \raisebox{-2.6179687499999993bp}{\makebox[55.65078125000001bp][l]{\includegraphics[trim=2.4000000000000004 2.4000000000000004 2.4000000000000004 2.4000000000000004]{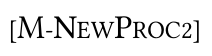}}}
and \raisebox{-2.6179687499999993bp}{\makebox[45.626562500000006bp][l]{\includegraphics[trim=2.4000000000000004 2.4000000000000004 2.4000000000000004 2.4000000000000004]{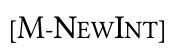}}} in figure~\hyperref[t:x28counter_x28x22figurex22_x22figx3amutate2x22x29x29]{\FigureRef{9}{t:x28counter_x28x22figurex22_x22figx3amutate2x22x29x29}}
handle the insertion of new branches in the conditional expression of inputs.
As we mention in \ChapRef{\SectionNumberLink{t:x28part_x22secx3ahowx2dgeneratex22x29}{4}}{Directed Evolution of Canonical Functions}, when \raisebox{-3.1874999999999982bp}{\makebox[5.682812500000001bp][l]{\includegraphics[trim=2.4000000000000004 2.4000000000000004 2.4000000000000004 2.4000000000000004]{pict_99.pdf}}} ends with a list of path constraints of
the form \raisebox{-3.2890624999999982bp}{\makebox[85.37968749999999bp][l]{\includegraphics[trim=2.4000000000000004 2.4000000000000004 2.4000000000000004 2.4000000000000004]{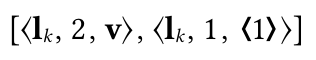}}},
there must be a corresponding canonical function in
\raisebox{-3.1874999999999982bp}{\makebox[5.471875000000001bp][l]{\includegraphics[trim=2.4000000000000004 2.4000000000000004 2.4000000000000004 2.4000000000000004]{pict_95.pdf}}}  with a conditional expression with just an \raisebox{-3.1874999999999982bp}{\makebox[19.200000000000003bp][l]{\includegraphics[trim=2.4000000000000004 2.4000000000000004 2.4000000000000004 2.4000000000000004]{pict_50.pdf}}} branch with label \raisebox{-2.3617187499999996bp}{\makebox[5.840624999999999bp][l]{\includegraphics[trim=2.4000000000000004 2.4000000000000004 2.4000000000000004 2.4000000000000004]{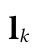}}}.
In this case either rule \raisebox{-2.6179687499999993bp}{\makebox[55.65078125000001bp][l]{\includegraphics[trim=2.4000000000000004 2.4000000000000004 2.4000000000000004 2.4000000000000004]{pict_165.pdf}}} or rule \raisebox{-2.6179687499999993bp}{\makebox[55.65078125000001bp][l]{\includegraphics[trim=2.4000000000000004 2.4000000000000004 2.4000000000000004 2.4000000000000004]{pict_166.pdf}}}
inserts a new \raisebox{-3.1874999999999982bp}{\makebox[48.00000000000001bp][l]{\includegraphics[trim=2.4000000000000004 2.4000000000000004 2.4000000000000004 2.4000000000000004]{pict_58.pdf}}} clause to this conditional expression,
depending on whether the tests of this conditional expression inspect a value that is
a function or a number.
If the suffix of \raisebox{-3.1874999999999982bp}{\makebox[5.682812500000001bp][l]{\includegraphics[trim=2.4000000000000004 2.4000000000000004 2.4000000000000004 2.4000000000000004]{pict_99.pdf}}} indicates that the
conditional expression has more than one branches and the tests of the conditional expression inspect a number,
rule \raisebox{-2.6179687499999993bp}{\makebox[45.626562500000006bp][l]{\includegraphics[trim=2.4000000000000004 2.4000000000000004 2.4000000000000004 2.4000000000000004]{pict_167.pdf}}} inserts a new branch using
 the expression trace of the number.
 For the body of the new branch, the rules
pick one of the options we discuss in \ChapRef{\SectionNumberLink{t:x28part_x22secx3ahowx2dgeneratex22x29}{4}}{Directed Evolution of Canonical Functions}.
using meta{-}function \raisebox{-2.3617187499999996bp}{\makebox[43.900000000000006bp][l]{\includegraphics[trim=2.4000000000000004 2.4000000000000004 2.4000000000000004 2.4000000000000004]{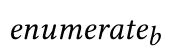}}}.
When the latter introduces a fresh concolic variable as the body of a
branch, it always maps it to a number in the corresponding store in its
result. Thus, inductively, all concolic variables in canonical functions
hold first{-}order data as we note above in our discussion of
instrumentation.

The last rule in figure~\hyperref[t:x28counter_x28x22figurex22_x22figx3amutate3x22x29x29]{\FigureRef{10}{t:x28counter_x28x22figurex22_x22figx3amutate3x22x29x29}}, \raisebox{-2.6179687499999993bp}{\makebox[46.04687500000001bp][l]{\includegraphics[trim=2.4000000000000004 2.4000000000000004 2.4000000000000004 2.4000000000000004]{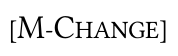}}}, shows  how the
concolic tester performs targeted branch constraint modification.
The goal of this rule is to
cause the evaluation to follow a particular branch of a
conditional expression during a call of a canonical function. The rule does so by adjusting
branch constraints in the argument \raisebox{-3.1874999999999982bp}{\makebox[5.682812500000001bp][l]{\includegraphics[trim=2.4000000000000004 2.4000000000000004 2.4000000000000004 2.4000000000000004]{pict_99.pdf}}} of
\raisebox{-3.1874999999999982bp}{\makebox[23.3828125bp][l]{\includegraphics[trim=2.4000000000000004 2.4000000000000004 2.4000000000000004 2.4000000000000004]{pict_100.pdf}}}. Specifically, it truncates a group of branch  constraints from a conditional expression
and attaches at the end a new branch constraint that would had
been there if the evaluation had followed a particular branch of the
conditional expression. In detail, the new branch constraint  asserts that
the traced value \raisebox{-3.2890624999999982bp}{\makebox[11.79609375bp][l]{\includegraphics[trim=2.4000000000000004 2.4000000000000004 2.4000000000000004 2.4000000000000004]{pict_108.pdf}}} that the tests of the conditional
expression inspect satisfies one of the tests \raisebox{-2.3617187499999996bp}{\makebox[20.41796875bp][l]{\includegraphics[trim=2.4000000000000004 2.4000000000000004 2.4000000000000004 2.4000000000000004]{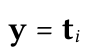}}} from the branches of the
conditional expression. Lastly, the rule consults the SMT solver
for a new store that matches the new list of path constraints.

\sectionNewpage

\Ssection{Correctness of Higher{-}Order Concolic Testing}{Correctness of Higher{-}Order Concolic Testing}\label{t:x28part_x22secx3atheoryx22x29}

\newtheorem{property}[theorem]{Property}

This section establishes three facts about our concolic tester that
together entail its correctness. First, if concolic evaluation of a program triggers a
bug so does the evaluation of the program in the user language (soundness).
Second, the concolic tester can produce inputs that triggers a bug in
the user program, if such input exists (completeness).  Third, for each
iteration of the concolic loop, the concolic tester produces inputs to explore a
specific and selected{-}in{-}advance control{-}flow path of the user program (concolic property). We discuss the formal statements of the
three facts and we provide some interesting details about their proofs.
The complete formal development with all the proofs is in the supplementary material.

Soundness guarantees that concolic evaluation respects the semantics of
the user language. In particular, concolic evaluation does not
discover spurious bugs. Thus, the information that the concolic machine
collects such as path constraints and expression traces does not affect
the behavior of programs. Formally, the  \nameref{correct:sound}
theorem states that for any proper store \raisebox{-3.1874999999999982bp}{\makebox[5.471875000000001bp][l]{\includegraphics[trim=2.4000000000000004 2.4000000000000004 2.4000000000000004 2.4000000000000004]{pict_95.pdf}}},\NoteBox{\NoteContent{ A store \raisebox{-3.1874999999999982bp}{\makebox[5.471875000000001bp][l]{\includegraphics[trim=2.4000000000000004 2.4000000000000004 2.4000000000000004 2.4000000000000004]{pict_95.pdf}}} is
\emph{proper} if (i) all concolic variables occurring free in canonical
functions in \raisebox{-3.1874999999999982bp}{\makebox[5.471875000000001bp][l]{\includegraphics[trim=2.4000000000000004 2.4000000000000004 2.4000000000000004 2.4000000000000004]{pict_95.pdf}}} are mapped to numbers by \raisebox{-3.1874999999999982bp}{\makebox[5.471875000000001bp][l]{\includegraphics[trim=2.4000000000000004 2.4000000000000004 2.4000000000000004 2.4000000000000004]{pict_95.pdf}}}, (ii) all labels in
\raisebox{-3.1874999999999982bp}{\makebox[5.471875000000001bp][l]{\includegraphics[trim=2.4000000000000004 2.4000000000000004 2.4000000000000004 2.4000000000000004]{pict_95.pdf}}} are unique and
(iii) the expression traces in \raisebox{-3.1874999999999982bp}{\makebox[19.200000000000003bp][l]{\includegraphics[trim=2.4000000000000004 2.4000000000000004 2.4000000000000004 2.4000000000000004]{pict_12.pdf}}} expression correspond to distinct numbers.
In this section, we only consider proper stores.  }}
if the concolic evaluation of user program \raisebox{-3.1874999999999982bp}{\makebox[4.69375bp][l]{\includegraphics[trim=2.4000000000000004 2.4000000000000004 2.4000000000000004 2.4000000000000004]{pict_93.pdf}}} with inputs \raisebox{-3.1874999999999982bp}{\makebox[5.471875000000001bp][l]{\includegraphics[trim=2.4000000000000004 2.4000000000000004 2.4000000000000004 2.4000000000000004]{pict_95.pdf}}} reduces to \raisebox{-3.1874999999999982bp}{\makebox[24.0bp][l]{\includegraphics[trim=2.4000000000000004 2.4000000000000004 2.4000000000000004 2.4000000000000004]{pict_4.pdf}}},
  the  evaluation of \raisebox{-3.1874999999999982bp}{\makebox[4.69375bp][l]{\includegraphics[trim=2.4000000000000004 2.4000000000000004 2.4000000000000004 2.4000000000000004]{pict_93.pdf}}} in the user language after retrieving its
  inputs from \raisebox{-3.1874999999999982bp}{\makebox[5.471875000000001bp][l]{\includegraphics[trim=2.4000000000000004 2.4000000000000004 2.4000000000000004 2.4000000000000004]{pict_95.pdf}}} also reduces to \raisebox{-3.1874999999999982bp}{\makebox[24.0bp][l]{\includegraphics[trim=2.4000000000000004 2.4000000000000004 2.4000000000000004 2.4000000000000004]{pict_4.pdf}}}. The metafunction
  \raisebox{-3.1874999999999982bp}{\makebox[31.7484375bp][l]{\includegraphics[trim=2.4000000000000004 2.4000000000000004 2.4000000000000004 2.4000000000000004]{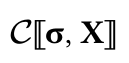}}} retrieves the input of  \raisebox{-3.1874999999999982bp}{\makebox[4.69375bp][l]{\includegraphics[trim=2.4000000000000004 2.4000000000000004 2.4000000000000004 2.4000000000000004]{pict_93.pdf}}} bound to \raisebox{-3.1874999999999982bp}{\makebox[6.8921875bp][l]{\includegraphics[trim=2.4000000000000004 2.4000000000000004 2.4000000000000004 2.4000000000000004]{pict_2.pdf}}}
  in \raisebox{-3.1874999999999982bp}{\makebox[5.471875000000001bp][l]{\includegraphics[trim=2.4000000000000004 2.4000000000000004 2.4000000000000004 2.4000000000000004]{pict_95.pdf}}} by traversing \raisebox{-3.1874999999999982bp}{\makebox[5.471875000000001bp][l]{\includegraphics[trim=2.4000000000000004 2.4000000000000004 2.4000000000000004 2.4000000000000004]{pict_95.pdf}}} recursively and turning traced values into
  their user language counterparts.

\begin{theorem}[Soundness]\label{correct:sound}
For any \raisebox{-3.1874999999999982bp}{\makebox[4.69375bp][l]{\includegraphics[trim=2.4000000000000004 2.4000000000000004 2.4000000000000004 2.4000000000000004]{pict_93.pdf}}} with concolic variables
$\raisebox{-2.3617187499999996bp}{\makebox[9.378125bp][l]{\includegraphics[trim=2.4000000000000004 2.4000000000000004 2.4000000000000004 2.4000000000000004]{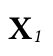}}}, \dots, \raisebox{-2.3617187499999996bp}{\makebox[9.79296875bp][l]{\includegraphics[trim=2.4000000000000004 2.4000000000000004 2.4000000000000004 2.4000000000000004]{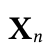}}}$
and any store \raisebox{-3.1874999999999982bp}{\makebox[5.471875000000001bp][l]{\includegraphics[trim=2.4000000000000004 2.4000000000000004 2.4000000000000004 2.4000000000000004]{pict_95.pdf}}} closing \raisebox{-3.1874999999999982bp}{\makebox[4.69375bp][l]{\includegraphics[trim=2.4000000000000004 2.4000000000000004 2.4000000000000004 2.4000000000000004]{pict_93.pdf}}}, if
\raisebox{-2.6179687499999993bp}{\makebox[137.61796875bp][l]{\includegraphics[trim=2.4000000000000004 2.4000000000000004 2.4000000000000004 2.4000000000000004]{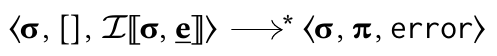}}}
then
\raisebox{-3.3617187499999996bp}{\makebox[131.48828125bp][l]{\includegraphics[trim=2.4000000000000004 2.4000000000000004 2.4000000000000004 2.4000000000000004]{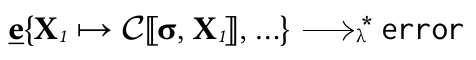}}}.
\end{theorem}

Completeness captures that if a user program has a bug, our concolic
tester finds it through the iterative evolution of initially default
inputs. Formally, the \nameref{correct:complete} theorem
assumes that, for some inputs, the user program returns \raisebox{-3.1874999999999982bp}{\makebox[24.0bp][l]{\includegraphics[trim=2.4000000000000004 2.4000000000000004 2.4000000000000004 2.4000000000000004]{pict_4.pdf}}} and
has five
conclusions that describe that (1) the concolic loop starts with a store
that contains inputs
for all concolic variables in a user program; (2) these initial inputs
  are arbitrary numbers or the default functions (3) each iteration \raisebox{-3.1874999999999982bp}{\makebox[3.0906249999999993bp][l]{\includegraphics[trim=2.4000000000000004 2.4000000000000004 2.4000000000000004 2.4000000000000004]{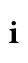}}}
  results in a pair \raisebox{-2.3617187499999996bp}{\makebox[26.296093750000004bp][l]{\includegraphics[trim=2.4000000000000004 2.4000000000000004 2.4000000000000004 2.4000000000000004]{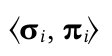}}} of a store and a list of path
  constraints; (4) after the end of iteration \raisebox{-3.1874999999999982bp}{\makebox[3.0906249999999993bp][l]{\includegraphics[trim=2.4000000000000004 2.4000000000000004 2.4000000000000004 2.4000000000000004]{pict_178.pdf}}} that does not trigger
  an \raisebox{-3.1874999999999982bp}{\makebox[24.0bp][l]{\includegraphics[trim=2.4000000000000004 2.4000000000000004 2.4000000000000004 2.4000000000000004]{pict_4.pdf}}}, \raisebox{-3.1874999999999982bp}{\makebox[23.3828125bp][l]{\includegraphics[trim=2.4000000000000004 2.4000000000000004 2.4000000000000004 2.4000000000000004]{pict_100.pdf}}} uses \raisebox{-2.3617187499999996bp}{\makebox[26.296093750000004bp][l]{\includegraphics[trim=2.4000000000000004 2.4000000000000004 2.4000000000000004 2.4000000000000004]{pict_179.pdf}}}
  to generate \raisebox{-2.3617187499999996bp}{\makebox[37.169531250000006bp][l]{\includegraphics[trim=2.4000000000000004 2.4000000000000004 2.4000000000000004 2.4000000000000004]{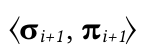}}}; and (5) the concolic
  loop terminates with the discovery of the \raisebox{-3.1874999999999982bp}{\makebox[24.0bp][l]{\includegraphics[trim=2.4000000000000004 2.4000000000000004 2.4000000000000004 2.4000000000000004]{pict_4.pdf}}}.
  An interesting point about the fourth conclusion is that it establishes
  that the list of path constraints \raisebox{-2.3617187499999996bp}{\makebox[12.66484375bp][l]{\includegraphics[trim=2.4000000000000004 2.4000000000000004 2.4000000000000004 2.4000000000000004]{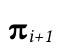}}} that \raisebox{-3.1874999999999982bp}{\makebox[23.3828125bp][l]{\includegraphics[trim=2.4000000000000004 2.4000000000000004 2.4000000000000004 2.4000000000000004]{pict_100.pdf}}} returns
   at the end of each iteration \raisebox{-3.1874999999999982bp}{\makebox[3.0906249999999993bp][l]{\includegraphics[trim=2.4000000000000004 2.4000000000000004 2.4000000000000004 2.4000000000000004]{pict_178.pdf}}} is
  equivalent to the actual list that  iteration \raisebox{-3.1874999999999982bp}{\makebox[14.400000000000004bp][l]{\includegraphics[trim=2.4000000000000004 2.4000000000000004 2.4000000000000004 2.4000000000000004]{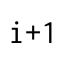}}} of the loop
  produces. Here, two lists of path constraints are \emph{equivalent} if they have
the same first{-}order and branch constraints.

\begin{theorem}[Completeness]\label{correct:complete}
For any \raisebox{-3.1874999999999982bp}{\makebox[4.69375bp][l]{\includegraphics[trim=2.4000000000000004 2.4000000000000004 2.4000000000000004 2.4000000000000004]{pict_93.pdf}}} with concolic variables
$\raisebox{-2.3617187499999996bp}{\makebox[9.378125bp][l]{\includegraphics[trim=2.4000000000000004 2.4000000000000004 2.4000000000000004 2.4000000000000004]{pict_174.pdf}}},\dots,\raisebox{-2.3617187499999996bp}{\makebox[9.79296875bp][l]{\includegraphics[trim=2.4000000000000004 2.4000000000000004 2.4000000000000004 2.4000000000000004]{pict_175.pdf}}}$,
if there exists closed values
$\raisebox{-2.3617187499999996bp}{\makebox[7.5640625bp][l]{\includegraphics[trim=2.4000000000000004 2.4000000000000004 2.4000000000000004 2.4000000000000004]{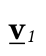}}},\dots,\raisebox{-2.3617187499999996bp}{\makebox[7.97890625bp][l]{\includegraphics[trim=2.4000000000000004 2.4000000000000004 2.4000000000000004 2.4000000000000004]{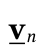}}}$ in the language of user programs
such that none of the values contain \raisebox{-3.1874999999999982bp}{\makebox[24.0bp][l]{\includegraphics[trim=2.4000000000000004 2.4000000000000004 2.4000000000000004 2.4000000000000004]{pict_4.pdf}}} and
\raisebox{-3.3617187499999996bp}{\makebox[104.81796874999999bp][l]{\includegraphics[trim=2.4000000000000004 2.4000000000000004 2.4000000000000004 2.4000000000000004]{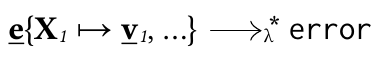}}}
then there exists a sequence of stores and paths
$\raisebox{-2.3617187499999996bp}{\makebox[28.177343750000006bp][l]{\includegraphics[trim=2.4000000000000004 2.4000000000000004 2.4000000000000004 2.4000000000000004]{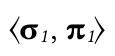}}}, \dots, \raisebox{-2.3617187499999996bp}{\makebox[31.974218750000006bp][l]{\includegraphics[trim=2.4000000000000004 2.4000000000000004 2.4000000000000004 2.4000000000000004]{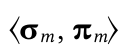}}}$ such that

\noindent \begin{enumerate}\atItemizeStart

\item $\text{dom}(\raisebox{-2.3617187499999996bp}{\makebox[7.957812499999999bp][l]{\includegraphics[trim=2.4000000000000004 2.4000000000000004 2.4000000000000004 2.4000000000000004]{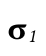}}})=\{\raisebox{-2.3617187499999996bp}{\makebox[9.378125bp][l]{\includegraphics[trim=2.4000000000000004 2.4000000000000004 2.4000000000000004 2.4000000000000004]{pict_174.pdf}}},\dots,\raisebox{-2.3617187499999996bp}{\makebox[9.79296875bp][l]{\includegraphics[trim=2.4000000000000004 2.4000000000000004 2.4000000000000004 2.4000000000000004]{pict_175.pdf}}}\}$,

\item For all $1\le k\le n$, either
\raisebox{-2.3617187499999996bp}{\makebox[37.83437500000001bp][l]{\includegraphics[trim=2.4000000000000004 2.4000000000000004 2.4000000000000004 2.4000000000000004]{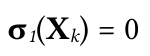}}} or
\raisebox{-3.2562499999999996bp}{\makebox[110.6625bp][l]{\includegraphics[trim=2.4000000000000004 2.4000000000000004 2.4000000000000004 2.4000000000000004]{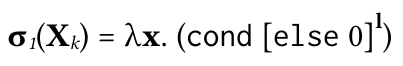}}}.

\item For all $1\le i< m$,
\raisebox{-3.3617187499999996bp}{\makebox[126.03828124999998bp][l]{\includegraphics[trim=2.4000000000000004 2.4000000000000004 2.4000000000000004 2.4000000000000004]{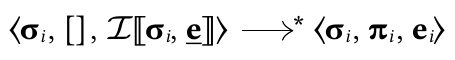}}}.

\item For all $1\le i< m$, there exists a pair
\raisebox{-2.9390624999999986bp}{\makebox[102.33749999999999bp][l]{\includegraphics[trim=2.4000000000000004 2.4000000000000004 2.4000000000000004 2.4000000000000004]{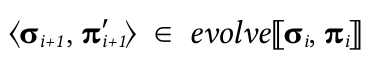}}} such that \raisebox{-2.3617187499999996bp}{\makebox[12.66484375bp][l]{\includegraphics[trim=2.4000000000000004 2.4000000000000004 2.4000000000000004 2.4000000000000004]{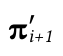}}} is equivalent to
a prefix of \raisebox{-2.3617187499999996bp}{\makebox[12.66484375bp][l]{\includegraphics[trim=2.4000000000000004 2.4000000000000004 2.4000000000000004 2.4000000000000004]{pict_181.pdf}}}.

\item \raisebox{-3.3617187499999996bp}{\makebox[155.15546875bp][l]{\includegraphics[trim=2.4000000000000004 2.4000000000000004 2.4000000000000004 2.4000000000000004]{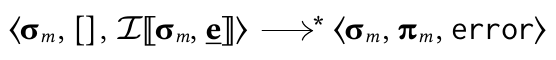}}}.\end{enumerate}

\noindent \end{theorem}

We prove the \nameref{correct:complete} theorem in two steps.
First, we show that if there is an \raisebox{-3.1874999999999982bp}{\makebox[24.0bp][l]{\includegraphics[trim=2.4000000000000004 2.4000000000000004 2.4000000000000004 2.4000000000000004]{pict_4.pdf}}} in the user program that an input
can trigger, there exists a store \raisebox{-3.1874999999999982bp}{\makebox[5.471875000000001bp][l]{\includegraphics[trim=2.4000000000000004 2.4000000000000004 2.4000000000000004 2.4000000000000004]{pict_95.pdf}}} that contains numbers and canonical functions
that also causes the \raisebox{-3.1874999999999982bp}{\makebox[24.0bp][l]{\includegraphics[trim=2.4000000000000004 2.4000000000000004 2.4000000000000004 2.4000000000000004]{pict_4.pdf}}} to manifest. Thus this step validates the
definition of canonical functions.

\begin{lemma}[Representation Completeness]\label{correct:repncomp}
Let \raisebox{-3.1874999999999982bp}{\makebox[23.20546875bp][l]{\includegraphics[trim=2.4000000000000004 2.4000000000000004 2.4000000000000004 2.4000000000000004]{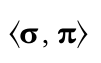}}} be a \emph{proper counterexample} for
a user program \raisebox{-3.1874999999999982bp}{\makebox[4.69375bp][l]{\includegraphics[trim=2.4000000000000004 2.4000000000000004 2.4000000000000004 2.4000000000000004]{pict_93.pdf}}} if $\mathit{FV}(\raisebox{-3.1874999999999982bp}{\makebox[4.69375bp][l]{\includegraphics[trim=2.4000000000000004 2.4000000000000004 2.4000000000000004 2.4000000000000004]{pict_93.pdf}}})\subset\text{dom}(\raisebox{-3.1874999999999982bp}{\makebox[5.471875000000001bp][l]{\includegraphics[trim=2.4000000000000004 2.4000000000000004 2.4000000000000004 2.4000000000000004]{pict_95.pdf}}})$,
\raisebox{-2.6179687499999993bp}{\makebox[137.61796875bp][l]{\includegraphics[trim=2.4000000000000004 2.4000000000000004 2.4000000000000004 2.4000000000000004]{pict_176.pdf}}}
and \raisebox{-3.1874999999999982bp}{\makebox[5.682812500000001bp][l]{\includegraphics[trim=2.4000000000000004 2.4000000000000004 2.4000000000000004 2.4000000000000004]{pict_99.pdf}}} does not contain branch constraints from \raisebox{-3.1874999999999982bp}{\makebox[19.200000000000003bp][l]{\includegraphics[trim=2.4000000000000004 2.4000000000000004 2.4000000000000004 2.4000000000000004]{pict_50.pdf}}} branches.
For any \raisebox{-3.1874999999999982bp}{\makebox[4.69375bp][l]{\includegraphics[trim=2.4000000000000004 2.4000000000000004 2.4000000000000004 2.4000000000000004]{pict_93.pdf}}} with concolic variables
$\raisebox{-2.3617187499999996bp}{\makebox[9.378125bp][l]{\includegraphics[trim=2.4000000000000004 2.4000000000000004 2.4000000000000004 2.4000000000000004]{pict_174.pdf}}},\dots,\raisebox{-2.3617187499999996bp}{\makebox[9.79296875bp][l]{\includegraphics[trim=2.4000000000000004 2.4000000000000004 2.4000000000000004 2.4000000000000004]{pict_175.pdf}}}$,
if there exists closed values
$\raisebox{-2.3617187499999996bp}{\makebox[7.5640625bp][l]{\includegraphics[trim=2.4000000000000004 2.4000000000000004 2.4000000000000004 2.4000000000000004]{pict_183.pdf}}},\dots,\raisebox{-2.3617187499999996bp}{\makebox[7.97890625bp][l]{\includegraphics[trim=2.4000000000000004 2.4000000000000004 2.4000000000000004 2.4000000000000004]{pict_184.pdf}}}$
such that no value contains \raisebox{-3.1874999999999982bp}{\makebox[24.0bp][l]{\includegraphics[trim=2.4000000000000004 2.4000000000000004 2.4000000000000004 2.4000000000000004]{pict_4.pdf}}} and
\raisebox{-3.3617187499999996bp}{\makebox[104.81796874999999bp][l]{\includegraphics[trim=2.4000000000000004 2.4000000000000004 2.4000000000000004 2.4000000000000004]{pict_185.pdf}}}
then there exists a store \raisebox{-3.1874999999999982bp}{\makebox[5.471875000000001bp][l]{\includegraphics[trim=2.4000000000000004 2.4000000000000004 2.4000000000000004 2.4000000000000004]{pict_95.pdf}}} and path constraints \raisebox{-3.1874999999999982bp}{\makebox[5.682812500000001bp][l]{\includegraphics[trim=2.4000000000000004 2.4000000000000004 2.4000000000000004 2.4000000000000004]{pict_99.pdf}}} such that
\raisebox{-3.1874999999999982bp}{\makebox[23.20546875bp][l]{\includegraphics[trim=2.4000000000000004 2.4000000000000004 2.4000000000000004 2.4000000000000004]{pict_195.pdf}}} is a proper counterexample of \raisebox{-3.1874999999999982bp}{\makebox[4.69375bp][l]{\includegraphics[trim=2.4000000000000004 2.4000000000000004 2.4000000000000004 2.4000000000000004]{pict_93.pdf}}}.
\end{lemma}

\begin{proof}[Proof Sketch]
To prove \autoref{correct:repncomp}, we define an intermediate
language that contains the union of the user language
and the concolic language. The intermediate
language also collects details about the interaction between
user programs and concolic functions
in a global map. Given \raisebox{-3.1874999999999982bp}{\makebox[4.69375bp][l]{\includegraphics[trim=2.4000000000000004 2.4000000000000004 2.4000000000000004 2.4000000000000004]{pict_93.pdf}}} and $\{\raisebox{-2.3617187499999996bp}{\makebox[6.6234375bp][l]{\includegraphics[trim=2.4000000000000004 2.4000000000000004 2.4000000000000004 2.4000000000000004]{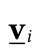}}}\}$,
we simulate the reduction sequence of \raisebox{-3.1874999999999982bp}{\makebox[4.69375bp][l]{\includegraphics[trim=2.4000000000000004 2.4000000000000004 2.4000000000000004 2.4000000000000004]{pict_93.pdf}}}
with the intermediate language and construct a store \raisebox{-3.1874999999999982bp}{\makebox[5.471875000000001bp][l]{\includegraphics[trim=2.4000000000000004 2.4000000000000004 2.4000000000000004 2.4000000000000004]{pict_95.pdf}}} using
the collected information from the global map.
\end{proof}

As the second step of the proof of \nameref{correct:complete},
we show that the evolution of inputs of the concolic loop produces an
input that can trigger the same \raisebox{-3.1874999999999982bp}{\makebox[24.0bp][l]{\includegraphics[trim=2.4000000000000004 2.4000000000000004 2.4000000000000004 2.4000000000000004]{pict_4.pdf}}} as the one an  arbitrary \raisebox{-3.1874999999999982bp}{\makebox[5.471875000000001bp][l]{\includegraphics[trim=2.4000000000000004 2.4000000000000004 2.4000000000000004 2.4000000000000004]{pict_95.pdf}}} triggers.
As a consequence, the concolic tester only needs to explore
inputs it obtains from \raisebox{-3.1874999999999982bp}{\makebox[23.3828125bp][l]{\includegraphics[trim=2.4000000000000004 2.4000000000000004 2.4000000000000004 2.4000000000000004]{pict_100.pdf}}}.

\begin{lemma}[Search Completeness]\label{correct:searchcomp}
For any \raisebox{-3.1874999999999982bp}{\makebox[4.69375bp][l]{\includegraphics[trim=2.4000000000000004 2.4000000000000004 2.4000000000000004 2.4000000000000004]{pict_93.pdf}}} with concolic variables
$\raisebox{-2.3617187499999996bp}{\makebox[9.378125bp][l]{\includegraphics[trim=2.4000000000000004 2.4000000000000004 2.4000000000000004 2.4000000000000004]{pict_174.pdf}}},\dots,\raisebox{-2.3617187499999996bp}{\makebox[9.79296875bp][l]{\includegraphics[trim=2.4000000000000004 2.4000000000000004 2.4000000000000004 2.4000000000000004]{pict_175.pdf}}}$,
if \raisebox{-3.1874999999999982bp}{\makebox[4.69375bp][l]{\includegraphics[trim=2.4000000000000004 2.4000000000000004 2.4000000000000004 2.4000000000000004]{pict_93.pdf}}} has a proper counterexample
then there exists a sequence of stores and paths
satisfying \autoref{correct:complete} (1){-}(5).
\end{lemma}

\begin{proof}[Proof Sketch]
Let \raisebox{-3.1874999999999982bp}{\makebox[23.20546875bp][l]{\includegraphics[trim=2.4000000000000004 2.4000000000000004 2.4000000000000004 2.4000000000000004]{pict_195.pdf}}} denote a proper counterexample of \raisebox{-3.1874999999999982bp}{\makebox[4.69375bp][l]{\includegraphics[trim=2.4000000000000004 2.4000000000000004 2.4000000000000004 2.4000000000000004]{pict_93.pdf}}}.
First, we show a property of the concolic loop. For any store
\raisebox{-2.3617187499999996bp}{\makebox[7.957812499999999bp][l]{\includegraphics[trim=2.4000000000000004 2.4000000000000004 2.4000000000000004 2.4000000000000004]{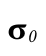}}}, the concolic evaluation of \raisebox{-3.1874999999999982bp}{\makebox[4.69375bp][l]{\includegraphics[trim=2.4000000000000004 2.4000000000000004 2.4000000000000004 2.4000000000000004]{pict_93.pdf}}} with inputs  \raisebox{-2.3617187499999996bp}{\makebox[7.957812499999999bp][l]{\includegraphics[trim=2.4000000000000004 2.4000000000000004 2.4000000000000004 2.4000000000000004]{pict_197.pdf}}}
 either follows the same control{-}flow path as the evaluation
of \raisebox{-3.1874999999999982bp}{\makebox[4.69375bp][l]{\includegraphics[trim=2.4000000000000004 2.4000000000000004 2.4000000000000004 2.4000000000000004]{pict_93.pdf}}} with inputs \raisebox{-3.1874999999999982bp}{\makebox[5.471875000000001bp][l]{\includegraphics[trim=2.4000000000000004 2.4000000000000004 2.4000000000000004 2.4000000000000004]{pict_95.pdf}}} or takes a different branch at some point.
By applying \raisebox{-3.1874999999999982bp}{\makebox[23.3828125bp][l]{\includegraphics[trim=2.4000000000000004 2.4000000000000004 2.4000000000000004 2.4000000000000004]{pict_100.pdf}}} to \raisebox{-2.3617187499999996bp}{\makebox[7.957812499999999bp][l]{\includegraphics[trim=2.4000000000000004 2.4000000000000004 2.4000000000000004 2.4000000000000004]{pict_197.pdf}}} and the list of path constraints,
we can obtain \raisebox{-2.3617187499999996bp}{\makebox[7.957812499999999bp][l]{\includegraphics[trim=2.4000000000000004 2.4000000000000004 2.4000000000000004 2.4000000000000004]{pict_188.pdf}}} such that
the evaluation of \raisebox{-3.1874999999999982bp}{\makebox[4.69375bp][l]{\includegraphics[trim=2.4000000000000004 2.4000000000000004 2.4000000000000004 2.4000000000000004]{pict_93.pdf}}} with inputs  \raisebox{-2.3617187499999996bp}{\makebox[7.957812499999999bp][l]{\includegraphics[trim=2.4000000000000004 2.4000000000000004 2.4000000000000004 2.4000000000000004]{pict_188.pdf}}} follows a control{-}flow path
that has one more branch in common with the evaluation of \raisebox{-3.1874999999999982bp}{\makebox[4.69375bp][l]{\includegraphics[trim=2.4000000000000004 2.4000000000000004 2.4000000000000004 2.4000000000000004]{pict_93.pdf}}} with inputs  \raisebox{-3.1874999999999982bp}{\makebox[5.471875000000001bp][l]{\includegraphics[trim=2.4000000000000004 2.4000000000000004 2.4000000000000004 2.4000000000000004]{pict_95.pdf}}}.
With this property in hand, we start with a store containing
default canonical functions (the ones with
the simplest shape) and we use the property to construct
a sequence of stores that gradually approximate \raisebox{-3.1874999999999982bp}{\makebox[5.471875000000001bp][l]{\includegraphics[trim=2.4000000000000004 2.4000000000000004 2.4000000000000004 2.4000000000000004]{pict_95.pdf}}}
until one of them triggers the \raisebox{-3.1874999999999982bp}{\makebox[24.0bp][l]{\includegraphics[trim=2.4000000000000004 2.4000000000000004 2.4000000000000004 2.4000000000000004]{pict_4.pdf}}} in \raisebox{-3.1874999999999982bp}{\makebox[4.69375bp][l]{\includegraphics[trim=2.4000000000000004 2.4000000000000004 2.4000000000000004 2.4000000000000004]{pict_93.pdf}}}.
\end{proof}

The last fact we establish for our concolic tester
is necessary for the proof of \autoref{correct:searchcomp}
but also has value on its own. It entails that, at each
iteration of its loop, the concolic tester aims to explore a specific control{-}flow path of the
user program and indeed produces new inputs that achieve this goal. We call this property the
\emph{concolic property}. Formally,  \autoref{search:predict} shows that
evaluating the user program with the store constructed by
\raisebox{-3.1874999999999982bp}{\makebox[23.3828125bp][l]{\includegraphics[trim=2.4000000000000004 2.4000000000000004 2.4000000000000004 2.4000000000000004]{pict_100.pdf}}} follows the control{-}flow path that \raisebox{-3.1874999999999982bp}{\makebox[23.3828125bp][l]{\includegraphics[trim=2.4000000000000004 2.4000000000000004 2.4000000000000004 2.4000000000000004]{pict_100.pdf}}} predicts with the
list of path constraints it returns along the store.

\begin{theorem}[Concolic]\label{search:predict}
Let the non{-}terminal \raisebox{-3.1874999999999982bp}{\makebox[9.67578125bp][l]{\includegraphics[trim=2.4000000000000004 2.4000000000000004 2.4000000000000004 2.4000000000000004]{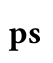}}} denote the subset of \raisebox{-3.1874999999999982bp}{\makebox[5.682812500000001bp][l]{\includegraphics[trim=2.4000000000000004 2.4000000000000004 2.4000000000000004 2.4000000000000004]{pict_99.pdf}}}
that corresponds to a first{-}order constraint or
a test constraint together with a block of branch constraints.
 For any \raisebox{-3.1874999999999982bp}{\makebox[4.69375bp][l]{\includegraphics[trim=2.4000000000000004 2.4000000000000004 2.4000000000000004 2.4000000000000004]{pict_93.pdf}}} and \raisebox{-2.3617187499999996bp}{\makebox[7.957812499999999bp][l]{\includegraphics[trim=2.4000000000000004 2.4000000000000004 2.4000000000000004 2.4000000000000004]{pict_188.pdf}}}, if

\noindent \begin{enumerate}\atItemizeStart

\item \raisebox{-3.3617187499999996bp}{\makebox[152.950625bp][l]{\includegraphics[trim=2.4000000000000004 2.4000000000000004 2.4000000000000004 2.4000000000000004]{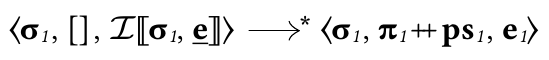}}}

\item \raisebox{-2.3617187499999996bp}{\makebox[8.16875bp][l]{\includegraphics[trim=2.4000000000000004 2.4000000000000004 2.4000000000000004 2.4000000000000004]{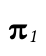}}} has no branch constraints corresponding to
\raisebox{-3.1874999999999982bp}{\makebox[19.200000000000003bp][l]{\includegraphics[trim=2.4000000000000004 2.4000000000000004 2.4000000000000004 2.4000000000000004]{pict_50.pdf}}} in the canonical functions in \raisebox{-2.3617187499999996bp}{\makebox[7.957812499999999bp][l]{\includegraphics[trim=2.4000000000000004 2.4000000000000004 2.4000000000000004 2.4000000000000004]{pict_188.pdf}}}.

\item \raisebox{-2.9390624999999986bp}{\makebox[137.15906249999998bp][l]{\includegraphics[trim=2.4000000000000004 2.4000000000000004 2.4000000000000004 2.4000000000000004]{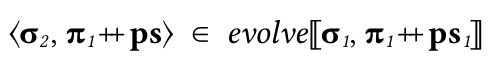}}}\end{enumerate}

\noindent then
\raisebox{-3.3617187499999996bp}{\makebox[152.950625bp][l]{\includegraphics[trim=2.4000000000000004 2.4000000000000004 2.4000000000000004 2.4000000000000004]{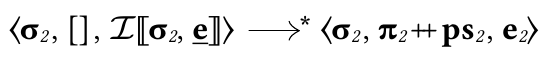}}}
such that \raisebox{-2.3617187499999996bp}{\makebox[27.89203125000001bp][l]{\includegraphics[trim=2.4000000000000004 2.4000000000000004 2.4000000000000004 2.4000000000000004]{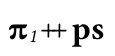}}} is equivalent to \raisebox{-2.3617187499999996bp}{\makebox[30.377968750000008bp][l]{\includegraphics[trim=2.4000000000000004 2.4000000000000004 2.4000000000000004 2.4000000000000004]{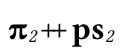}}}.
\end{theorem}

\begin{proof}[Proof Sketch]
We show by simulation that there exists \raisebox{-2.3617187499999996bp}{\makebox[7.1796875bp][l]{\includegraphics[trim=2.4000000000000004 2.4000000000000004 2.4000000000000004 2.4000000000000004]{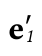}}}, \raisebox{-2.3617187499999996bp}{\makebox[7.1796875bp][l]{\includegraphics[trim=2.4000000000000004 2.4000000000000004 2.4000000000000004 2.4000000000000004]{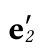}}} and \raisebox{-2.3617187499999996bp}{\makebox[8.16875bp][l]{\includegraphics[trim=2.4000000000000004 2.4000000000000004 2.4000000000000004 2.4000000000000004]{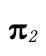}}} such that

\noindent \begin{itemize}\atItemizeStart

\item \raisebox{-2.3617187499999996bp}{\makebox[8.16875bp][l]{\includegraphics[trim=2.4000000000000004 2.4000000000000004 2.4000000000000004 2.4000000000000004]{pict_200.pdf}}} is equivalent to \raisebox{-2.3617187499999996bp}{\makebox[8.16875bp][l]{\includegraphics[trim=2.4000000000000004 2.4000000000000004 2.4000000000000004 2.4000000000000004]{pict_207.pdf}}},

\item \raisebox{-3.3617187499999996bp}{\makebox[218.34437499999999bp][l]{\includegraphics[trim=2.4000000000000004 2.4000000000000004 2.4000000000000004 2.4000000000000004]{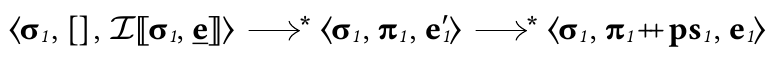}}} and

\item \raisebox{-3.3617187499999996bp}{\makebox[130.74140624999998bp][l]{\includegraphics[trim=2.4000000000000004 2.4000000000000004 2.4000000000000004 2.4000000000000004]{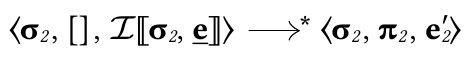}}}\end{itemize}

\noindent Since \raisebox{-2.3617187499999996bp}{\makebox[7.957812499999999bp][l]{\includegraphics[trim=2.4000000000000004 2.4000000000000004 2.4000000000000004 2.4000000000000004]{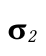}}} makes the test of the branch from \raisebox{-3.1874999999999982bp}{\makebox[9.67578125bp][l]{\includegraphics[trim=2.4000000000000004 2.4000000000000004 2.4000000000000004 2.4000000000000004]{pict_198.pdf}}} succeed
and \raisebox{-2.3617187499999996bp}{\makebox[7.1796875bp][l]{\includegraphics[trim=2.4000000000000004 2.4000000000000004 2.4000000000000004 2.4000000000000004]{pict_205.pdf}}} and \raisebox{-2.3617187499999996bp}{\makebox[7.1796875bp][l]{\includegraphics[trim=2.4000000000000004 2.4000000000000004 2.4000000000000004 2.4000000000000004]{pict_206.pdf}}} are built using simulation, we establish that
\raisebox{-3.3617187499999996bp}{\makebox[128.65453125bp][l]{\includegraphics[trim=2.4000000000000004 2.4000000000000004 2.4000000000000004 2.4000000000000004]{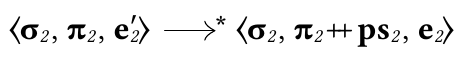}}}
where \raisebox{-2.3617187499999996bp}{\makebox[12.161718749999999bp][l]{\includegraphics[trim=2.4000000000000004 2.4000000000000004 2.4000000000000004 2.4000000000000004]{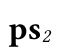}}} is equivalent to \raisebox{-3.1874999999999982bp}{\makebox[9.67578125bp][l]{\includegraphics[trim=2.4000000000000004 2.4000000000000004 2.4000000000000004 2.4000000000000004]{pict_198.pdf}}}.
\end{proof}

\sectionNewpage

\Ssection{Prototype Implementation}{Prototype Implementation}\label{t:x28part_x22secx3aprototypex22x29}

To evaluate whether our higher{-}order concolic tester
can find bugs in practice, we
built a prototype implementation
and used it with higher{-}order programs of our own devising and programs from the literature. In
\SecRef{\SectionNumberLink{t:x28part_x22secx3aprototypex2darchitecturex22x29}{7.1}}{Implementation Architecture} we describe the
prototype implementation and in \SecRef{\SectionNumberLink{t:x28part_x22secx3abenchmarkx22x29}{7.2}}{Benchmark Programs} we
report on the programs we tried and how well the implementation can
find the bugs they contain.

\Ssubsection{Implementation Architecture}{Implementation Architecture}\label{t:x28part_x22secx3aprototypex2darchitecturex22x29}

At the heart of our prototype is a language built on top of
Rosette\Autobibref{~(\hyperref[t:x28autobib_x22Emina_Torlak_and_Rastislav_BodikGrowing_Solverx2daided_Languages_with_RosetteIn_Procx2e_International_Symposium_on_New_Ideasx2c_New_Paradigmsx2c_and_Reflections_on_Programming_and_Softwarex2c_Onwardx21x2c_ppx2e_135x2dx2d1522013x22x29]{\AutobibLink{Torlak and Bodik}} \hyperref[t:x28autobib_x22Emina_Torlak_and_Rastislav_BodikGrowing_Solverx2daided_Languages_with_RosetteIn_Procx2e_International_Symposium_on_New_Ideasx2c_New_Paradigmsx2c_and_Reflections_on_Programming_and_Softwarex2c_Onwardx21x2c_ppx2e_135x2dx2d1522013x22x29]{\AutobibLink{2013}}, \hyperref[t:x28autobib_x22Emina_Torlak_and_Rastislav_BodikA_Lightweight_Symbolic_Virtual_Machine_for_Solverx2dAided_Host_LanguagesIn_Procx2e_ACM_Conference_on_Programming_Language_Design_and_Implementationx2c_ppx2e_530x2dx2d5412014x22x29]{\AutobibLink{2014}})}, a framework for constructing
domain{-}specific languages that employ SAT solvers. We start with most of
the features from Rosette itself, but use Racket{'}s domain{-}specific
language support\Autobibref{~(\hyperref[t:x28autobib_x22Matthias_Felleisenx2c_Robert_Bx2e_Findlerx2c_Matthew_Flattx2c_Shriram_Krishnamurthix2c_Eli_Barzilayx2c_Jay_McCarthyx2c_and_Sam_Tobinx2dHochstadtA_Programmable_Programming_LanguageCommunications_of_the_ACM_61x2c_ppx2e_3x3a62x2dx2d3x3a712018x22x29]{\AutobibLink{Felleisen et al\Sendabbrev{.}}} \hyperref[t:x28autobib_x22Matthias_Felleisenx2c_Robert_Bx2e_Findlerx2c_Matthew_Flattx2c_Shriram_Krishnamurthix2c_Eli_Barzilayx2c_Jay_McCarthyx2c_and_Sam_Tobinx2dHochstadtA_Programmable_Programming_LanguageCommunications_of_the_ACM_61x2c_ppx2e_3x3a62x2dx2d3x3a712018x22x29]{\AutobibLink{2018}}; \hyperref[t:x28autobib_x22Vincent_Stx2dAmourx2c_Daniel_Felteyx2c_Spencer_Px2e_Florencex2c_Shux2dHung_Youx2c_and_Robert_Bx2e_FindlerHerbarium_Racketensisx3a_A_Stroll_Through_the_WoodsProceedings_of_the_ACM_on_Programming_Languages_x28ICFPx29_1x2c_ppx2e_1x3a1x2dx2d1x3a152017x22x29]{\AutobibLink{St{-}Amour et al\Sendabbrev{.}}} \hyperref[t:x28autobib_x22Vincent_Stx2dAmourx2c_Daniel_Felteyx2c_Spencer_Px2e_Florencex2c_Shux2dHung_Youx2c_and_Robert_Bx2e_FindlerHerbarium_Racketensisx3a_A_Stroll_Through_the_WoodsProceedings_of_the_ACM_on_Programming_Languages_x28ICFPx29_1x2c_ppx2e_1x3a1x2dx2d1x3a152017x22x29]{\AutobibLink{2017}})} to adjust Rosette
to concolic evaluation. More specifically, we replace Rosette{'}s
conditional expressions (that employ symbolic execution) with ones that
compute specific values at conditional tests and in order to avoid exploring both
branches of any given conditional. These replacement conditional
expressions also record path constraints.  We use this language for the
concolic evaluation of the programs we aim to test.

To properly test these programs we also need to be able to
construct higher{-}order inputs (as described in
\ChapRef{\SectionNumberLink{t:x28part_x22secx3aapproximatex22x29}{3}}{Canonical Functions Are All We Need} and
\ChapRef{\SectionNumberLink{t:x28part_x22secx3ahowx2dgeneratex22x29}{4}}{Directed Evolution of Canonical Functions}). To do so, we have designed a data
structure that closely follows the grammar for
canonical forms. We interpret that data structure, producing
higher{-}order functions that we supply as inputs to the user program,
and we compute new instances of the data structure based on
the results of calling the SMT solver, leveraging Rosette.
Thus, in essence the canonical forms in the prototype have their own small
language. This language is slightly more
general than the description of canonical forms in \ChapRef{\SectionNumberLink{t:x28part_x22secx3amodelx22x29}{5}}{Formalizing Higher{-}Order Concolic Testing}. Specifically, in addition to unary
functions and integers, it also supports lists and booleans.

The programs we set to test, however, use some fairly
sophisticated features of Racket that the concolic language
in our prototype does not support. To bridge the gap, we
added a number of libraries. First, we implemented a
rudimentary contract system and a rudimentary complex number
library in our concolic adaptation of Rosette. We also
implemented a library that provides conversion wrappers to
adjust values flowing in and out of the user program.
It supports several forms of conversions: it can adjust
curried functions to appear as n{-}ary functions; it can
adjust a list of functions to appear as an object,
and it also supports a {``}lump{''} conversion
where integers are translated back and forth to a specific
set of interesting constants (typically symbols and
strings). This encoding of inputs also enables us to keep the language of the
prototype{'}s canonical forms small (thereby making the
evolution of new inputs simpler) while still being able
to find bugs in programs that use some of Racket{'}s more
sophisticated features.

As the formal model does not specify an explicit search strategy,
 the prototype comes with a very naive search strategy.
Specifically, we simply follow a breadth{-}first approach.
After each iteration of the concolic loop, we add all possible next
evolutions of the input (following \raisebox{-3.1874999999999982bp}{\makebox[47.208593750000006bp][l]{\includegraphics[trim=2.4000000000000004 2.4000000000000004 2.4000000000000004 2.4000000000000004]{pict_156.pdf}}}, from
figure~\hyperref[t:x28counter_x28x22figurex22_x22figx3amutate1x22x29x29]{\FigureRef{8}{t:x28counter_x28x22figurex22_x22figx3amutate1x22x29x29}}, figure~\hyperref[t:x28counter_x28x22figurex22_x22figx3amutate2x22x29x29]{\FigureRef{9}{t:x28counter_x28x22figurex22_x22figx3amutate2x22x29x29}}, and
figure~\hyperref[t:x28counter_x28x22figurex22_x22figx3amutate3x22x29x29]{\FigureRef{10}{t:x28counter_x28x22figurex22_x22figx3amutate3x22x29x29}}) in a queue to continue the
search.

All together, our prototype is a bit more than 7,000 lines
of code (with a bit less than 7,000 lines of code in the
prototype{'}s test suites).

\Ssubsection{Benchmark Programs}{Benchmark Programs}\label{t:x28part_x22secx3abenchmarkx22x29}

Our benchmark programs come from three sources. The first,
and primary source, is \Autobibref{\hyperref[t:x28autobib_x22Phxfac_Nguyx1ec5nx2c_Sam_Tobinx2dHochstadtx2c_and_David_Van_HornHigher_order_symbolic_execution_for_contract_verification_and_refutationx2eJournal_of_Functional_Programmingx2827x29x2c_ppx2e_e3x3a1x2dx2de3x3a542017x22x29]{\AutobibLink{Nguy\~{\^{e}}n et al\Sendabbrev{.}}}~(\hyperref[t:x28autobib_x22Phxfac_Nguyx1ec5nx2c_Sam_Tobinx2dHochstadtx2c_and_David_Van_HornHigher_order_symbolic_execution_for_contract_verification_and_refutationx2eJournal_of_Functional_Programmingx2827x29x2c_ppx2e_e3x3a1x2dx2de3x3a542017x22x29]{\AutobibLink{2017}})}{'}s work,
specifically from the \Scribtexttt{jfp} branch of
\href{https://github.com/philnguyen/soft-contract}{\Snolinkurl{https://github.com/philnguyen/soft-contract}}. These programs
ultimately come from other papers, as cited in
figure~\hyperref[t:x28counter_x28x22figurex22_x22figx3abenchmarkx2dresultsx22x29x29]{\FigureRef{11}{t:x28counter_x28x22figurex22_x22figx3abenchmarkx2dresultsx22x29x29}}. The second source is
CutEr\Autobibref{~(\hyperref[t:x28autobib_x22Aggelos_Giantsiosx2c_Nikolaos_Papaspyroux2c_and_Konstantinos_SagonasConcolic_Testing_for_Functional_LanguagesScience_of_Computer_Programmingx2c_ppx2e_109x2dx2d1342017x22x29]{\AutobibLink{Giantsios et al\Sendabbrev{.}}} \hyperref[t:x28autobib_x22Aggelos_Giantsiosx2c_Nikolaos_Papaspyroux2c_and_Konstantinos_SagonasConcolic_Testing_for_Functional_LanguagesScience_of_Computer_Programmingx2c_ppx2e_109x2dx2d1342017x22x29]{\AutobibLink{2017}})}, the tool for concolic testing functional
programs in Erlang. We collected
all of the test cases in CutEr{'}s test suite that seem to use
higher{-}order functions and translated them to use our
prototype. The last source are small examples that we
invented as part of this work; they are in the supplementary
material.

\begin{Figure}\begin{Centerfigure}\begin{FigureInside}\begin{SCentered}\begin{bigtabular}{@{\bigtableleftpad}l@{}l@{}r@{}r@{}r@{}r@{}l@{}l@{}l@{}l@{}l@{}l@{}l@{}}
\hbox{Name} &
\hbox{\mbox{\hphantom{\Scribtexttt{xx}}}} &
\hbox{Programs} &
\hbox{\mbox{\hphantom{\Scribtexttt{xx}}}} &
\hbox{Failures} &
\hbox{\mbox{\hphantom{\Scribtexttt{xx}}}} &
\multicolumn{7}{l}{\hbox{Source}} \\
\hline \hbox{games} &
\hbox{\mbox{\hphantom{\Scribtexttt{xx}}}} &
\hbox{3} &
\hbox{\mbox{\hphantom{\Scribtexttt{xx}}}} &
\hbox{3} &
\hbox{\mbox{\hphantom{\Scribtexttt{xx}}}} &
\multicolumn{7}{l}{\begin{minipage}[t]{0.5384615384615384\linewidth}
\Autobibref{\hyperref[t:x28autobib_x22Phxfac_Nguyx1ec5nx2c_Sam_Tobinx2dHochstadtx2c_and_David_Van_HornHigher_order_symbolic_execution_for_contract_verification_and_refutationx2eJournal_of_Functional_Programmingx2827x29x2c_ppx2e_e3x3a1x2dx2de3x3a542017x22x29]{\AutobibLink{Nguy\~{\^{e}}n et al\Sendabbrev{.}}}~(\hyperref[t:x28autobib_x22Phxfac_Nguyx1ec5nx2c_Sam_Tobinx2dHochstadtx2c_and_David_Van_HornHigher_order_symbolic_execution_for_contract_verification_and_refutationx2eJournal_of_Functional_Programmingx2827x29x2c_ppx2e_e3x3a1x2dx2de3x3a542017x22x29]{\AutobibLink{2017}})}, \Autobibref{\hyperref[t:x28autobib_x22Phxfac_Nguyx1ec5nx2c_Sam_Tobinx2dHochstadtx2c_and_David_Van_HornRelatively_complete_counterexamples_for_higherx2dorder_programsx2eIn_Procx2e_ACM_Conference_on_Programming_Language_Design_and_Implementationx2c_ppx2e_446x2dx2d4562015x22x29]{\AutobibLink{Nguy\~{\^{e}}n et al\Sendabbrev{.}}}~(\hyperref[t:x28autobib_x22Phxfac_Nguyx1ec5nx2c_Sam_Tobinx2dHochstadtx2c_and_David_Van_HornRelatively_complete_counterexamples_for_higherx2dorder_programsx2eIn_Procx2e_ACM_Conference_on_Programming_Language_Design_and_Implementationx2c_ppx2e_446x2dx2d4562015x22x29]{\AutobibLink{2015}})} \end{minipage}
} \\
\hbox{hors} &
\hbox{\mbox{\hphantom{\Scribtexttt{xx}}}} &
\hbox{23} &
\hbox{\mbox{\hphantom{\Scribtexttt{xx}}}} &
\hbox{1} &
\hbox{\mbox{\hphantom{\Scribtexttt{xx}}}} &
\multicolumn{7}{l}{\begin{minipage}[t]{0.5384615384615384\linewidth}
\Autobibref{\hyperref[t:x28autobib_x22Naoki_Kobayashix2c_Ryosuke_Satox2c_and_Hiroshi_UnnoPredicate_Abstraction_and_CEGAR_for_Higherx2dOrder_Model_CheckingIn_Procx2e_ACM_Conference_on_Programming_Language_Design_and_Implementationx2c_ppx2e_222x2dx2d2332011x22x29]{\AutobibLink{Kobayashi et al\Sendabbrev{.}}}~(\hyperref[t:x28autobib_x22Naoki_Kobayashix2c_Ryosuke_Satox2c_and_Hiroshi_UnnoPredicate_Abstraction_and_CEGAR_for_Higherx2dOrder_Model_CheckingIn_Procx2e_ACM_Conference_on_Programming_Language_Design_and_Implementationx2c_ppx2e_222x2dx2d2332011x22x29]{\AutobibLink{2011}})} \end{minipage}
} \\
\hbox{mochi{-}new} &
\hbox{\mbox{\hphantom{\Scribtexttt{xx}}}} &
\hbox{11} &
\hbox{\mbox{\hphantom{\Scribtexttt{xx}}}} &
\hbox{0} &
\hbox{\mbox{\hphantom{\Scribtexttt{xx}}}} &
\multicolumn{7}{l}{\begin{minipage}[t]{0.5384615384615384\linewidth}
\Autobibref{\hyperref[t:x28autobib_x22Naoki_Kobayashix2c_Ryosuke_Satox2c_and_Hiroshi_UnnoPredicate_Abstraction_and_CEGAR_for_Higherx2dOrder_Model_CheckingIn_Procx2e_ACM_Conference_on_Programming_Language_Design_and_Implementationx2c_ppx2e_222x2dx2d2332011x22x29]{\AutobibLink{Kobayashi et al\Sendabbrev{.}}}~(\hyperref[t:x28autobib_x22Naoki_Kobayashix2c_Ryosuke_Satox2c_and_Hiroshi_UnnoPredicate_Abstraction_and_CEGAR_for_Higherx2dOrder_Model_CheckingIn_Procx2e_ACM_Conference_on_Programming_Language_Design_and_Implementationx2c_ppx2e_222x2dx2d2332011x22x29]{\AutobibLink{2011}})} \end{minipage}
} \\
\hbox{octy} &
\hbox{\mbox{\hphantom{\Scribtexttt{xx}}}} &
\hbox{13} &
\hbox{\mbox{\hphantom{\Scribtexttt{xx}}}} &
\hbox{0} &
\hbox{\mbox{\hphantom{\Scribtexttt{xx}}}} &
\multicolumn{7}{l}{\begin{minipage}[t]{0.5384615384615384\linewidth}
\Autobibref{\hyperref[t:x28autobib_x22Sam_Tobinx2dHochstadt_and_Matthias_FelleisenLogical_Types_for_Untyped_LanguagesIn_Procx2e_ACM_International_Conference_on_Functional_Programmingx2c_ppx2e_117x2dx2d1282010x22x29]{\AutobibLink{Tobin{-}Hochstadt and Felleisen}}~(\hyperref[t:x28autobib_x22Sam_Tobinx2dHochstadt_and_Matthias_FelleisenLogical_Types_for_Untyped_LanguagesIn_Procx2e_ACM_International_Conference_on_Functional_Programmingx2c_ppx2e_117x2dx2d1282010x22x29]{\AutobibLink{2010}})} \end{minipage}
} \\
\hbox{others} &
\hbox{\mbox{\hphantom{\Scribtexttt{xx}}}} &
\hbox{26} &
\hbox{\mbox{\hphantom{\Scribtexttt{xx}}}} &
\hbox{4} &
\hbox{\mbox{\hphantom{\Scribtexttt{xx}}}} &
\multicolumn{7}{l}{\begin{minipage}[t]{0.5384615384615384\linewidth}
\Autobibref{\hyperref[t:x28autobib_x22Phxfac_Nguyx1ec5nx2c_Sam_Tobinx2dHochstadtx2c_and_David_Van_HornHigher_order_symbolic_execution_for_contract_verification_and_refutationx2eJournal_of_Functional_Programmingx2827x29x2c_ppx2e_e3x3a1x2dx2de3x3a542017x22x29]{\AutobibLink{Nguy\~{\^{e}}n et al\Sendabbrev{.}}}~(\hyperref[t:x28autobib_x22Phxfac_Nguyx1ec5nx2c_Sam_Tobinx2dHochstadtx2c_and_David_Van_HornHigher_order_symbolic_execution_for_contract_verification_and_refutationx2eJournal_of_Functional_Programmingx2827x29x2c_ppx2e_e3x3a1x2dx2de3x3a542017x22x29]{\AutobibLink{2017}})}, \Autobibref{\hyperref[t:x28autobib_x22Phxfac_Nguyx1ec5nx2c_Sam_Tobinx2dHochstadtx2c_and_David_Van_HornRelatively_complete_counterexamples_for_higherx2dorder_programsx2eIn_Procx2e_ACM_Conference_on_Programming_Language_Design_and_Implementationx2c_ppx2e_446x2dx2d4562015x22x29]{\AutobibLink{Nguy\~{\^{e}}n et al\Sendabbrev{.}}}~(\hyperref[t:x28autobib_x22Phxfac_Nguyx1ec5nx2c_Sam_Tobinx2dHochstadtx2c_and_David_Van_HornRelatively_complete_counterexamples_for_higherx2dorder_programsx2eIn_Procx2e_ACM_Conference_on_Programming_Language_Design_and_Implementationx2c_ppx2e_446x2dx2d4562015x22x29]{\AutobibLink{2015}})}, \Autobibref{\hyperref[t:x28autobib_x22Sam_Tobinx2dHochstadt_and_David_Van_HornHigherx2dOrder_Symbolic_Execution_via_ContractsIn_Procx2e_ACM_Conference_on_Objectx2dOriented_Programmingx2c_Systemsx2c_Languages_and_Applicationsx2c_ppx2e_537x2dx2d5542012x22x29]{\AutobibLink{Tobin{-}Hochstadt and Horn}}~(\hyperref[t:x28autobib_x22Sam_Tobinx2dHochstadt_and_David_Van_HornHigherx2dOrder_Symbolic_Execution_via_ContractsIn_Procx2e_ACM_Conference_on_Objectx2dOriented_Programmingx2c_Systemsx2c_Languages_and_Applicationsx2c_ppx2e_537x2dx2d5542012x22x29]{\AutobibLink{2012}})} \end{minipage}
} \\
\hbox{softy} &
\hbox{\mbox{\hphantom{\Scribtexttt{xx}}}} &
\hbox{12} &
\hbox{\mbox{\hphantom{\Scribtexttt{xx}}}} &
\hbox{1} &
\hbox{\mbox{\hphantom{\Scribtexttt{xx}}}} &
\multicolumn{7}{l}{\begin{minipage}[t]{0.5384615384615384\linewidth}
\Autobibref{\hyperref[t:x28autobib_x22Robert_Cartwright_and_Mike_FaganSoft_TypingIn_Procx2e_ACM_Conference_on_Programming_Language_Design_and_Implementationx2c_ppx2e_278x2dx2d2921991x22x29]{\AutobibLink{Cartwright and Fagan}}~(\hyperref[t:x28autobib_x22Robert_Cartwright_and_Mike_FaganSoft_TypingIn_Procx2e_ACM_Conference_on_Programming_Language_Design_and_Implementationx2c_ppx2e_278x2dx2d2921991x22x29]{\AutobibLink{1991}})}, \Autobibref{\hyperref[t:x28autobib_x22Andrew_Kx2e_Wright_and_Robert_CartwrightA_Practical_Soft_Type_System_for_SchemeACM_Transactions_on_Programming_Languages_and_Systems_19x281x29x2c_ppx2e_87x2dx2d1521997x22x29]{\AutobibLink{Wright and Cartwright}}~(\hyperref[t:x28autobib_x22Andrew_Kx2e_Wright_and_Robert_CartwrightA_Practical_Soft_Type_System_for_SchemeACM_Transactions_on_Programming_Languages_and_Systems_19x281x29x2c_ppx2e_87x2dx2d1521997x22x29]{\AutobibLink{1997}})} \end{minipage}
} \\
\hbox{terauchi} &
\hbox{\mbox{\hphantom{\Scribtexttt{xx}}}} &
\hbox{7} &
\hbox{\mbox{\hphantom{\Scribtexttt{xx}}}} &
\hbox{0} &
\hbox{\mbox{\hphantom{\Scribtexttt{xx}}}} &
\multicolumn{7}{l}{\begin{minipage}[t]{0.5384615384615384\linewidth}
\Autobibref{\hyperref[t:x28autobib_x22Tachio_TerauchiDependent_Types_from_CounterexamplesIn_Procx2e_ACM_Symposium_on_Principles_of_Programming_Languagesx2c_ppx2e_119x2dx2d1302010x22x29]{\AutobibLink{Terauchi}}~(\hyperref[t:x28autobib_x22Tachio_TerauchiDependent_Types_from_CounterexamplesIn_Procx2e_ACM_Symposium_on_Principles_of_Programming_Languagesx2c_ppx2e_119x2dx2d1302010x22x29]{\AutobibLink{2010}})} \end{minipage}
} \\
\hbox{cuter} &
\hbox{\mbox{\hphantom{\Scribtexttt{xx}}}} &
\hbox{20} &
\hbox{\mbox{\hphantom{\Scribtexttt{xx}}}} &
\hbox{0} &
\hbox{\mbox{\hphantom{\Scribtexttt{xx}}}} &
\multicolumn{7}{l}{\begin{minipage}[t]{0.5384615384615384\linewidth}
\Autobibref{\hyperref[t:x28autobib_x22Aggelos_Giantsiosx2c_Nikolaos_Papaspyroux2c_and_Konstantinos_SagonasConcolic_Testing_for_Functional_LanguagesScience_of_Computer_Programmingx2c_ppx2e_109x2dx2d1342017x22x29]{\AutobibLink{Giantsios et al\Sendabbrev{.}}}~(\hyperref[t:x28autobib_x22Aggelos_Giantsiosx2c_Nikolaos_Papaspyroux2c_and_Konstantinos_SagonasConcolic_Testing_for_Functional_LanguagesScience_of_Computer_Programmingx2c_ppx2e_109x2dx2d1342017x22x29]{\AutobibLink{2017}})} \end{minipage}
} \\
\hbox{c{-}hop} &
\hbox{\mbox{\hphantom{\Scribtexttt{xx}}}} &
\hbox{7} &
\hbox{\mbox{\hphantom{\Scribtexttt{xx}}}} &
\hbox{0} &
\hbox{\mbox{\hphantom{\Scribtexttt{xx}}}} &
\multicolumn{7}{l}{\begin{minipage}[t]{0.5384615384615384\linewidth}
Interesting examples we discovered \end{minipage}
} \\
\hbox{\textbf{total}} &
\hbox{\mbox{\hphantom{\Scribtexttt{xx}}}} &
\hbox{\textbf{122}} &
\hbox{\mbox{\hphantom{\Scribtexttt{xx}}}} &
\hbox{\textbf{9}} &
\hbox{\mbox{\hphantom{\Scribtexttt{xx}}}} &
\multicolumn{7}{l}{\hbox{}}\end{bigtabular}\end{SCentered}\end{FigureInside}\end{Centerfigure}

\Centertext{\Legend{\FigureTarget{\label{t:x28counter_x28x22figurex22_x22figx3abenchmarkx2dresultsx22x29x29}\textsf{Fig.}~\textsf{11}. }{t:x28counter_x28x22figurex22_x22figx3abenchmarkx2dresultsx22x29x29}\textsf{Benchmark Results}}}\end{Figure}

Out of 122 benchmarks, our prototype fails to discover
bugs in 9 of the programs. These programs can be collected into groups that
correspond to specific limitations of our prototype and that explain why
our tool fails to test them successfully. First, our
search strategy is naive, so one benchmarks  time
out after an hour. Second, our prototype does not handle
Racket{'}s \RktSym{struct} declaration so for
five benchmarks the prototype fails to find a bug. Third, our prototype does not
generate pairs of inputs when the contract is \RktSym{any/c}
which results in one more
failure. Finally,
 two benchmarks use fairly complex
syntactic features of Racket that our prototype cannot
accommodate.

\sectionNewpage

\Ssection{Related Work}{Related Work}\label{t:x28part_x22secx3arelatedx22x29}

\paragraph{Concolic Testing.}
CutEr\Autobibref{~(\hyperref[t:x28autobib_x22Aggelos_Giantsiosx2c_Nikolaos_Papaspyroux2c_and_Konstantinos_SagonasConcolic_Testing_for_Functional_LanguagesIn_Procx2e_ACM_International_Conference_on_Principles_and_Practice_of_Declarative_Programmingx2c_ppx2e_137x2dx2d1482015x22x29]{\AutobibLink{Giantsios et al\Sendabbrev{.}}} \hyperref[t:x28autobib_x22Aggelos_Giantsiosx2c_Nikolaos_Papaspyroux2c_and_Konstantinos_SagonasConcolic_Testing_for_Functional_LanguagesIn_Procx2e_ACM_International_Conference_on_Principles_and_Practice_of_Declarative_Programmingx2c_ppx2e_137x2dx2d1482015x22x29]{\AutobibLink{2015}}, \hyperref[t:x28autobib_x22Aggelos_Giantsiosx2c_Nikolaos_Papaspyroux2c_and_Konstantinos_SagonasConcolic_Testing_for_Functional_LanguagesScience_of_Computer_Programmingx2c_ppx2e_109x2dx2d1342017x22x29]{\AutobibLink{2017}})} is a concolic testing tool
for Erlang\Autobibref{~(\hyperref[t:x28autobib_x22Joe_Armstrongx2c_Robert_Virdingx2c_Claes_Wikstrxf6mx2c_and_Mike_WilliamsProgramming_Erlangx3a_Software_for_a_Concurrent_WorldPrentice_Hall2007x22x29]{\AutobibLink{Armstrong et al\Sendabbrev{.}}} \hyperref[t:x28autobib_x22Joe_Armstrongx2c_Robert_Virdingx2c_Claes_Wikstrxf6mx2c_and_Mike_WilliamsProgramming_Erlangx3a_Software_for_a_Concurrent_WorldPrentice_Hall2007x22x29]{\AutobibLink{2007}})}. It supports the generation
of functions, but it is not complete in our sense. More specifically
it does not generate inputs that contain calls in their
bodies.\NoteBox{\NoteContent{Personal communication with Kostis Sagonas.}}
\Autobibref{\hyperref[t:x28autobib_x22Adrixe1n_Palacios_and_Germxe1n_VidalConcolic_Execution_in_Functional_Programming_by_Program_InstrumentationIn_Procx2e_International_Symposium_on_Logicx2dBased_Program_Synthesis_and_TRansformationx2c_ppx2e_277x2dx2d2922015x22x29]{\AutobibLink{Palacios and Vidal}}~(\hyperref[t:x28autobib_x22Adrixe1n_Palacios_and_Germxe1n_VidalConcolic_Execution_in_Functional_Programming_by_Program_InstrumentationIn_Procx2e_International_Symposium_on_Logicx2dBased_Program_Synthesis_and_TRansformationx2c_ppx2e_277x2dx2d2922015x22x29]{\AutobibLink{2015}})} proposes an
instrumentation approach for concolic testers of
functional languages but do not address the generation
of higher{-}order inputs.

\Autobibref{\hyperref[t:x28autobib_x22Lian_Lix2c_Yi_Lux2c_and_Jingling_XueDynamic_Symbolic_Execution_for_PolymorphismIn_Procx2e_International_Conference_on_Compiler_Constructionx2c_ppx2e_120x2dx2d1302017x22x29]{\AutobibLink{Li et al\Sendabbrev{.}}}~(\hyperref[t:x28autobib_x22Lian_Lix2c_Yi_Lux2c_and_Jingling_XueDynamic_Symbolic_Execution_for_PolymorphismIn_Procx2e_International_Conference_on_Compiler_Constructionx2c_ppx2e_120x2dx2d1302017x22x29]{\AutobibLink{2017}})} extends the design of path constraints with
symbolic subtype expressions in order to handle control{-}flow structure
involving polymorphism in object{-}oriented languages.
However, this design only generates inputs using existing
classes and does not synthesize new class definitions.

Path explosion remains a central challenge
for concolic testing techniques\Autobibref{~(\hyperref[t:x28autobib_x22Roberto_Baldonix2c_Emilio_Coppax2c_Daniele_Cono_Dx27Eliax2c_Camil_Demetrescux2c_and_Irene_FinocchiA_Survey_of_Symbolic_Execution_TechniquesACM_Computx2e_Survx2e_51x283x292018x22x29]{\AutobibLink{Baldoni et al\Sendabbrev{.}}} \hyperref[t:x28autobib_x22Roberto_Baldonix2c_Emilio_Coppax2c_Daniele_Cono_Dx27Eliax2c_Camil_Demetrescux2c_and_Irene_FinocchiA_Survey_of_Symbolic_Execution_TechniquesACM_Computx2e_Survx2e_51x283x292018x22x29]{\AutobibLink{2018}}; \hyperref[t:x28autobib_x22Cristian_Cadar_and_Koushik_SenSymbolic_Execution_for_Software_Testingx3a_Three_Decades_LaterCommunications_of_the_ACMx2c_ppx2e_82x2dx2d902013x22x29]{\AutobibLink{Cadar and Sen}} \hyperref[t:x28autobib_x22Cristian_Cadar_and_Koushik_SenSymbolic_Execution_for_Software_Testingx3a_Three_Decades_LaterCommunications_of_the_ACMx2c_ppx2e_82x2dx2d902013x22x29]{\AutobibLink{2013}})},
and it is a challenge
that partially motivates our work.
\Autobibref{\hyperref[t:x28autobib_x22Patrice_GodefroidCompositional_Dynamic_Test_GenerationIn_Procx2e_ACM_Symposium_on_Principles_of_Programming_Languagesx2c_ppx2e_47x2dx2d542007x22x29]{\AutobibLink{Godefroid}}~(\hyperref[t:x28autobib_x22Patrice_GodefroidCompositional_Dynamic_Test_GenerationIn_Procx2e_ACM_Symposium_on_Principles_of_Programming_Languagesx2c_ppx2e_47x2dx2d542007x22x29]{\AutobibLink{2007}})} approaches the problem by computing function summaries
on{-}the{-}fly to tame the combinatorial explosion of the search space of
control{-}flow paths.
Similarly, \Autobibref{\hyperref[t:x28autobib_x22Saswat_Anandx2c_Patrice_Godefroidx2c_and_Nikolai_TillmannDemandx2ddriven_Compositional_Symbolic_ExecutionIn_Procx2e_International_Conference_on_Tools_and_Algorithms_for_the_Construction_and_Analysis_of_Systemsx2c_ppx2e_367x2dx2d38312008x22x29]{\AutobibLink{Anand et al\Sendabbrev{.}}}~(\hyperref[t:x28autobib_x22Saswat_Anandx2c_Patrice_Godefroidx2c_and_Nikolai_TillmannDemandx2ddriven_Compositional_Symbolic_ExecutionIn_Procx2e_International_Conference_on_Tools_and_Algorithms_for_the_Construction_and_Analysis_of_Systemsx2c_ppx2e_367x2dx2d38312008x22x29]{\AutobibLink{2008}})} performs
symbolic execution compositionally using function summaries.
In both cases, due to the first{-}order nature of the programming
languages they are working with,
the summaries are first{-}order and
do not include possible higher{-}order interactions
between functions.

FOCAL\Autobibref{~(\hyperref[t:x28autobib_x22Yunho_Kimx2c_Shin_Hongx2c_and_Moonzo_KimTargetx2dDriven_Compositional_Concolic_Testing_with_Function_Summary_Refinement_for_Effective_Bug_DetectionIn_Procx2e_International_Symposium_on_on_the_Foundations_of_Software_Engineeringx2c_ppx2e_16x2dx2d262019x22x29]{\AutobibLink{Kim et al\Sendabbrev{.}}} \hyperref[t:x28autobib_x22Yunho_Kimx2c_Shin_Hongx2c_and_Moonzo_KimTargetx2dDriven_Compositional_Concolic_Testing_with_Function_Summary_Refinement_for_Effective_Bug_DetectionIn_Procx2e_International_Symposium_on_on_the_Foundations_of_Software_Engineeringx2c_ppx2e_16x2dx2d262019x22x29]{\AutobibLink{2019}})} addresses the path
explosion problem by breaking programs down into
small units to reduce the search space.
FOCAL tests units individually and tries to construct
a system{-}level test by refining path constraints using
function summaries.

\paragraph{Symbolic Execution.}

\Autobibref{\hyperref[t:x28autobib_x22Phxfac_Nguyx1ec5nx2c_Sam_Tobinx2dHochstadtx2c_and_David_Van_HornHigher_order_symbolic_execution_for_contract_verification_and_refutationx2eJournal_of_Functional_Programmingx2827x29x2c_ppx2e_e3x3a1x2dx2de3x3a542017x22x29]{\AutobibLink{Nguy\~{\^{e}}n et al\Sendabbrev{.}}}~(\hyperref[t:x28autobib_x22Phxfac_Nguyx1ec5nx2c_Sam_Tobinx2dHochstadtx2c_and_David_Van_HornHigher_order_symbolic_execution_for_contract_verification_and_refutationx2eJournal_of_Functional_Programmingx2827x29x2c_ppx2e_e3x3a1x2dx2de3x3a542017x22x29]{\AutobibLink{2017}})} and \Autobibref{\hyperref[t:x28autobib_x22Sam_Tobinx2dHochstadt_and_David_Van_HornHigherx2dOrder_Symbolic_Execution_via_ContractsIn_Procx2e_ACM_Conference_on_Objectx2dOriented_Programmingx2c_Systemsx2c_Languages_and_Applicationsx2c_ppx2e_537x2dx2d5542012x22x29]{\AutobibLink{Tobin{-}Hochstadt and Horn}}~(\hyperref[t:x28autobib_x22Sam_Tobinx2dHochstadt_and_David_Van_HornHigherx2dOrder_Symbolic_Execution_via_ContractsIn_Procx2e_ACM_Conference_on_Objectx2dOriented_Programmingx2c_Systemsx2c_Languages_and_Applicationsx2c_ppx2e_537x2dx2d5542012x22x29]{\AutobibLink{2012}})}
introduce higher{-}order symbolic execution and propose
the idea of canonical forms for generating counterexamples.
Their work on symbolic execution inspired our approach to
higher{-}order concolic testing.
Comparing to their work, our concolic tester
incorporates the notion of path constraints to support
incremental and systematic search over the control{-}flow
graph of higher{-}order programs.

\paragraph{Random Testing.}

QuickCheck\Autobibref{~(\hyperref[t:x28autobib_x22Koen_Claessen_and_John_HughesQuickCheckx3a_A_Lightweight_Tool_for_Random_Testing_of_Haskell_ProgramsIn_Procx2e_ACM_International_Conference_on_Functional_Programmingx2c_ppx2e_268x2dx2d2792000x22x29]{\AutobibLink{Claessen and Hughes}} \hyperref[t:x28autobib_x22Koen_Claessen_and_John_HughesQuickCheckx3a_A_Lightweight_Tool_for_Random_Testing_of_Haskell_ProgramsIn_Procx2e_ACM_International_Conference_on_Functional_Programmingx2c_ppx2e_268x2dx2d2792000x22x29]{\AutobibLink{2000}})} supports random testing
of higher{-}order functions by using user{-}provided maps
from the input type to integers and from integers to the
output type. \Autobibref{\hyperref[t:x28autobib_x22Pieter_Koopman_and_Rinus_PlasmeijerAutomatic_Testing_of_Higher_Order_FunctionsIn_Procx2e_Asian_Symposium_on_Programming_Languages_and_Systemsx2c_ppx2e_148x2dx2d1642006x22x29]{\AutobibLink{Koopman and Plasmeijer}}~(\hyperref[t:x28autobib_x22Pieter_Koopman_and_Rinus_PlasmeijerAutomatic_Testing_of_Higher_Order_FunctionsIn_Procx2e_Asian_Symposium_on_Programming_Languages_and_Systemsx2c_ppx2e_148x2dx2d1642006x22x29]{\AutobibLink{2006}})} improves upon QuickCheck
by using a predefined datatype representing the
syntax of higher{-}order functions to generate inputs with
richer behavior.

LambdaTester\Autobibref{~(\hyperref[t:x28autobib_x22Marija_Selakovicx2c_Michael_Pradelx2c_Rezwana_Karimx2c_and_Frank_TipTest_Generation_for_Higherx2dorder_Functions_in_Dynamic_LanguagesProceedings_of_the_ACM_on_Programming_Languages_x28OOPSLAx29_2x2c_ppx2e_161x3a1x2dx2d161x3a272018x22x29]{\AutobibLink{Selakovic et al\Sendabbrev{.}}} \hyperref[t:x28autobib_x22Marija_Selakovicx2c_Michael_Pradelx2c_Rezwana_Karimx2c_and_Frank_TipTest_Generation_for_Higherx2dorder_Functions_in_Dynamic_LanguagesProceedings_of_the_ACM_on_Programming_Languages_x28OOPSLAx29_2x2c_ppx2e_161x3a1x2dx2d161x3a272018x22x29]{\AutobibLink{2018}})} focuses on testing and
generating higher{-}order functions that mutate
an object state in order to affect control{-}flow paths
that depend on the object{'}s state.

\Autobibref{\hyperref[t:x28autobib_x22Casey_Kleinx2c_Matthew_Flattx2c_and_Robert_Bruce_FindlerRandom_Testing_for_Higherx2dorderx2c_Stateful_ProgramsIn_Procx2e_ACM_Conference_on_Objectx2dOriented_Programmingx2c_Systemsx2c_Languages_and_Applicationsx2c_ppx2e_555x2dx2d5662010x22x29]{\AutobibLink{Klein et al\Sendabbrev{.}}}~(\hyperref[t:x28autobib_x22Casey_Kleinx2c_Matthew_Flattx2c_and_Robert_Bruce_FindlerRandom_Testing_for_Higherx2dorderx2c_Stateful_ProgramsIn_Procx2e_ACM_Conference_on_Objectx2dOriented_Programmingx2c_Systemsx2c_Languages_and_Applicationsx2c_ppx2e_555x2dx2d5662010x22x29]{\AutobibLink{2010}})} uses the idea of calling
higher{-}order inputs in order to discover bugs in them,
combining the output of a function with the input of another.
Their techniques target a stateful setting and
are designed to work with opaque types.

\paragraph{Program Synthesis.}

The study of program synthesis for functional languages
also faces the challenge of generating higher{-}order programs.
Myth\Autobibref{~(\hyperref[t:x28autobib_x22Peterx2dMichael_Osera_and_Steve_ZdancewicTypex2dandx2dExamplex2dDirected_Program_Synthesis2015x22x29]{\AutobibLink{Osera and Zdancewic}} \hyperref[t:x28autobib_x22Peterx2dMichael_Osera_and_Steve_ZdancewicTypex2dandx2dExamplex2dDirected_Program_Synthesis2015x22x29]{\AutobibLink{2015}})} synthesizes higher{-}order functions
that generalize a given set of
input{-}output examples over inductive datatypes. Concolic testing
for higher{-}order functions can be viewed as synthesizing higher{-}order inputs
from an evolving set of examples dynamically
collected from the user programs, while having
the goal of exposing all control{-}flow paths.

\sectionNewpage

\Ssectionstarx{References}{References}\label{t:x28part_x22docx2dbibliographyx22x29}

\begin{AutoBibliography}\begin{SingleColumn}\label{t:x28autobib_x22Saswat_Anandx2c_Patrice_Godefroidx2c_and_Nikolai_TillmannDemandx2ddriven_Compositional_Symbolic_ExecutionIn_Procx2e_International_Conference_on_Tools_and_Algorithms_for_the_Construction_and_Analysis_of_Systemsx2c_ppx2e_367x2dx2d38312008x22x29}\Autobibentry{Saswat Anand, Patrice Godefroid, and Nikolai Tillmann. Demand{-}driven Compositional Symbolic Execution. In \textit{Proc. International Conference on Tools and Algorithms for the Construction and Analysis of Systems}, pp. 367{--}3831, 2008.}

\label{t:x28autobib_x22Saswat_Anandx2c_Mayur_Naikx2c_Mary_Jean_Harroldx2c_and_Hongseok_YangAutomated_Concolic_Testing_of_Smartphone_AppsIn_Procx2e_International_Symposium_on_on_the_Foundations_of_Software_Engineeringx2c_ppx2e_59x3a1x2dx2d59x3a112012x22x29}\Autobibentry{Saswat Anand, Mayur Naik, Mary Jean Harrold, and Hongseok Yang. Automated Concolic Testing of Smartphone Apps. In \textit{Proc. International Symposium on on the Foundations of Software Engineering}, pp. 59:1{--}59:11, 2012.}

\label{t:x28autobib_x22Joe_Armstrongx2c_Robert_Virdingx2c_Claes_Wikstrxf6mx2c_and_Mike_WilliamsProgramming_Erlangx3a_Software_for_a_Concurrent_WorldPrentice_Hall2007x22x29}\Autobibentry{Joe Armstrong, Robert Virding, Claes Wikstr\"{o}m, and Mike Williams. \textit{Programming Erlang: Software for a Concurrent World}. Prentice Hall, 2007.}

\label{t:x28autobib_x22Roberto_Baldonix2c_Emilio_Coppax2c_Daniele_Cono_Dx27Eliax2c_Camil_Demetrescux2c_and_Irene_FinocchiA_Survey_of_Symbolic_Execution_TechniquesACM_Computx2e_Survx2e_51x283x292018x22x29}\Autobibentry{Roberto Baldoni, Emilio Coppa, Daniele Cono D{\textquotesingle}Elia, Camil Demetrescu, and Irene Finocchi. A Survey of Symbolic Execution Techniques. \textit{ACM Comput. Surv.} 51(3), 2018.}

\label{t:x28autobib_x22Jacob_Burnim_and_Koushik_SenHeuristics_for_Scalable_Dynamic_Test_GenerationIn_Procx2e_ACMx2fIEEE_International_Conference_on_Automated_Software_Engineeringx2c_ppx2e_443x2dx2d4462008x22x29}\Autobibentry{Jacob Burnim and Koushik Sen. Heuristics for Scalable Dynamic Test Generation. In \textit{Proc. ACM/IEEE International Conference on Automated Software Engineering}, pp. 443{--}446, 2008.}

\label{t:x28autobib_x22Cristian_Cadarx2c_Daniel_Dunbarx2c_and_Dawson_EnglerKLEEx3a_Unassisted_and_Automatic_Generation_of_Highx2dcoverage_Tests_for_Complex_Systems_ProgramsIn_Procx2e_USENIX_Symposium_on_Operating_Systems_Design_and_Implementationx2c_ppx2e_209x2dx2d2242008x22x29}\Autobibentry{Cristian Cadar, Daniel Dunbar, and Dawson Engler. KLEE: Unassisted and Automatic Generation of High{-}coverage Tests for Complex Systems Programs. In \textit{Proc. USENIX Symposium on Operating Systems Design and Implementation}, pp. 209{--}224, 2008.}

\label{t:x28autobib_x22Cristian_Cadar_and_Dawson_EnglerExecution_Generated_Test_Casesx3a_How_to_Make_Systems_Code_Crash_ItselfIn_Procx2e_International_SPINConference_on_Model_Cheching_Softwarex2c_ppx2e_2x2dx2d232005x22x29}\Autobibentry{Cristian Cadar and Dawson Engler. Execution Generated Test Cases: How to Make Systems Code Crash Itself. In \textit{Proc. International SPINConference on Model Cheching Software}, pp. 2{--}23, 2005.}

\label{t:x28autobib_x22Cristian_Cadar_and_Koushik_SenSymbolic_Execution_for_Software_Testingx3a_Three_Decades_LaterCommunications_of_the_ACMx2c_ppx2e_82x2dx2d902013x22x29}\Autobibentry{Cristian Cadar and Koushik Sen. Symbolic Execution for Software Testing: Three Decades Later. \textit{Communications of the ACM}, pp. 82{--}90, 2013.}

\label{t:x28autobib_x22Robert_Cartwright_and_Mike_FaganSoft_TypingIn_Procx2e_ACM_Conference_on_Programming_Language_Design_and_Implementationx2c_ppx2e_278x2dx2d2921991x22x29}\Autobibentry{Robert Cartwright and Mike Fagan. Soft Typing. In \textit{Proc. ACM Conference on Programming Language Design and Implementation}, pp. 278{--}292, 1991.}

\label{t:x28autobib_x22Koen_Claessen_and_John_HughesQuickCheckx3a_A_Lightweight_Tool_for_Random_Testing_of_Haskell_ProgramsIn_Procx2e_ACM_International_Conference_on_Functional_Programmingx2c_ppx2e_268x2dx2d2792000x22x29}\Autobibentry{Koen Claessen and John Hughes. QuickCheck: A Lightweight Tool for Random Testing of Haskell Programs. In \textit{Proc. ACM International Conference on Functional Programming}, pp. 268{--}279, 2000.}

\label{t:x28autobib_x22Marko_Dimjax161evix107x2c_Dimitra_Giannakopouloux2c_Falk_Howarx2c_Falk_Howarx2c_Falk_Howarx2c_and_Falk_HowarThe_Dartx2c_the_Psycox2c_and_the_Doopx3a_Concolic_Execution_in_JavaACM_SIGSOFT_Software_Engineering_Notes_40x281x29x2c_ppx2e_1x2dx2d52015x22x29}\Autobibentry{Marko Dimja\v{s}evi\'{c}, Dimitra Giannakopoulou, Falk Howar, Falk Howar, Falk Howar, and Falk Howar. The Dart, the Psyco, and the Doop: Concolic Execution in Java. \textit{ACM SIGSOFT Software Engineering Notes} 40(1), pp. 1{--}5, 2015.}

\label{t:x28autobib_x22Michael_Emmix2c_Rupak_Majumdarx2c_and_Koushik_SenDynamic_Test_Input_Generation_for_Database_ApplicationsIn_Procx2e_International_Symposium_on_Software_Testing_and_Analysisx2c_ppx2e_151x2dx2d1622007x22x29}\Autobibentry{Michael Emmi, Rupak Majumdar, and Koushik Sen. Dynamic Test Input Generation for Database Applications. In \textit{Proc. International Symposium on Software Testing and Analysis}, pp. 151{--}162, 2007.}

\label{t:x28autobib_x22Azadeh_Farzanx2c_Andreas_Holzerx2c_Niloofar_Razavix2c_and_Helmut_VeithCon2Colic_TestingIn_Procx2e_International_Symposium_on_on_the_Foundations_of_Software_Engineeringx2c_ppx2e_37x2dx2d472013x22x29}\Autobibentry{Azadeh Farzan, Andreas Holzer, Niloofar Razavi, and Helmut Veith. Con2Colic Testing. In \textit{Proc. International Symposium on on the Foundations of Software Engineering}, pp. 37{--}47, 2013.}

\label{t:x28autobib_x22Matthias_Felleisenx2c_Robert_Bx2e_Findlerx2c_Matthew_Flattx2c_Shriram_Krishnamurthix2c_Eli_Barzilayx2c_Jay_McCarthyx2c_and_Sam_Tobinx2dHochstadtA_Programmable_Programming_LanguageCommunications_of_the_ACM_61x2c_ppx2e_3x3a62x2dx2d3x3a712018x22x29}\Autobibentry{Matthias Felleisen, Robert B. Findler, Matthew Flatt, Shriram Krishnamurthi, Eli Barzilay, Jay McCarthy, and Sam Tobin{-}Hochstadt. A Programmable Programming Language. \textit{Communications of the ACM} 61, pp. 3:62{--}3:71, 2018.}

\label{t:x28autobib_x22Aggelos_Giantsiosx2c_Nikolaos_Papaspyroux2c_and_Konstantinos_SagonasConcolic_Testing_for_Functional_LanguagesIn_Procx2e_ACM_International_Conference_on_Principles_and_Practice_of_Declarative_Programmingx2c_ppx2e_137x2dx2d1482015x22x29}\Autobibentry{Aggelos Giantsios, Nikolaos Papaspyrou, and Konstantinos Sagonas. Concolic Testing for Functional Languages. In \textit{Proc. ACM International Conference on Principles and Practice of Declarative Programming}, pp. 137{--}148, 2015.}

\label{t:x28autobib_x22Aggelos_Giantsiosx2c_Nikolaos_Papaspyroux2c_and_Konstantinos_SagonasConcolic_Testing_for_Functional_LanguagesScience_of_Computer_Programmingx2c_ppx2e_109x2dx2d1342017x22x29}\Autobibentry{Aggelos Giantsios, Nikolaos Papaspyrou, and Konstantinos Sagonas. Concolic Testing for Functional Languages. \textit{Science of Computer Programming}, pp. 109{--}134, 2017.}

\label{t:x28autobib_x22Patrice_GodefroidCompositional_Dynamic_Test_GenerationIn_Procx2e_ACM_Symposium_on_Principles_of_Programming_Languagesx2c_ppx2e_47x2dx2d542007x22x29}\Autobibentry{Patrice Godefroid. Compositional Dynamic Test Generation. In \textit{Proc. ACM Symposium on Principles of Programming Languages}, pp. 47{--}54, 2007.}

\label{t:x28autobib_x22Patrice_Godefroidx2c_Nils_Klarlundx2c_and_Koushik_SenDARTx3a_Directed_Automated_Random_TestingIn_Procx2e_ACM_Conference_on_Programming_Language_Design_and_Implementationx2c_ppx2e_213x2dx2d2232005x22x29}\Autobibentry{Patrice Godefroid, Nils Klarlund, and Koushik Sen. DART: Directed Automated Random Testing. In \textit{Proc. ACM Conference on Programming Language Design and Implementation}, pp. 213{--}223, 2005.}

\label{t:x28autobib_x22Patrice_Godefroidx2c_Michael_Yx2e_Levinx2c_and_David_MolnarAutomated_Whitebox_Fuzz_TestingIn_Procx2e_Network_and_Distributed_System_Security_Symposium2008x22x29}\Autobibentry{Patrice Godefroid, Michael Y. Levin, and David Molnar. Automated Whitebox Fuzz Testing. In \textit{Proc. Network and Distributed System Security Symposium}, 2008.}

\label{t:x28autobib_x22Patrice_Godefroidx2c_Michael_Yx2e_Levinx2c_and_David_MolnarSAGEx3a_Whitebox_Fuzzing_for_Security_TestingACM_Queue_10x281x29x2c_ppx2e_20x3a20x2dx2d20x3a272012x22x29}\Autobibentry{Patrice Godefroid, Michael Y. Levin, and David Molnar. SAGE: Whitebox Fuzzing for Security Testing. \textit{ACM Queue} 10(1), pp. 20:20{--}20:27, 2012.}

\label{t:x28autobib_x22Li_Guodongx2c_Esben_Andreasenx2c_and_Indradeep_GhoshSymJSx3a_Automatic_Symbolic_Testing_of_JavaScript_Web_ApplicationsIn_Procx2e_International_Symposium_on_on_the_Foundations_of_Software_Engineeringx2c_ppx2e_449x2dx2d4592014x22x29}\Autobibentry{Li Guodong, Esben Andreasen, and Indradeep Ghosh. SymJS: Automatic Symbolic Testing of JavaScript Web Applications. In \textit{Proc. International Symposium on on the Foundations of Software Engineering}, pp. 449{--}459, 2014.}

\label{t:x28autobib_x22Yunho_Kimx2c_Shin_Hongx2c_and_Moonzo_KimTargetx2dDriven_Compositional_Concolic_Testing_with_Function_Summary_Refinement_for_Effective_Bug_DetectionIn_Procx2e_International_Symposium_on_on_the_Foundations_of_Software_Engineeringx2c_ppx2e_16x2dx2d262019x22x29}\Autobibentry{Yunho Kim, Shin Hong, and Moonzo Kim. Target{-}Driven Compositional Concolic Testing with Function Summary Refinement for Effective Bug Detection. In \textit{Proc. International Symposium on on the Foundations of Software Engineering}, pp. 16{--}26, 2019.}

\label{t:x28autobib_x22Yunho_Kim_and_Moonzoo_KimSCOREx3a_A_Scalable_Concolic_Testing_Tool_for_Reliable_Embedded_SoftwareIn_Procx2e_International_Symposium_on_on_the_Foundations_of_Software_Engineeringx2c_ppx2e_420x2dx2d4232011x22x29}\Autobibentry{Yunho Kim and Moonzoo Kim. SCORE: A Scalable Concolic Testing Tool for Reliable Embedded Software. In \textit{Proc. International Symposium on on the Foundations of Software Engineering}, pp. 420{--}423, 2011.}

\label{t:x28autobib_x22Casey_Kleinx2c_Matthew_Flattx2c_and_Robert_Bruce_FindlerRandom_Testing_for_Higherx2dorderx2c_Stateful_ProgramsIn_Procx2e_ACM_Conference_on_Objectx2dOriented_Programmingx2c_Systemsx2c_Languages_and_Applicationsx2c_ppx2e_555x2dx2d5662010x22x29}\Autobibentry{Casey Klein, Matthew Flatt, and Robert Bruce Findler. Random Testing for Higher{-}order, Stateful Programs. In \textit{Proc. ACM Conference on Object{-}Oriented Programming, Systems, Languages and Applications}, pp. 555{--}566, 2010.}

\label{t:x28autobib_x22Naoki_Kobayashix2c_Ryosuke_Satox2c_and_Hiroshi_UnnoPredicate_Abstraction_and_CEGAR_for_Higherx2dOrder_Model_CheckingIn_Procx2e_ACM_Conference_on_Programming_Language_Design_and_Implementationx2c_ppx2e_222x2dx2d2332011x22x29}\Autobibentry{Naoki Kobayashi, Ryosuke Sato, and Hiroshi Unno. Predicate Abstraction and CEGAR for Higher{-}Order Model Checking. In \textit{Proc. ACM Conference on Programming Language Design and Implementation}, pp. 222{--}233, 2011.}

\label{t:x28autobib_x22Pieter_Koopman_and_Rinus_PlasmeijerAutomatic_Testing_of_Higher_Order_FunctionsIn_Procx2e_Asian_Symposium_on_Programming_Languages_and_Systemsx2c_ppx2e_148x2dx2d1642006x22x29}\Autobibentry{Pieter Koopman and Rinus Plasmeijer. Automatic Testing of Higher Order Functions. In \textit{Proc. Asian Symposium on Programming Languages and Systems}, pp. 148{--}164, 2006.}

\label{t:x28autobib_x22Guodong_Lix2c_Indradeep_Ghoshx2c_and_Sreeranga_Px2e_RajanKLOVERx3a_A_Symbolic_Execution_and_Automatic_Test_Generation_Tool_for_Cx2bx2b_ProgramsIn_Procx2e_International_Conference_on_Computer_Aided_Verificationx2c_ppx2e_609x2dx2d6152011x22x29}\Autobibentry{Guodong Li, Indradeep Ghosh, and Sreeranga P. Rajan. KLOVER: A Symbolic Execution and Automatic Test Generation Tool for C++ Programs. In \textit{Proc. International Conference on Computer Aided Verification}, pp. 609{--}615, 2011.}

\label{t:x28autobib_x22Guodong_Lix2c_Peng_Lix2c_Geof_Sawayax2c_Ganesh_Gopalakrishnanx2c_Indradeep_Ghoshx2c_and_Sreeranga_Px2e_RajanGKLEEx3a_Concolic_Verification_and_Test_Generation_for_GPUsIn_Procx2e_Symposium_on_Principles_and_Practice_of_Parallel_Programmingx2c_ppx2e_215x2dx2d2242012x22x29}\Autobibentry{Guodong Li, Peng Li, Geof Sawaya, Ganesh Gopalakrishnan, Indradeep Ghosh, and Sreeranga P. Rajan. GKLEE: Concolic Verification and Test Generation for GPUs. In \textit{Proc. Symposium on Principles and Practice of Parallel Programming}, pp. 215{--}224, 2012.}

\label{t:x28autobib_x22Lian_Lix2c_Yi_Lux2c_and_Jingling_XueDynamic_Symbolic_Execution_for_PolymorphismIn_Procx2e_International_Conference_on_Compiler_Constructionx2c_ppx2e_120x2dx2d1302017x22x29}\Autobibentry{Lian Li, Yi Lu, and Jingling Xue. Dynamic Symbolic Execution for Polymorphism. In \textit{Proc. International Conference on Compiler Construction}, pp. 120{--}130, 2017.}

\label{t:x28autobib_x22Phxfac_Nguyx1ec5nx2c_Sam_Tobinx2dHochstadtx2c_and_David_Van_HornSoft_Contract_VerificationIn_Procx2e_ACM_International_Conference_on_Functional_Programmingx2c_ppx2e_139x2dx2d1522014x22x29}\Autobibentry{Ph\'{u}c Nguy\~{\^{e}}n, Sam Tobin{-}Hochstadt, and David Van Horn. Soft Contract Verification. In \textit{Proc. ACM International Conference on Functional Programming}, pp. 139{--}152, 2014.}

\label{t:x28autobib_x22Phxfac_Nguyx1ec5nx2c_Sam_Tobinx2dHochstadtx2c_and_David_Van_HornRelatively_complete_counterexamples_for_higherx2dorder_programsx2eIn_Procx2e_ACM_Conference_on_Programming_Language_Design_and_Implementationx2c_ppx2e_446x2dx2d4562015x22x29}\Autobibentry{Ph\'{u}c Nguy\~{\^{e}}n, Sam Tobin{-}Hochstadt, and David Van Horn. Relatively complete counterexamples for higher{-}order programs. In \textit{Proc. ACM Conference on Programming Language Design and Implementation}, pp. 446{--}456, 2015.}

\label{t:x28autobib_x22Phxfac_Nguyx1ec5nx2c_Sam_Tobinx2dHochstadtx2c_and_David_Van_HornHigher_order_symbolic_execution_for_contract_verification_and_refutationx2eJournal_of_Functional_Programmingx2827x29x2c_ppx2e_e3x3a1x2dx2de3x3a542017x22x29}\Autobibentry{Ph\'{u}c Nguy\~{\^{e}}n, Sam Tobin{-}Hochstadt, and David Van Horn. Higher order symbolic execution for contract verification and refutation. \textit{Journal of Functional Programming}(27), pp. e3:1{--}e3:54, 2017.}

\label{t:x28autobib_x22Peterx2dMichael_Osera_and_Steve_ZdancewicTypex2dandx2dExamplex2dDirected_Program_Synthesis2015x22x29}\Autobibentry{Peter{-}Michael Osera and Steve Zdancewic. Type{-}and{-}Example{-}Directed Program Synthesis. 2015.}

\label{t:x28autobib_x22Adrixe1n_Palacios_and_Germxe1n_VidalConcolic_Execution_in_Functional_Programming_by_Program_InstrumentationIn_Procx2e_International_Symposium_on_Logicx2dBased_Program_Synthesis_and_TRansformationx2c_ppx2e_277x2dx2d2922015x22x29}\Autobibentry{Adri\'{a}n Palacios and Germ\'{a}n Vidal. Concolic Execution in Functional Programming by Program Instrumentation. In \textit{Proc. International Symposium on Logic{-}Based Program Synthesis and TRansformation}, pp. 277{--}292, 2015.}

\label{t:x28autobib_x22Niloofar_Razavix2c_Franjo_Ivanx10dix107x2c_Vineet_Kahlonx2c_and_Aarti_GuptaConcurrent_Test_Generation_Using_Concolic_Multix2dtrace_AnalysisIn_Procx2e_Asian_Symposium_on_Programming_Languages_and_Systemsx2c_ppx2e_239x2dx2d2552012x22x29}\Autobibentry{Niloofar Razavi, Franjo Ivan\v{c}i\'{c}, Vineet Kahlon, and Aarti Gupta. Concurrent Test Generation Using Concolic Multi{-}trace Analysis. In \textit{Proc. Asian Symposium on Programming Languages and Systems}, pp. 239{--}255, 2012.}

\label{t:x28autobib_x22Marija_Selakovicx2c_Michael_Pradelx2c_Rezwana_Karimx2c_and_Frank_TipTest_Generation_for_Higherx2dorder_Functions_in_Dynamic_LanguagesProceedings_of_the_ACM_on_Programming_Languages_x28OOPSLAx29_2x2c_ppx2e_161x3a1x2dx2d161x3a272018x22x29}\Autobibentry{Marija Selakovic, Michael Pradel, Rezwana Karim, and Frank Tip. Test Generation for Higher{-}order Functions in Dynamic Languages. \textit{Proceedings of the ACM on Programming Languages (OOPSLA)} 2, pp. 161:1{--}161:27, 2018.}

\label{t:x28autobib_x22Koushik_Sen_and_Gul_AghaCUTE_and_jCUTEx3a_Concolic_Unit_Testing_and_Explicit_Path_Modelx2dchecking_ToolsIn_Procx2e_International_Conference_on_Computer_Aided_Verificationx2c_ppx2e_419x2dx2d4232006x22x29}\Autobibentry{Koushik Sen and Gul Agha. CUTE and jCUTE: Concolic Unit Testing and Explicit Path Model{-}checking Tools. In \textit{Proc. International Conference on Computer Aided Verification}, pp. 419{--}423, 2006.}

\label{t:x28autobib_x22Koushik_Senx2c_Swaroop_Kalasapurx2c_Brutch_Tasneemx2c_and_Simon_GibbsJalangix3a_A_Selective_Recordx2dreplay_and_Dynamic_Analysis_Framework_for_JavaScriptIn_Procx2e_International_Symposium_on_on_the_Foundations_of_Software_Engineeringx2c_ppx2e_488x2dx2d4982013x22x29}\Autobibentry{Koushik Sen, Swaroop Kalasapur, Brutch Tasneem, and Simon Gibbs. Jalangi: A Selective Record{-}replay and Dynamic Analysis Framework for JavaScript. In \textit{Proc. International Symposium on on the Foundations of Software Engineering}, pp. 488{--}498, 2013.}

\label{t:x28autobib_x22Koushik_Senx2c_Darko_Marinovx2c_and_Gul_AghaCUTEx3a_A_Concolic_Unit_Testing_Engine_for_CIn_Procx2e_International_Symposium_on_on_the_Foundations_of_Software_Engineeringx2c_ppx2e_263x2dx2d2722005x22x29}\Autobibentry{Koushik Sen, Darko Marinov, and Gul Agha. CUTE: A Concolic Unit Testing Engine for C. In \textit{Proc. International Symposium on on the Foundations of Software Engineering}, pp. 263{--}272, 2005.}

\label{t:x28autobib_x22Vincent_Stx2dAmourx2c_Daniel_Felteyx2c_Spencer_Px2e_Florencex2c_Shux2dHung_Youx2c_and_Robert_Bx2e_FindlerHerbarium_Racketensisx3a_A_Stroll_Through_the_WoodsProceedings_of_the_ACM_on_Programming_Languages_x28ICFPx29_1x2c_ppx2e_1x3a1x2dx2d1x3a152017x22x29}\Autobibentry{Vincent St{-}Amour, Daniel Feltey, Spencer P. Florence, Shu{-}Hung You, and Robert B. Findler. Herbarium Racketensis: A Stroll Through the Woods. \textit{Proceedings of the ACM on Programming Languages (ICFP)} 1, pp. 1:1{--}1:15, 2017.}

\label{t:x28autobib_x22Youcheng_Sunx2c_Min_Wux2c_Wenjie_Ruanx2c_Xiaowei_Huangx2c_Marta_Kwiatkowskax2c_and_Daniel_KroeningConcolic_Testing_for_Deep_Neural_NetworksIn_Procx2e_ACMx2fIEEE_International_Conference_on_Automated_Software_Engineeringx2c_ppx2e_109x2dx2d1192018x22x29}\Autobibentry{Youcheng Sun, Min Wu, Wenjie Ruan, Xiaowei Huang, Marta Kwiatkowska, and Daniel Kroening. Concolic Testing for Deep Neural Networks. In \textit{Proc. ACM/IEEE International Conference on Automated Software Engineering}, pp. 109{--}119, 2018.}

\label{t:x28autobib_x22Tachio_TerauchiDependent_Types_from_CounterexamplesIn_Procx2e_ACM_Symposium_on_Principles_of_Programming_Languagesx2c_ppx2e_119x2dx2d1302010x22x29}\Autobibentry{Tachio Terauchi. Dependent Types from Counterexamples. In \textit{Proc. ACM Symposium on Principles of Programming Languages}, pp. 119{--}130, 2010.}

\label{t:x28autobib_x22Nikolai_Tillmann_and_Jonathan_de_HalleuxPexx3a_White_Box_Test_Generation_for_x2eNETIn_Procx2e_International_Conference_on_Tests_and_Proofsx2c_ppx2e_134x2dx2d1532008x22x29}\Autobibentry{Nikolai Tillmann and Jonathan de Halleux. Pex: White Box Test Generation for .NET. In \textit{Proc. International Conference on Tests and Proofs}, pp. 134{--}153, 2008.}

\label{t:x28autobib_x22Sam_Tobinx2dHochstadt_and_Matthias_FelleisenLogical_Types_for_Untyped_LanguagesIn_Procx2e_ACM_International_Conference_on_Functional_Programmingx2c_ppx2e_117x2dx2d1282010x22x29}\Autobibentry{Sam Tobin{-}Hochstadt and Matthias Felleisen. Logical Types for Untyped Languages. In \textit{Proc. ACM International Conference on Functional Programming}, pp. 117{--}128, 2010.}

\label{t:x28autobib_x22Sam_Tobinx2dHochstadt_and_David_Van_HornHigherx2dOrder_Symbolic_Execution_via_ContractsIn_Procx2e_ACM_Conference_on_Objectx2dOriented_Programmingx2c_Systemsx2c_Languages_and_Applicationsx2c_ppx2e_537x2dx2d5542012x22x29}\Autobibentry{Sam Tobin{-}Hochstadt and David Van Horn. Higher{-}Order Symbolic Execution via Contracts. In \textit{Proc. ACM Conference on Object{-}Oriented Programming, Systems, Languages and Applications}, pp. 537{--}554, 2012.}

\label{t:x28autobib_x22Emina_Torlak_and_Rastislav_BodikGrowing_Solverx2daided_Languages_with_RosetteIn_Procx2e_International_Symposium_on_New_Ideasx2c_New_Paradigmsx2c_and_Reflections_on_Programming_and_Softwarex2c_Onwardx21x2c_ppx2e_135x2dx2d1522013x22x29}\Autobibentry{Emina Torlak and Rastislav Bodik. Growing Solver{-}aided Languages with Rosette. In \textit{Proc. International Symposium on New Ideas, New Paradigms, and Reflections on Programming and Software, Onward!}, pp. 135{--}152, 2013.}

\label{t:x28autobib_x22Emina_Torlak_and_Rastislav_BodikA_Lightweight_Symbolic_Virtual_Machine_for_Solverx2dAided_Host_LanguagesIn_Procx2e_ACM_Conference_on_Programming_Language_Design_and_Implementationx2c_ppx2e_530x2dx2d5412014x22x29}\Autobibentry{Emina Torlak and Rastislav Bodik. A Lightweight Symbolic Virtual Machine for Solver{-}Aided Host Languages. In \textit{Proc. ACM Conference on Programming Language Design and Implementation}, pp. 530{--}541, 2014.}

\label{t:x28autobib_x22Andrew_Kx2e_Wright_and_Robert_CartwrightA_Practical_Soft_Type_System_for_SchemeACM_Transactions_on_Programming_Languages_and_Systems_19x281x29x2c_ppx2e_87x2dx2d1521997x22x29}\Autobibentry{Andrew K. Wright and Robert Cartwright. A Practical Soft Type System for Scheme. \textit{ACM Transactions on Programming Languages and Systems} 19(1), pp. 87{--}152, 1997.}\end{SingleColumn}\end{AutoBibliography}

\postDoc
\end{document}